\newif\iflong
\newif\ifshort
\newif\ifarxiv

\arxivtrue

\iflong 
\else
	\ifarxiv
	\else
	\shorttrue
	\fi
\fi

\documentclass{article}
\usepackage{ijcai23}

\usepackage{times}
\usepackage{soul}
\usepackage{url}
\usepackage[hidelinks]{hyperref}
\usepackage[utf8]{inputenc}
\usepackage[small]{caption}
\usepackage{graphicx}
\usepackage{amsmath}
\usepackage{amsthm}
\usepackage{booktabs}
\usepackage{algorithm}
\usepackage{algorithmic}
\usepackage[switch]{lineno}
\usepackage{xcolor}
\usepackage{dsfont,amssymb,mathtools} %

\usepackage{tikz}
\usetikzlibrary{decorations,arrows,arrows.meta,petri,topaths,backgrounds,shapes,positioning,fit,calc,decorations.pathreplacing,patterns,intersections}
\usetikzlibrary{positioning}%
\usetikzlibrary{calc}		%
\usetikzlibrary{arrows}
\usetikzlibrary {arrows.meta}
\usetikzlibrary{shapes.geometric}
\usetikzlibrary{graphs}
\usetikzlibrary {graphs.standard}
\tikzstyle{nodeW} = [draw=green!60!black, circle, fill=green!60!black, minimum size=1ex, inner sep=1pt, text centered, align=center]
\tikzstyle{nodeU} = [draw=blue!65!black, circle, fill=blue!65!black, minimum size=1ex, inner sep=1pt, text centered, align=center]
\tikzstyle{nodeRecA} = [draw=red, rectangle, minimum size=0.8ex, inner sep=0.8pt,fill=red]
\tikzstyle{nodeRecB} = [draw=gray, rectangle, minimum size=0.8ex, inner sep=1.2pt,fill=gray]
\tikzstyle{nn} = [draw=blue, circle, inner sep=1.2pt,fill=blue]
\definecolor{myRed}{RGB}{204,0,0}
\definecolor{myGreen}{RGB}{217,240,217}
\definecolor{myBlue}{RGB}{217,217,240}

\usepackage[textsize=tiny,textwidth=1.5cm,linecolor=green!70!black, backgroundcolor=green!10, bordercolor=black,disable]{todonotes}
\newcommand{\todoE}[1]{\todo[linecolor=pink, backgroundcolor=red!10]{E: #1}}

\newcommand{\newE}[1]{#1}
\newcommand{\newH}[1]{#1}

\newtheorem{theorem}{Theorem}

\newtheorem{claim}{Claim}[theorem]

\newtheorem{obs}{Observation}

\newtheorem{corollary}{Corollary}

\usepackage[capitalize]{cleveref}
\crefname{table}{Table}{Tables}
\crefname{figure}{Figure}{Figures}
\crefname{theorem}{Theorem}{Theorems}
\crefname{corollary}{Corollary}{Corollaries}
\crefname{observation}{Observation}{Observations}
\crefname{obs}{Observation}{Observations}
\crefname{lemma}{Lemma}{Lemmas}
\crefname{example}{Example}{Examples}
\crefname{reduction}{Reduction}{Reductions}
\crefname{construction}{Construction}{Constructions}
\crefname{subsection}{Subsection}{Subsections}
\crefname{section}{Section}{Sections}
\crefname{claim}{Claim}{Claims}
\crefname{clm}{Claim}{Claims}
\crefname{algorithm}{Algorithm}{Algorithm}
\crefname{definition}{Definition}{Definitions}

\usepackage{xifthen}
\usepackage{paralist}

\newcommand{\mytitle}{Optimal Seat Arrangement: What Are the Hard and Easy Cases?}
\title{\mytitle}
\newcommand{\appendixtitle}{Supplementary Material: \mytitle}

\author{
    Esra Ceylan$^1$
   \and
   Jiehua Chen$^1$ \And
   Sanjukta Roy$^2$ \\
   \affiliations
   $^1$TU Wien, Austria\\
   $^2$ Pennsylvania State University\\
   \emails
   e.ceylan.96@hotmail.com,
   jiehua.chen@tuwien.ac.at,
   sanjukta@psu.edu
}

\newcommand{\nonisolated}{k}
\newcommand{\maxoutdeg}{\Delta^{\!+}}

\newcommand{\W}[1][1]{{\color{purple}\textsf{w#1}h}}
\newcommand{\Wh}[1][1]{{\normalfont\textsf{W[#1]}-hard}}
\newcommand{\fpt}{{\normalfont\textsf{FPT}}}
\newcommand{\xp}{{\normalfont\textsf{XP}}}
\newcommand{\sfpt}{{\color{teal}\textsf{fpt}}}
\newcommand{\fptf}{fixed-parameter tractable}

\newcommand{\poly}{{\color{green!60!black}\textsf{P}}}
\newcommand{\NP}{{\color{red!90!black}\textsf{nph}}}
\newcommand{\NPB}{\textsf{NP}}
\newcommand{\PP}{\textsf{P}}
\newcommand{\NPh}{{\normalfont\textsf{NP}-hard}}
\newcommand{\NPc}{\textsf{NP}-complete}
\newcommand{\Oh}{\mathcal{O}}

\newcommand{\tableBodl}{\small$\clubsuit$}

\newcommand{\strictremark}{\textsuperscript{\tiny$\diamondsuit$}}

\newcommand{\R}{\mathbb{R}}

\newcommand{\agents}{P}
\newcommand{\sat}[1][p]{\mathsf{sat}_{#1}}
\newcommand{\util}[2][p]{\mathsf{util}_{#1}(#2)}
\newcommand{\wel}[1][\sigma]{\mathsf{wel}(#1)}
\newcommand{\egal}[1][\sigma]{\mathsf{egal}(#1)}

\newcommand{\swap}[2]{\sigma_{[#1 \leftrightarrow #2]}}

\newcommand{\Neigh}[1]{N_{#1}}
\newcommand{\Nout}[1]{N^+_{#1}}
\newcommand{\Nin}[1]{N^-_{#1}}
\newcommand{\Noutpos}[1]{N^+_{#1^+}}

\newcommand{\degr}[1]{\delta_{#1}}

\newcommand{\prefgraph}{\mathcal{F}}
\newcommand{\posprefgraph}{\mathcal{F}^+}
\newcommand{\negprefgraph}{\mathcal{F}^-}
\newcommand{\pathgraph}{\mathcal{P}}

\newcommand{\graphminus}{-}

\newcommand{\myemph}[1]{{\color{green!40!black}#1}}

\newcommand{\taskprob}[3]{

  \noindent
  \begin{minipage}{0.9\linewidth}
    \centering
    \begin{tabular}{p{.15\linewidth}@{\;}p{.8\linewidth}}
      \multicolumn{2}{l}{\textsc{#1}}\\
      \textbf{Input:} & #2\\
      \textbf{Task:} & #3
    \end{tabular}
  \end{minipage}
    \smallskip
}

\newcommand{\decprob}[3]{

  \noindent
  \begin{minipage}{0.9\linewidth}
    \centering
    \begin{tabular}{p{0.18\linewidth}@{\;}p{.8\linewidth}}
      \multicolumn{2}{l}{\textsc{#1}}\\
      \textbf{Input:} & #2\\
      \textbf{Question:} & #3
    \end{tabular}
  \end{minipage}
    \smallskip
}

\usepackage{xspace}
\newcommand{\IS}{\textsc{Independent Set}\xspace}
\newcommand{\Clique}{\textsc{Clique}\xspace}
\newcommand{\HamPath}{\textsc{Hamiltonian Path}\xspace}
\newcommand{\HamCycle}{\textsc{Hamiltonian Cycle}\xspace}
\newcommand{\ExactCover}{\textsc{Exact Cover by 3-Sets}\xspace}
\newcommand{\PathPartition}{\textsc{$P_3$-Partition}\xspace}
\newcommand{\ESRthree}{\textsc{$3$-Exchange Stable Roommates}\xspace}
\newcommand{\DkS}{\textsc{Densest $h$-Subgraph}\xspace}

\newcommand{\MWAf}{\textsc{Max Welfare Arrangement}\xspace}

\newcommand{\MUAf}{\textsc{Maxmin Utility Arrangement}\xspace}

\newcommand{\EFAf}{\textsc{Envy-Free Arrangement}\xspace}
\newcommand{\ESAf}{\textsc{Exchange-Stable Arrangement}\xspace}

\newcommand{\efArr}{envy-free arrangement\xspace}
\newcommand{\esArr}{exchange-stable arrangement\xspace}

\newcommand{\MWA}{\textsc{MWA}\xspace}
\newcommand{\MUA}{\textsc{MUA}\xspace}
\newcommand{\EFA}{\textsc{EFA}\xspace}
\newcommand{\ESA}{\textsc{ESA}\xspace}

\newcommand{\clique}{\small clique}
\newcommand{\nor}{\small no res.}

\newcommand{\pcycle}{\small{path/cycle}}

\newcommand{\matching}{\small matching}
\newcommand{\stars}{\small stars}

\newcommand{\pathcyclestars}{\small path/cycle/stars}

\newcommand{\noneg}{\footnotesize nonneg.}
\newcommand{\symm}{\footnotesize symm.}

\newcommand{\niv}{u}
\newcommand{\iv}{v}
\newcommand{\colr}{\mathsf{r}}
\newcommand{\colb}{\mathsf{b}}
\newcommand{\incgraph}{G_{\textit{inc}}}
\newcommand{\compgraph}{G_{\textit{comp}}}
\newcommand{\smallgraph}{\mathcal{G}_{\textit{small}}}
\newcommand{\largegraph}{\mathcal{G}_{\textit{large}}}
\newcommand{\redcomp}{\agents_{\colr}}
\newcommand{\bluecomp}{\agents_{\colb}}

\hypersetup{%
  backref=true, 
  pagebackref=true, 
  hypertexnames=true,
  colorlinks=true,citecolor=green!35!black,linkcolor=red!60!black%
} 

\newcommand{\appendixsymb}{$\ast$}
\usepackage{etoolbox} %

\newcommand{\appendixproofwithstatement}[3]{%
	\iflong{
  \gappto{\appendixtext}
  {%
    \subsection{Proof of \cref{#1}}\label{proof:#1}%
    \noindent {\normalfont\emph{#2}}
    #3
  }}
	\else
	\ifarxiv{#3}
	\fi
	\fi 
}

\newcommand{\appendixsection}[1]{%
  \gappto{\appendixtext}{
    \section{Additional Material for Section~\ref{#1}}
    \label{appsec:#1}
  }
}

\newcommand{\appendixcontinue}[4]{%
	\iflong{
	  #2
	  \gappto{\appendixtext}{
	    \subsection{Continuation of \cref{#1}}\label{example:#1}
	    \emph{#3}
	    
	    {#4}
	  }
	}
	\else
	\ifarxiv{#4}
	\fi
	\fi 
}

\newcommand{\statementarxiv}[3]{%
	\ifarxiv{
	\begin{#1}\label{#2}
		#3
	\end{#1}
	}
	\else{
	\begin{#1}[\appendixsymb]\label{#2}
		#3
	\end{#1}
	}
	\fi 
}

\begin{document}
\maketitle

\begin{abstract}
  We study four \NPh\ optimal seat arrangement problems~\cite{bodlaender2020}, which each have as input a set of $n$ agents, where each agent has cardinal preferences over other agents, and an $n$-vertex undirected graph (called \emph{seat graph}).
  The task is to assign each agent to a distinct vertex in the seat graph such that either the \emph{sum of utilities} or the \emph{minimum utility} is maximized, or it is \emph{envy-free} or \emph{exchange-stable}. 
  Aiming at identifying hard and easy cases, we {extensively study the algorithmic complexity of the four problems} by looking into natural graph classes for the seat graph (e.g., paths, cycles, stars, or matchings), problem-specific parameters (e.g., the number of non-isolated vertices in the seat graph or the maximum number of agents towards whom an agent has non-zero preferences), and preference structures (e.g., non-negative or symmetric preferences).
  For strict preferences and seat graphs with disjoint edges and isolated vertices, we correct an error by Bodlaender et al.~\shortcite{bodlaender2020tech} and show that finding an envy-free arrangement remains \NPh\ in this case.
\end{abstract}

\section{Introduction}
To designate seating for self-interested agents--\emph{seat arrangement}--is a fundamental and ubiquitous task in various situations, \newH{including seating for offices and events~\cite{LewisCarrollSeatPlan2016,vangerven22}, one-sided matching~\cite{gale62,alcalde94}, graphical resource allocation with preferences between the agents~\cite{massand19}, project management and team sports~\cite{GABMC2016MultiTeam}, and hedonic games with additive preferences~\cite{bogomolnaia02,aziz13,Woeginger2013}.}
The available seats, either physical or not, can be modeled via an undirected graph, called \emph{seat graph}, where each vertex corresponds to a seat and two vertices are connected through an edge if the corresponding seats are adjacent.
\newH{Simple graphs can already model many real-world situations, such as paths for rows of seats in meetings, grids or graphs consisting of disjoint edges for offices, or cliques (complete subgraphs) for groups or teams in games and coalition formation.}
Since agents have preferences over each other, their \emph{utility} for a seat may depend on who sits next to them. As such, not every seat arrangement is desired.
From a social welfare perspective, one may aim for an arrangement that maximizes the total or the minimum utility of the agents.
\newH{The corresponding combinatorial problems are called \MWAf (\MWA) and \MUAf (\MUA), respectively.}
From a game-theoretical perspective, however, one may aim for an arrangement where no agent envies the seat of another agent or no two agents would rather want to exchange their seats.
\newH{The corresponding decision problems are called \EFAf (\EFA) and \ESAf (\ESA), respectively.}
We provide an example with four agents~$p_1,p_2,p_3,p_4$ in \cref{fig:ex1}.
\newH{The preferences of the agents are depicted as an arc-weighted directed graph, called \emph{preference graph}.} An arc from agent~$p$ to 
agent~$q$ with weight $w\neq 0$ means that $p$ has a 
preference of $w$ towards $q$. %
The missing arcs represent~$0$ preferences. 
The utility of each agent is \emph{additive}, i.e., it is the sum of preferences of all the agents seated next to him. 
In arrangement~$\sigma_1$ (see \cref{fig:ex1}),
the utilities of~$p_1, p_2, p_3$, and $p_4$ are $-1$, $3$, $0$, and $2$, respectively. 
Hence, the minimum utility is  $-1$ and
the sum of utilities, i.e., \emph{welfare}, of~$\sigma_1$ is $4$. 
Arrangement $\sigma_1$ maximizes the welfare, but it is not exchange-stable since $p_1$ and $p_3$ \emph{envy} each other's seat and form an \emph{exchange-blocking pair}, i.e., they can increase their utility by swapping their seats.
\newH{Consequently, it is not envy-free.}
Arrangement $\sigma_2$ %
only maximizes the minimum utility, and is envy-free and exchange-stable.

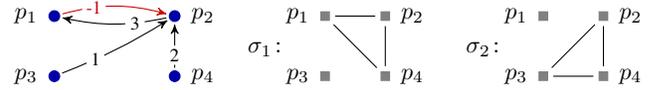
\begin{figure}[t]
  \begin{tikzpicture}[scale=1,every node/.style={scale=0.9}, >=stealth', shorten <= 1pt, shorten >= 1pt]
    \foreach  \x / \y / \n in
    {0/0/p1, 2/0/p2, 0/-1/p3, 2/-1/p4}
    {
      \node[nodeU] at (\x*0.8, \y*0.8) (\n) {};
    }
    \foreach \n / \p / \l in {1/left/1, 2/right/1, 3/left/1, 4/right/1} {
      \node[\p = \l pt of p\n] {$p_\n$}; 
    }
    \begin{scope}[every node/.style={fill=white,circle}]
      \begin{scriptsize}
		\foreach \s / \t/ \b /\v / \c in {1/2/15/-1/myRed, 2/1/12/3/black, 3/2/-5/1/black, 4/2/0/2/black} 
		{
	          \draw[->] (p\s) edge[bend left=\b, \c]  node[pos=0.3, fill=white, inner sep=0.5pt] {\v} (p\t);
		}
      \end{scriptsize}
    \end{scope}
    \begin{scope}[xshift=3.6cm,scale=1,every node/.style={scale=0.9}, >=stealth', shorten <= 2pt, shorten >= 2pt]
      \foreach  \x / \y / \n in
      {0/0/p1, 2/0/p2, 0/-1/p3, 2/-1/p4}
      {
        \node[nodeRecB] at (\x*0.4, \y*0.8) (\n) {};
      }
      \foreach \n / \p / \l in {1/left/1, 2/right/1, 3/left/1, 4/right/1} {
        \node[\p = \l pt of p\n] {$p_\n$}; 
      }
      \foreach \s / \t in {p1/p2, p2/p4, p4/p1} {
        \draw (\s) edge (\t);
      }
      \node at (-1.5*0.5, -0.55*0.8) {$\sigma_1\colon$};
    \end{scope}
    \begin{scope}[xshift=6.5cm,scale=1,every node/.style={scale=0.9}, >=stealth', shorten <= 2pt, shorten >= 2pt]
      \foreach  \x / \y / \n in
      {0/0/p1, 2/0/p2, 0/-1/p3, 2/-1/p4}
      {
        \node[nodeRecB] at (\x*0.4, \y*0.8) (\n) {};
      }
      \foreach \n / \p / \l in {1/left/1, 2/right/1, 3/left/1, 4/right/1} {
        \node[\p = \l pt of p\n] {$p_\n$}; 
      }
      \foreach \s / \t in {p2/p3, p3/p4, p4/p2} {
        \draw (\s) edge (\t);
      }
      \node at (-1.5*0.5, -0.55*0.8) {$\sigma_2\colon$};
    \end{scope}
  \end{tikzpicture}
  \caption{Left: Preference graph. The seat graph is a triangle with one isolated vertex. We use red arcs in the preference graph to denote negative preferences and black for positive ones. %
  Middle and right: Two possible seat arrangements.}\label{fig:ex1}
\end{figure}

\newH{In this work, we provide a refined complexity analysis of the four seat arrangement problems \MWA, \MUA, \EFA, and \ESA; the first three problems are known to be \NPh\ even for rather restricted cases, such as when the largest component of the seat graph has constant size~$\ell = 3$ and the preferences are symmetric and non-negative, while \ESA\ is \NPh\ even when $\ell=2$~\cite{bodlaender2020} and the preferences are strict~\cite{cechlarova05}, i.e., no agent has the same preferences towards two distinct agents. 
However, these hard\-ness results do not necessarily transfer to the case when $\ell > 3$ as all four problems become trivial if the seat graph is just a complete graph.
In other words, the complexity of seat arrangement may also depend on other aspects, such as the seat graph classes and other parameters. %
Hence, we systematically and extensively study the algorithmic complexity through a combination of three aspects:}

\noindent \textbf{Aspect 1: Seat graph classes.} We distinguish between \emph{path}\mbox{-,} \emph{cycle}-, \emph{clique}-, \emph{stars}-, and \emph{matching}-graphs which means that the graphs induced by the non-isolated vertices consist of a single path, a single cycle, a complete subgraph, disjoint stars, and disjoint edges, respectively. %
  Note that seat arrangement can also be used to model coalition formation or task management~\cite{shehory1998methods}, so the connectivity between the seats can be more sophisticated than just simple paths and cycles.
  For instance, in task assignment, when communication between agents must go through a central point, a stars-graph may be needed.
  More concretely, in~a sensor drone network one may divide drones into groups to save communication costs, and choose a center drone for each group through which other drones talk to the outside.
  In coalition formation, clique-graphs model the case where the task is to form a large coalition, together with singleton agents.

\noindent \textbf{Aspect 2: {Parameters}.} We specifically look at two structural parameters, namely
  the number~$\nonisolated$ of non-isolated vertices in the seat graph and
  the maximum number~$\maxoutdeg$ of agents known to an agent.
  Parameter~$\nonisolated$ has computational motivation since all four problems would be trivial if every seat is isolated. Hence, $\nonisolated$ measures the distance from triviality.
  Moreover, in many scenarios, such as task \mbox{assignment}, there may be a limited number of tasks that should be worked on by more than one agent, but also many single-agent tasks (i.e., isolated vertices). 
  Parameter~$\maxoutdeg$ is motivated by the observation that each agent typically only knows a few other agents, and thus the number of non-zero preferences of an agent is bounded~\cite{bachrach2013optimal,CsehIrvingManlove2019}.

\noindent \textbf{Aspect 3: Preference structures.} We primarily focus on two natural restrictions: \emph{non-negative} preferences which occur when no agent has enemies and \emph{symmetric} preferences where each pair of agents has exactly the same preferences over each other and which can model mutual acquaintances. %
\paragraph{Our contributions.}
We provide a comprehensive complexity picture (see \cref{tbl:overview-results}) and   %
  summarize our key contributions as follows:
  {      
    \begin{table*}[t!]
      \renewcommand{\aboverulesep}{0pt}
      \renewcommand{\belowrulesep}{0pt}
      \renewcommand{\arraystretch}{1.2}
      \setlength{\tabcolsep}{5pt}
      \small
      \centering
      \begin{tabular}{@{\,}l@{\;}
        c@{\,}l@{\;}c@{}l@{} c@{\,\;} 
        c@{\,}l@{\;}c@{}l@{} c@{\,\;}
        c@{\,}l@{\;}c@{}l@{} c@{\,\;\,}
        c@{\,}l@{\;}c@{}l@{} c@{\,}l@{}}
        \toprule
        & \multicolumn{4}{c}{\MWA} && \multicolumn{4}{c}{\MUA} && \multicolumn{4}{c}{\EFA} && \multicolumn{6}{c}{\ESA}\\
        \cline{2-5} \cline{7-10} \cline{12-15} \cline{17-22}
        & \multicolumn{2}{c}{\nor/} & \multicolumn{2}{c}{\pcycle/}  && \multicolumn{2}{c}{\nor/} & \multicolumn{2}{c}{\pcycle/} && \multicolumn{2}{c}{\nor/\clique/} & \multicolumn{2}{c}{\matching} & & \multicolumn{2}{c}{\nor/} & \multicolumn{2}{c}{\pcycle/} & \multicolumn{2}{c}{\matching} \\
        Param.& \multicolumn{2}{c}{\clique} & \multicolumn{2}{c}{\stars}&& \multicolumn{2}{c}{\clique}& \multicolumn{2}{c}{\stars} && \multicolumn{2}{c}{\pathcyclestars} & &  &&  \multicolumn{2}{c}{\clique}&\multicolumn{2}{c}{\stars} \\
     	   \midrule  \\[-2.8ex]
        $\nonisolated$ & \W\strictremark & [T\ref{thm:MWA_k_clique_bin-symm+strict}] %
        & ~{\sfpt} & [T\ref{thm:MWA_k_path-cycle+stars}] &&  \W\strictremark & [T\ref{thm:MUA_k_clique_symm+strict}] & {~\sfpt} & [T\ref{thm:MUA_k_path-cycle+stars}] & & \W\strictremark&[T\ref{thm:EFA_k_matching+clique+path-cycle_bin}] & \W\strictremark& [T\ref{thm:EFA_k_matching+clique+path-cycle_bin}] & & {\W} &[T\ref{thm:ESA_k_clique+matching}]& ~{?/?/\W}~ &[T\ref{thm:ESA_k_clique+matching}] & {\W}&[T\ref{thm:ESA_k_clique+matching}]\\
        \,\noneg & \W & [T\ref{thm:MWA_k_clique_bin-symm+strict}] & ~\sfpt & [T\ref{thm:MWA_k_path-cycle+stars}] && ~\sfpt/\poly& [P\ref{obs:MUA_nonneg-clique-k}] & ~\sfpt & [T\ref{thm:MUA_k_path-cycle+stars}] && \W&[T\ref{thm:EFA_k_matching+clique+path-cycle_bin}] & \W &  [T\ref{thm:EFA_k_matching+clique+path-cycle_bin}] & &?/\poly&[O\ref{obs:ESA_clique_nonneg}] & ? & -- & ? & --\\ 
        \,\symm & \W & [T\ref{thm:MWA_k_clique_bin-symm+strict}] & ~\sfpt & [T\ref{thm:MWA_k_path-cycle+stars}]  
        &&\W  &  [T\ref{thm:MUA_k_clique_symm+strict}]  & ~\sfpt & [T\ref{thm:MUA_k_path-cycle+stars}] &&                                                                                                                                        \,\W&[T\ref{thm:EFA_k_clique+stars+path-cycle_symm}] & \poly& [\tableBodl]& & {?}  & -- & ~\sfpt &[\tableBodl,T\ref{thm:MWA_k_path-cycle+stars}] & \poly & [\tableBodl]  \\\hline \\[-2.9ex]
        $\maxoutdeg$ & ~\NP & [T\ref{thm:MWA_delta_clique+path-cycle+stars_bin-symm}] & ~\NP & [T\ref{thm:MWA_delta_clique+path-cycle+stars_bin-symm}] && ~\NP & [T\ref{thm:MUA_delta_clique+star+path-cycle_symm}] & ~\NP & [T\ref{thm:MUA_delta_clique+star+path-cycle_symm}]   && ~\NP&[T\ref{thm:EFA_delta_matching+clique+path-cycle_bin}]  & ~\NP&[T\ref{thm:EFA_delta_matching+clique+path-cycle_bin}] && ~\NP&[T\ref{thm:ESA_delta_clique+path-cycle+matching}] & ~\NP&[T\ref{thm:ESA_delta_clique+path-cycle+matching}] & ~\NP&[T\ref{thm:ESA_delta_clique+path-cycle+matching}]\\
        \,\noneg &  ~\NP &[T\ref{thm:MWA_delta_clique+path-cycle+stars_bin-symm}] & ~\NP & [T\ref{thm:MWA_delta_clique+path-cycle+stars_bin-symm}] && ~\NP/\poly &[T\ref{thm:MUA_delta_clique+star+path-cycle_symm},P\ref{obs:MUA_nonneg-clique-k}] & ~\NP & [T\ref{thm:MUA_delta_clique+star+path-cycle_symm}]  &&~\NP&[T\ref{thm:EFA_delta_matching+clique+path-cycle_bin}] & ~\NP&[T\ref{thm:EFA_delta_matching+clique+path-cycle_bin}] && ?/\poly&[O\ref{obs:ESA_clique_nonneg}] & ?  & -- & ? & -- \\
        \,\symm &~\NP & [T\ref{thm:MWA_delta_clique+path-cycle+stars_bin-symm}] & ~\NP & [T\ref{thm:MWA_delta_clique+path-cycle+stars_bin-symm}]  &&   ~\NP & [T\ref{thm:MUA_delta_clique+star+path-cycle_symm}] & ~\NP & [T\ref{thm:MUA_delta_clique+star+path-cycle_symm}] && ~\NP&[T\ref{thm:EFA_delta_clique+path-cycle+stars_symm}] & \poly & [\tableBodl]  && ? & -- & ? & -- & \poly & [\tableBodl] \\                    \hline
        \\[-2.9ex]
        $\nonisolated + \maxoutdeg$ &  \W&[T\ref{thm:MWA_k+delta_clique_bin}] & ~\sfpt & [T\ref{thm:MWA_k_path-cycle+stars}]  && ~\sfpt & [T\ref{thm:MUA_k+delta}] &  ~\sfpt & [T\ref{thm:MUA_k_path-cycle+stars}] && \W&[T\ref{thm:EFA_k+delta_clique+stars},\ref{thm:EFA_k+delta_path-cycle}]& ~\sfpt &[T\ref{thm:EFA_k+delta_matching}] && ~\sfpt&[T\ref{thm:ESA_k+delta}] & ~\sfpt&[T\ref{thm:ESA_k+delta}] & ~\sfpt&[T\ref{thm:ESA_k+delta}]\\
        \,\noneg & \W &[T\ref{thm:MWA_k+delta_clique_bin}] & ~\sfpt & [T\ref{thm:MWA_k_path-cycle+stars}]  && ~\sfpt & [T\ref{thm:MUA_k+delta}] & ~\sfpt & [T\ref{thm:MUA_k_path-cycle+stars}] && ~\sfpt&[T\ref{thm:EFA_k+delta_nonneg}] & ~\sfpt&[T\ref{thm:EFA_k+delta_nonneg}] && ~\sfpt&[T\ref{thm:ESA_k+delta}] & ~\sfpt&[T\ref{thm:ESA_k+delta}] & ~\sfpt&[T\ref{thm:ESA_k+delta}] \\
        \,\symm & ~\sfpt &[T\ref{thm:MWA_k+delta_symm}] & ~\sfpt & [T\ref{thm:MWA_k_path-cycle+stars}]&& ~\sfpt & [T\ref{thm:MUA_k+delta}] & ~\sfpt & [T\ref{thm:MUA_k_path-cycle+stars}] && ~\sfpt & [T\ref{thm:EFA_k+delta_nonneg}] & \poly & [\tableBodl] && ~\sfpt&[T\ref{thm:ESA_k+delta}] & ~\sfpt&[T\ref{thm:ESA_k+delta}] & ~\sfpt&[T\ref{thm:ESA_k+delta}] \\ 
        \bottomrule
      \end{tabular}
      \caption{All \Wh\ (\W) problems are also in \xp. The \Wh{ness} 
        and \NPh{ness} (\NP) results for non-negative preferences always hold for binary preferences (except \MUA on a path-graph wrt.\ $\maxoutdeg$). ``\strictremark'' means that hardness holds even for strict preferences. We omit the case with matching-graphs for \MWA and \MUA since it is polynomial-time solvable~{\protect\cite{bodlaender2020}}~[\tableBodl]. \ifarxiv
        In \cref{thm:EFA_NP_path-cycle_bin+symm}, we show that \EFA remains \NPh\ even if the seat graph is a single path or cyle, while in \cref{cor:ESA_NP_path}, we show that \ESA\ remains \NPh\ even if the seat graph is a single path, and in \cref{thm:ESA_NP_cycle_nonneg} we show that \ESA is \NPh\ even if the seat graph is a single cycle and the preferences are non-negative.\fi} 
      \label{tbl:overview-results}
    \end{table*}
  }
  \begin{compactenum}[(1)]
  \item We obtain a number of fixed-parameter tractable (\fpt) algorithms for either $\nonisolated$ or $(\nonisolated,\maxoutdeg)$, i.e., the corresponding problems can be solved in time $f(\nonisolated)\cdot |I|^{O(1)}$ or $f(\nonisolated+\maxoutdeg)\cdot |I|^{O(1)}$, where $f$ is a computable function solely depending on the argument, and $|I|$ denotes the input size.
  The \fpt\ algorithms for $\nonisolated$ are for \MWA\ and \MUA\ under simple seat graphs and are based on color-coding coupled with book keeping of either the sum of utilities or the minimum utility.
  The \fpt\ algorithms for $(\nonisolated,\maxoutdeg)$ apply to all four problems, mostly without any restrictions on the seat graphs.
  They are based on random separation or kernelization. 

  \item We also obtain a number of \Wh{ness} results, most\-ly wrt.~$\nonisolated$; these exclude any \fpt\ algorithm for~$\nonisolated$ in~such cases.
  For \MUA\ and \MWA, these remain so for clique-graphs and symmetric preferences, while for \EFA\ and \ESA, they hold for almost all considered seat graph classes. 
  \newH{For \EFA\ and \ESA, the proofs are based on a novel all-or-nothing gadget (see \cref{fig:EFA_k_path-cycle_symm}) which may be of independent interest.}

  \item We strengthen existing NP-hardness results by showing that
  all four problems remain \NPh\ even for constant~$\maxoutdeg$ and for severe restrictions on the seat graphs and preference structures. 
  This is done by cleverly tweaking the preference graph with constant~$\maxoutdeg$.

\end{compactenum}
\newH{Summarizing, we show that for \MUA and \ESA, the combined parameter $(\nonisolated, \maxoutdeg)$ (aspect 1 alone) always gives rise to an \fpt\ algorithm, while for \MWA and \EFA, this is only the case for symmetric preferences.}
\newH{Additionally, we correct an error by Bodlaender et al.~\shortcite{bodlaender2020tech} and show that \EFA\ remains \NPh\ for matching-graphs and strict preferences (see \cref{thm:EFA_k_matching+clique+path-cycle_bin})}. %
\ifarxiv
\newE{Furthermore, we extend and improve existing knowledge about the complexity of \EFA and \ESA. 
Our results include that both \EFA and \ESA remain \NPh\ for restricted preferences even if the seat graph is a single path or cycle (see \cref{thm:EFA_NP_path-cycle_bin+symm,thm:ESA_NP_cycle_nonneg,cor:ESA_NP_path}).} %
\fi 

\paragraph{Paper outline.} %
In \cref{sec:prelim}, we introduce basic concepts and the four central problems.
In \cref{sec:mwa,sec:mua,sec:efa,sec:esa}, we discuss results for
\MWA, \MUA, \EFA, and \ESA, respectively.
In all four sections, we first consider parameter~$\nonisolated$, then $\maxoutdeg$, and finally the combination~$(\nonisolated, \maxoutdeg$).
\ifarxiv
\newE{In \cref{sec:complexity-EFA-ESA}, we give additional complexity results for \EFA and \ESA.}
\fi 
We conclude in \cref{sec:conclude}.
\ifshort
Due to space constraints, proofs for results or additional material marked with (\appendixsymb) are deferred to the technical report~\cite{optimal-arxiv}\todoE{add cite}.
\iflong
Due to space constraints, proofs for results or additional material marked with (\appendixsymb) are deferred to the appendix.
\fi 
\fi 

\paragraph{Related work.}
The solution concepts considered in the four problems are well studied in economics, social choice, and political sciences~\cite{caragiannis12,shapley12,aziz13,brandt2016}. 
Bodlaender et al.~\shortcite{bodlaender2020} initiated the study of the four optimal seat arrangement problems (OSA) and observed that 
\MWA\ and \MUA\ generalize the \NPh\ \textsc{Spanning Subgraph Isomorphism} problem, while \ESA\ generalizes the \NPh\ \textsc{Exchange-Stable Roommates} problem~\cite{cechlarova05}.
Very recently, Chen et al.~\shortcite{chen21} prove that \textsc{Exchange-Stable Roommates} remains \NPh\ even if each agent has positive preferences over at most three agents.
Derived from this, we show the same holds for \ESA\ under matching-graphs and with constant~$\maxoutdeg$.
OSA has been getting more attention recently. %
  Tomi{\'c} and Uro{\v{s}}evi{\'c}~\shortcite{tomic21} provide heuristic approaches for solving \MWA where the seat graph consists of equal-sized cliques.
\newH{
  Vangerven et al.~\shortcite{vangerven22} studied a related problem for parliament seating but the objective is different from ours. 
  }

  OSA generalizes multi-dimensional matchings ~\cite{CFH2019,BredHeeKnoNie2020multidimensional,ChenRoy2022esa}~and hedonic games with fixed-sized coalitions~\cite{bilo22} 
  where \newH{the seat graph consists of equal-sized cliques and cliques of fixed sizes, respectively.}
\newH{Bil{\`{o}} et al.~\shortcite{bilo22} consider paths to exchange stability and \MWA and strengthened the complexity result by Bodlaender et al.~\shortcite{bodlaender2020} by showing that \MWA\ remains highly inapproximable even if the seat graph consists of cliques of constant sizes.}
Massand and Simon~\shortcite{massand19} studied a generalization of OSA where the agents additionally have non-negative valuations over the seats such that 
the utility of an agent is the sum of his valuation of the assigned seat and his preferences over the agents in his neighborhood.
Their results imply that \EFA\ remains \NPh\ even if the preference graph is a cycle with binary preferences, and that for symmetric preferences, an \esArr can always be obtained from a given arrangement via a finite number of swaps.

OSA can also be conceptualized as hedonic games with overlapping coalitions similar to Schelling games which has been shown to be \NPh\ for simple models~\cite{agarwal2021schelling,kreisel2022equilibria}. 
However, in Schelling games the preferences are more restricted.
We refer to the work of Bodlaender et al.~\shortcite{bodlaender2020} for additional references.

\section{Preliminaries}\label{sec:prelim}
\appendixsection{sec:prelim}
Given an integer~$t$, let \myemph{$[t]$} denote the set~$\{1,2,\ldots,t\}$. We recall the following graph theoretic concepts.
Given an undirected graph $G$ and a vertex~$v\in V(G)$, we use \myemph{$\Neigh{G}(v)$} to denote $v$'s open neighborhood~$\{w\in V(G) \mid \{v,w\}\in E(G)\}$. 
Given a directed graph~$F$, and a vertex~$v\in V(F)$ we use~\myemph{$\Nin{F}(v)$} and~\myemph{$\Nout{F}(v)$} to denote the set of in- and out-neighbors of~$v$, respectively. 

A seat arrangement instance consists of a set~$\agents$ %
of $n$ agents, where each agent~$p \in \agents$ has cardinal preferences over the other agents, specified by the satisfaction function $\sat\colon \agents \setminus \{p\} \to \R$ for all~$p \in \agents$, and an undirected graph~$G$, called \myemph{seat graph}, where the number of vertices is the same as the number of agents, i.e., $|V(G)|=n$.
We derive a weighted directed graph $\prefgraph = (\agents, A, (\sat)_{p\in \agents})$ from $\sat$, called \myemph{preference graph}, where the vertex set is the set~$P$ of agents and~\myemph{$A$} denotes the set of arcs such that an arc from agent $p$ to $q$ means that $\sat(q) \neq 0$. 
\ifshort\else
With \myemph{$\posprefgraph$} (resp.\ \myemph{$\negprefgraph$}) we denote the preference graph restricted to only positive (resp.\ negative) weighted arcs.
\fi 
Note that, intuitively, negative (resp.\ positive) preference values model the degree of dislike (likeness) while zero values model indifference.
The goal of optimal seat arrangement is to find a bijection~$\sigma\colon \agents \to V(G)$, called \myemph{arrangement}, which is optimal or fair.
A \myemph{partial arrangement} is an injective function~$\sigma$ which assigns only a subset of the agents to the vertices of $G$, i.e., $\sigma^{-1}(V(G)) \subset \agents$.

We look into four different criteria, namely utilitarian and egalitarian welfares, envy-freeness, and exchange stability.
\ifarxiv
To this end, given 
\else
Given 
\fi 
an arrangement~$\sigma$, we define the \myemph{utility} of each agent~$p\in \agents$ as the sum of the satisfactions of~$p$ towards his neighbors in~$\sigma$, i.e., \mbox{\myemph{$\util{\sigma}$} $\coloneqq \sum_{v\in N_G(\sigma(p))}\sat(\sigma^{-1}(v))$}. 
By convention, the \myemph{utilitarian} (resp.\ \myemph{egalitarian}) welfare of~$\sigma$ is the sum (resp.\ minimum) of utilities of the agents towards~$\sigma$, denoted as \myemph{$\wel$} $\coloneqq \sum_{p}\util{\sigma}$ (resp.\ \mbox{\myemph{$\egal$} $\coloneqq \min_p(\util{\sigma})$}).
Additionally, for two agents $p, q \in \agents$ we define the \myemph{swap-arrangement}~$\swap{p}{q}$ as the arrangement where just $p$ and~$q$ switch their seats in $\sigma$, i.e.,  {$\swap{p}{q}(p)$} $\coloneqq \sigma(q)$, \mbox{$\swap{p}{q}(q) \coloneqq \sigma(p)$}, and all other agents~$x \in \agents\setminus \{p, q\}$ remain in their seats, i.e., ~$\swap{p}{q}(x) \coloneqq \sigma(x)$.
An arrangement~$\sigma$ is called \myemph{envy-free} (resp.\ \myemph{exchange-stable}) if no agent envies any other agent (resp.\ no two agents envy each other).
An agent~$p \in \agents$ \myemph{envies} another agent $q \in \agents\setminus\{p\}$ (in~$\sigma$) if $p$ finds the seat of~$q$ more attractive than his own, i.e., $\util{\sigma} < \util{\swap{p}{q}}$, and he is \myemph{envy-free} if he does not envy any other agents.
By definition, envy-freeness implies exchange stability.

Based on the different criteria, we define four computational/decision problems.
\taskprob{\MWA\ ({\normalfont\text{resp.\,}}\MUA)}
{An instance $(\agents, (\sat)_{p\in \agents}, G)$.}
{Find an arrangement $\sigma\colon \agents \to V(G)$ with maximum~$\wel$ (resp.\ $\egal$).}

\decprob{\EFA\ (\text{\normalfont{resp.}} \ESA)} %
{An instance $(\agents, (\sat)_{p\in \agents}, G)$.}
{Is there an arrangement $\sigma \colon \agents \to V(G)$ which is envy-free (resp.\ exchange-stable)?}

In the decision variant of \MWA (resp.\ \MUA), the input additonally has an integer~$L$, and the question is whether there is an arrangement~$\sigma$ with $\wel \geq L$ (resp.\ $\egal \geq L$).
It is straightforward that \EFA, \ESA, and the decision variants of \MWA and \MUA belong to \NPB.

Let $I=(\agents, (\sat)_{p\in \agents}, G)$ be an instance of our problems.
We say that the preferences of the agents are
\begin{compactitem}[--]
  \item \myemph{binary} if $\sat(q)\in \{0,1\}$ holds for each~$\{p, q\}\subseteq 
  \agents$,
  \item \myemph{non-negative} if  $\sat(q)\ge 0$ holds for each~$\{p, q\}\subseteq \agents$, %
  \item \myemph{positive} if $\sat(q) > 0$ holds for each~$\{p, q\}\subseteq \agents$, %
  \item \myemph{symmetric} if $\sat(q)$\,=\,$\sat[q](p)$ holds for each~$\{p, q\} \subseteq\agents$,
  \item \myemph{strict} if $\sat(q_1) \neq \sat(q_2)$ for each~$\{p, q_1,q_2\} \subseteq \agents$.
\end{compactitem}
We say that the seat graph is a
\begin{compactitem}[--]
  \item \myemph{clique-graph} if it consists of a complete subgraph (aka.\ clique) and isolated vertices,
  \item \myemph{stars-graph} if each connected component is a star, 
  \item \myemph{path-} (resp.\ \myemph{cycle-}) \myemph{graph} if it consists of a path (resp.\ \mbox{cycle}) and isolated vertices,
  \item \myemph{matching-graph} if it consists of disjoint edges and isolated vertices.
\end{compactitem}
We use \myemph{$\nonisolated$} (resp.\ \myemph{$\maxoutdeg$}) to denote the number of non-isolated vertices in the seat graph (resp.\ maximum out-degree of the vertices in the preference graph). 
Briefly put, $\nonisolated$ bounds the number of relevant seats that have neighbors, while~$\maxoutdeg$ bounds the maximum number of agents that an agent likes or dislikes.
We observe that all four problems are polynomial-time solvable for a constant number of non-isolated vertices by a simple brute-forcing algorithm.

\newcommand{\fptk}{%
  \MWA, \MUA, \EFA, and \ESA are in \xp\ wrt.\ $\nonisolated$.
}

\statementarxiv{obs}{obs:xp-k}{\fptk}
\appendixproofwithstatement{obs:xp-k}{\fptk}{
  \begin{proof}
    Let $I$ be an instance of \MWA, \MUA, \EFA resp.\ \ESA. 
    We obtain an \xp\ algorithm for each of these four problems by first 
    considering each size-$\nonisolated$ 
    subset $S_\nonisolated$ of $\agents$ and then each arrangement of 
    $S_\nonisolated$ to the $\nonisolated$ 
    non-isolated vertices.
    The remaining agents are assigned to isolated vertices. 
    For \EFA and \ESA we check whether this arrangement is envy-free resp.\ 
    exchange-stable. 
    If this is the case, we return yes; otherwise we continue with the next arrangement. 
    If we cannot find an \efArr resp.\ \esArr with this procedure, we return no. 
    For \MWA and \MUA we go through each arrangement and return one which maximizes the welfare resp.\ minimum utility. 
    
    There are $\Oh(n^\nonisolated)$ subsets of $\agents$ of size $\nonisolated$ and $\Oh(\nonisolated!)$ arrangements of $\nonisolated$ agents to $\nonisolated$ seats.
    Since computing the welfare and the minimum utility of an arrangement and checking if an arrangement is envy-free resp.\ exchange-stable can be done in polynomial time, we obtain that this can algorithm runs in time $\Oh(n^\nonisolated \cdot \nonisolated!)$. 
  \end{proof}
}

We assume basic knowledge of parameterized complexity such as fixed-parameter tractability~(\fpt), \Wh{ness}, and \xp, and refer to the following textbooks~\cite{Nie06,cygan15} for more details.

\looseness=-1
\section{\MWAf}\label{sec:mwa}
\appendixsection{sec:mwa}
In this section we study the (parameterized) complexity of finding an arrangement that maximizes the welfare.
Using color-coding coupled with a book keeping for welfare, we get fixed-parameter tractability (\fpt) for simple graphs:
\newcommand{\mwakpathcyclestarsfpt}{%
  \MWA is \fpt\ wrt.\ $\nonisolated$ if the seat graph is a path-, cycle-, or stars-graph. 
}
\statementarxiv{theorem}{thm:MWA_k_path-cycle+stars}{\mwakpathcyclestarsfpt}

\appendixproofwithstatement{thm:MWA_k_path-cycle+stars}{\mwakpathcyclestarsfpt}{ 
  \begin{proof}
    Let $I = (\agents, (\sat)_{p\in \agents}, G)$ be an \MWA instance.		
    We design a simple algorithm using color-coding as follows. 
    We color the agents $\agents$ uniformly at random with colors 
    $[\nonisolated]$. 
    Let $\sigma$ be a hypothetical optimal arrangement and $\chi\colon 
    \agents\to[\nonisolated]$ a coloring.
    We say $\chi$ is \myemph{good} (wrt.\ $\sigma$) if the $\nonisolated$ 
    non-isolated agents have pairwise distinct colors, i.e., $\chi(p) \neq 
    \chi(q)$ for each non-isolated $p,q \in \agents$. 
    Note that given a solution $\sigma$, the probability that the 
    $\nonisolated$ non-isolated agents are colored with pairwise distinct 
    colors is at least~$e^{-\nonisolated}$~\cite{cygan15}.		

    \paragraph*{Path-graphs.}
    We use dynamic programming (DP) to assign $\nonisolated$ agents on a 
    path which maximize the welfare. 
    We define a DP table where an entry $T[S,p]$ %
    returns the maximum welfare of a partial arrangement assigning agents on an $|S|$-path, where 
    \begin{inparaenum}[(i)]
      \item no two agents are colored with the same color,
      \item each color in $S$ is used once, and
      \item agent $p$ is assigned to the last vertex in the path. 
    \end{inparaenum}
    
    We start with filling our table $T$ for $|S| = 1$ and $p \in P$. Since there is only one assigned agent, the welfare is zero if the agent's color is the single color in $S$. Otherwise, there is no valid partial arrangement.
    \begin{align*}
      T[S,p] = \begin{cases}
        0 &\text{if } S = \{\chi(p)\}, \\
        -\infty & \text{otherwise.}
      \end{cases}
    \end{align*}
    
    For $|S| > 1$, we consider the maximum welfare of partial paths using all colors in $S \setminus \{\chi(p)\}$ with an agent~$q$ at the endpoint. 
    Then, we choose an $|S|$-path which has maximum welfare when adding agent~$p$ after~$q$. 
    Since each color in~$S$ is used exactly once, we can ensure that no agent is assigned to more than one seat on the path.
    Hence, the following recurrence holds:
    \begin{align*}
      &T[S,p] = \\
      &\max_{\substack{q \in \agents, \\\chi(q) \in 
      S\setminus\{\chi(p)\}}} ( T[S\setminus 
      \{\chi(p)\}, q] + \sat[p](q) + \sat[q](p)).
    \end{align*}
    
    Since $S \subseteq [\nonisolated]$, the entries of this table can be computed in time $2^\nonisolated\cdot n^{\Oh(1)}$ by applying the above 
    recurrence.
    The maximum welfare can be obtained by considering $\max_{p \in \agents} T[[\nonisolated], p]$. 
    
    \paragraph*{Cycle-graphs.}
    W extend the above DP table by adding a start-vertex $p_s$ of the path, i.e., $T'[S, p_s, p]$. 
    The base case and recurrence is the same as for a path-graph. 
    Hence, we first construct an $|S|$-path. 
    Then, we obtain the maximum welfare by closing the path to a cycle, i.e., we compute $\max_{p_s\neq p \in \agents} \left(T'[[\nonisolated], p_s, p] + \sat[p_s](p) + \sat[p](p_s)\right)$. 
    The time needed to compute this table's entries is $2^\nonisolated\cdot n^{\Oh(1)}$.
    
    \paragraph*{Stars-graphs.}
    For stars-graphs the maximum welfare can by computed in polynomial time. 
    Let the vertices of the seat graph~$G$ be denoted by $\{1,\ldots, k\}$.
    The $\nonisolated$ colors one-to-one correspond to the $\nonisolated$ non-isolated seats in the seat graph.
    Since agents in different stars do not interact with each other, we will go through each star individually and compute the maximum welfare of this star. 
    Let $k_c$ be the center of the corresponding star and $k_1, \ldots, \nonisolated_\ell$ its leaves. 
    For each possible center-agent we select the leaf-agents in such a way that $\sat(q) + \sat[q](p)$ is maximized. 
    Hence, we compute the welfare of a star by
    \begin{align*}
      \max_{\substack{p \in \agents,\\ \chi(p)=k_c}}\sum_{i = 1}^{\ell} \max_{\substack{q \in \agents\setminus\{p\},\\ \chi(q)=k_i}} (\sat(q) + \sat[q](p)).
    \end{align*}
    With the given coloring we can ensure that each agent is selected at most once. 
    Since the number of stars as well as the number of leaves of each star is $\Oh(\nonisolated)$, the computation time is $\Oh(k^2 \cdot n^2)$. 
    
    Summing up, given a coloring $\chi$ the algorithms run in $2^\nonisolated \cdot n^{\Oh(1)}$ resp.\ $\Oh(k^2 \cdot n^2)$ time.
    The probability of a good coloring is $e^{-k}$. 
    Hence, after repeating this algorithm $e^{k}$ times we obtain a solution with high probability.
    Finally, we can de-randomize the color-coding approach while maintaining fixed-parameter tractability~\cite{cygan15}.
  \end{proof}}

However, it is \Wh\ even for restricted preferences and when the seat graph is a clique-graph since the problem generalizes the \Clique problem which is \Wh\ wrt.\ the solution size~\cite{DF13}.
\newcommand{\mwakcliquebinsymmstrict}{%
  For clique-graphs, \MWA is \Wh\ wrt.~$\nonisolated$, even for strict, or binary and symmetric preferences.%
}
\statementarxiv{theorem}{thm:MWA_k_clique_bin-symm+strict}{\mwakcliquebinsymmstrict}
\appendixproofwithstatement{thm:MWA_k_clique_bin-symm+strict}{\mwakcliquebinsymmstrict}{
  \begin{proof}
    \textbf{Binary and symmetric preferences.}
    We provide a parameterized reduction from the \Wh\ \Clique\ problem, parameterized by the solution size~\cite{DF13}. 
    \ifarxiv
    \decprob{\Clique}
    {An undirected graph~$\hat{G}$ and a non-negative integer~$h$.}
    {Does $\hat{G}$ admit a size-$h$ clique, i.e., a complete subgraph?}
    \fi 

    Let $(\hat{G}, h)$ denote an instance of \Clique.
    For each vertex~$v_i \in V(\hat{G})$ we create one agent named $p_i$ and set $\sat[p_i](p_j)=\sat[p_j](p_i)=1$ if and only if $\{v_i,v_j\} \in E(\hat{G})$. 
    The non-mentioned preferences are set to zero. 
    The seat graph~$G$ consists of a clique of size~$\nonisolated \coloneqq h$ and $V(\hat{G})|-h$ isolated vertices.
    
    It is straightforward to verify that $\hat{G}$ admits a size-$h$ clique~if and only if there is an arrangement~$\sigma$ with $\wel \geq h(h-1)$.
    
    \textbf{Strict preferences.}                
    Now, we turn to strict preferences and modify the reduction above.
    The only differences are the preferences.
    For each neighbor~$v_j$ of $v_i$ in $\hat{G}$ we set $\sat[p_i](p_j)$ to a positive value from $[\hat{n}]$, where $\hat{n}\coloneqq |V(\hat{G})|$. 
    For each non-neighbor $v_j$ of $v_i$ we set $\sat[p_i](p_j)$ to a negative value from $\{-\hat{n}(\hat{n}+1), \ldots, -\hat{n}^2\}$. 
    In any case, the values can be assigned arbitrary, but each preference value can appear only at most once.
    Again, it is to verify that $\hat{G}$ has a size-$h$ clique if and only if there is an arrangement~$\sigma$ with positive welfare, i.e.,~$\wel > 0$.
  \end{proof}
}
\noindent Surprisingly, even if each agent only knows two others, \fpt\ algorithms (wrt.\ $\nonisolated$) for \MWA\ are unlikely:

\newcommand{\mwakcliquebinwhard}{%
  For clique graphs, \MWA remains \Wh\ wrt.\ $\nonisolated$, even for binary preferences with $\maxoutdeg = 2$. %
}
\begin{theorem}\label{thm:MWA_k+delta_clique_bin}
  \mwakcliquebinwhard
\end{theorem}
{
  \begin{proof}
    We prove this via a parameterized reduction from \Clique, parameterized by the solution 
\ifarxiv
	size~$h$~\cite{DF13}. 
\else
    size~$h$~\cite{DF13}. %
\fi  
    For each vertex $v_i \in V(\hat{G})$ (resp.\ edge $e_\ell \in E(\hat{G})$), we create one \myemph{vertex-agent}~$p_i$ (resp.\ \myemph{edge-agent}~$q_{\ell}$),
    and set $\sat[q_{\ell}](p_i)= 1$ if~$v_i \in e_\ell$. 
    The non-mentioned preferences are set to zero. 
    The seat graph~$G$ consists of a clique of size $\nonisolated\coloneqq h + \binom{h}{2}$ and $|V(\hat{G})|+|E(\hat{G})|-\nonisolated$ isolated vertices. 
    Clearly, the pref\-er\-enc\-es are binary. Since only edge-agents have positive pref\-er\-enc\-es (towards incident vertices), we infer $\maxoutdeg = 2$, as desired.
    
    We show that $\hat{G}$ admits a size-$h$ clique if and only if there is an arrangement $\sigma$ with $\wel \geq h(h-1)$.
    The forward direction is straightforward by assigning the corresponding vertex- and edge-agents contained in the $h$-clique to the non-isolated vertices in the clique-graph.
    
    For the backward direction, we first observe that at least~$\binom{h}{2}$ edge-agents are non-isolated; otherwise this would imply that $\wel < h(h-1)$. 
    If~$\binom{h}{2} + x$ edge-agents are assigned to non-isolated vertices with $x \ge 1$, then we observe:
    \begin{inparaenum}[(i)]
      \item Each non-isolated edge-agent $q_{\ell}$ has a positive utility. 
      Otherwise, we can exchange $q_{\ell}$ with a vertex-agent which is 
      incident to some other non-isolated edge-agent. 
      This would only increase the welfare.			
      \item There are~$h-x$ vertex-agents assigned to non-isolated vertices. 
      Hence, there can be at most $\binom{h-x}{2}$ edge-agents with both endpoints in the clique. 
      Because of $\binom{h}{2} + x - \binom{h-x}{2} \geq 1 + x \geq 2$, we can always find two edge-agents with only one endpoint in the clique.
      If we exchange one edge-agent with a missing vertex-agent of the other edges, then the welfare does not decrease. 
    \end{inparaenum}
    Therefore, if there is an arrangement $\sigma$ with $\wel \geq h(h-1)$, we can apply these exchange-arguments until there is no more such pair. 
    This implies that exactly~$\binom{h}{2}$ edge-agents (and exactly~$h$ vertex-agents) are non-isolated. 
    Moreover, each of these edge-agents has a utility of two, since otherwise we cannot reach the desired welfare. 
    As there are~$h$ non-isolated vertex-agents which are incident to~$\binom{h}{2}$ edge-agents, it follows that the corresponding~$h$ vertices in~$\hat{G}$ form a clique. 
  \end{proof}

  \noindent
  \newH{Except for matching-graphs, \MWA\ remains intractable for constant~$\maxoutdeg$ %
  since it can be rephrased as finding a densest subgraph in the preference graph which remains hard for constant degree; this is also independently noted by Bil{\'o} et al.~\shortcite{bilo22}: %
  }
  \newcommand{\mwaconstantoutdeghard}{%
    For binary and symmetric preferences, \MWA remains \NPh\ for constant $\maxoutdeg$ and for
    each considered seat graph class except the matching-graphs. 
  }

\statementarxiv{theorem}{thm:MWA_delta_clique+path-cycle+stars_bin-symm}{\mwaconstantoutdeghard}
  \appendixproofwithstatement{thm:MWA_delta_clique+path-cycle+stars_bin-symm}{\mwaconstantoutdeghard}{ 
    \begin{proof}
      \textbf{Path- and cycle-graphs.}
      We prove this via a polynomial-time reduction from \HamPath, which is \NPh\ even for cubic graphs~\cite{garey1979}. 		
      \decprob{\HamPath}
      {An undirected graph~$\hat{G}$.}
      {Does $\hat{G}$ admit a Hamiltonian path, i.e., a path that visits every vertex $v \in V(\hat{G})$ exactly once?}
      
      Given an instance $\hat{G}$ of \HamPath, we create for each vertex $v_i \in V(\hat{G})$ one agent named $p_i$. 
      We set $\sat[p_i](p_j)=\sat[p_j](p_i) = 1$ if and only if $\{v_i,v_j\} \in E(\hat{G})$. 
      Otherwise, the preference value is set to zero.
      The seat graph is a path (resp.\ cycle) with $|V(\hat{G})|$ vertices. 
      Clearly, the preferences are binary and symmetric. 
      Since we can assume that~$\hat{G}$ is a cubic graph, we obtain $\maxoutdeg = 3$ for the created \MWA instance. 
      
      It remains to show that $\hat{G}$ has a Hamiltonian path if and only if there is an arrangement~$\sigma$ with $\wel \geq 2(n-1)$.
      This follows from the fact, that for each edge $\{v_i, v_j\} \in E(\hat{G})$, if the seats of agents $p_i, p_j$ are adjacent, then this contributes $2$ to the welfare. 
      
      Adding isolated vertices to the seat graph does not make the problem easier since we can add dummy agents to the above reduction which have zero preference towards every agent. 

      \paragraph*{Clique-graphs.} 
      If $G$ is a clique-graph we can reduce from the problem \DkS, which is \NPh\ even on graphs with maximum degree three. 
      \decprob{\DkS}
      {An undirected graph~$\hat{G}$, non-negative integers~$h$ and~$m$.}
      {Does $\hat{G}$ have an $h$-vertex subgraph with at least~$m$ edges?}
      
      Let $(\hat{G}, h,m)$ denote an instance of \DkS. 
      For each vertex~$v_i \in V(\hat{G})$ we create one agent named $p_i$ and set $\sat[p_i](p_j)=\sat[p_j](p_i)=1$ if and only if $\{v_i,v_j\} \in E(\hat{G})$. 
      The non-mentioned preferences are set to zero. 
      The seat graph~$G$ consists of a clique of size~$\nonisolated\coloneqq h$ and $|V(\hat{G})|-h$ isolated vertices. 		
      Clearly, the preferences are binary an symmetric.
      Moreover, since each vertex in $\hat{G}$ has maximum degree three, it follows that $\maxoutdeg = 3$. 
      
      It remains to show that $\hat{G}$ has an $h$-vertex subgraph with at least~$m$ edges if and only if there is an arrangement~$\sigma$ with $\wel \geq 2m$.
      This follows from the fact that each edge $\{v_i, v_j\}$ in an $h$-vertex subgraph of $\hat{G}$ contributes exactly two to the welfare of an arrangement~$\sigma$. 		
      
      \paragraph*{Stars-graphs.} 
      We provide a polynomial-time reduction from the \NPc\ problem \ExactCover~\cite{garey1979}. 
      \decprob{\ExactCover}
      {A finite set of elements $U = \{u_1, \ldots , u_{3n}\}$, subsets $S_1, \ldots, S_m \subseteq U$ with $|S_j| = 3$ for each $j \in [m]$.}
      {Does there exist an exact cover $J \subseteq [m]$ for $U$, i.e., $\bigcup_{j \in J} S_j = U$ and each element occurs in exactly one set in the cover?}
      This problem remains \NPh\ even if each element appears in at most three subsets. 
      
      Let $(U, (S_j)_{j \in [m]})$ denote an instance of \ExactCover. 			
      For each element $u_i \in U$ we create an \myemph{element-agent}~$p_i$. 
      For each set $S_j$ we create two \myemph{set-agents}~$q_j^1, q_j^2$. 
      The binary and symmetric preferences are defined in the following way (see \cref{fig:MWA_delta_stars} for the corresponding preference graph):
      \begin{compactitem}[--]
        \item For each element $u_i \in U$, set $\sat[p_i](q_j^1) = \sat[q_j^1](p_i) = 1$ if $u_i \in S_j$.
        \item For each subset $S_j$, set $\sat[q_j^1](q_j^2) = \sat[q_j^2](q_j^1) = 1$. 
        \item The non-mentioned preferences are set to zero. 
      \end{compactitem}
  		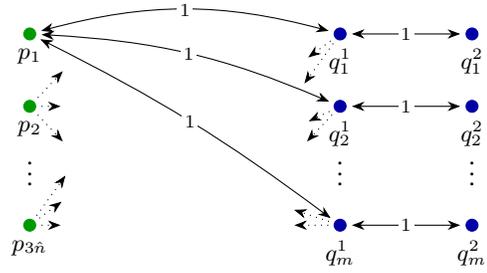
\begin{figure}[t]
  			\centering
  			\begin{tikzpicture}[>=stealth', shorten <= 2pt, shorten >= 2pt]
  			\def \xss {10ex}
  			\def \yss {5ex}
  			\def \ss {2.5ex}

  			\node[nodeW] (p1) {};
  			\node[nodeU, right = 2.5*\xss  of p1] (q11) {};
  			\node[nodeU, right = \xss  of q11] (q21) {};
  			\foreach \x / \typ in {p/W,q1/U,q2/U} {
  				\foreach \y / \skp  in {2/1,3/3} {
  					\node[node\typ, below = \skp * \yss  of \x1] (\x\y) {};
  				}
  				\path (\x2) -- node[auto=false]{\vdots} (\x3);
  			}

  			\begin{footnotesize}
	  			\foreach \y  in {1,2} {
	  				\node[below = 0pt of p\y] {$p_{\y}$};
	  			}
	  			\node[below = 0pt of p3] {$p_{3\hat{n}}$};
	  			\foreach \x / \pos in {1/below,2/below} {
	  				\foreach \y  in {1,2} {
	  					\node[\pos = 0pt of q\x\y] {$q_\y^{\x}$};
	  				}
	  				\node[\pos = 0pt of q\x3] {$q_m^{\x}$};
	  			}
	  		\end{footnotesize}

  			\begin{scriptsize}
  			\foreach \i in {1,2,3} {
  				\draw[Stealth-Stealth] (q1\i) edge node[pos=0.5, fill=white, inner sep=1pt] {$1$} (q2\i);
  			}
  			
  			\foreach \i/\b in {1/15,2/10,3/5} {
  				\draw[Stealth-Stealth] (p1) edge[bend left = \b] node[pos=0.5, fill=white, inner sep=1pt] {$1$} (q1\i);
  			}
  			\end{scriptsize}

  			\node[right = \ss of p2] (p22) {};
  			\node[above = \ss of p22] (p21) {};
  			\node[below = \ss of p22] (p23) {};
  			
  			\node[right = \ss of p3] (p33) {};
  			\node[above = 0.5*\ss of p33] (p31) {};
  			\node[above = 1.5*\ss of p33] (p32) {};
  			
  			\foreach \i in {2,3} {
  				\foreach \y in {1,2,3} {
  					\draw[dotted, -Stealth](p\i)--(p\i\y);
  				}
  			}
  			
  			\node[left = \ss of q11] (r11) {};
  			\node[below = 0.5*\ss of r11] (r12) {};
  			\node[below = 1.5*\ss of r11] (r13) {};
  			
  			\node[left = \ss of q12] (r21) {};
  			\node[below = 0.1*\ss of r21] (r22) {};
  			\node[below = 1*\ss of r21] (r23) {};
  			
  			\node[left = \ss*1.5 of q13] (r33) {};
  			\node[above = 0.05*\ss of r33] (r32) {};
  			
  			\foreach \i in {1,2,3} {
  				\foreach \y in {2,3} {
  					\draw[dotted, -Stealth](q1\i)--(r\i\y);
  				}
  			}
  			\end{tikzpicture}
  			\caption{Preference graph for \cref{thm:MWA_delta_clique+path-cycle+stars_bin-symm} for stars-graphs as seat graph, where $u_1 \in S_1 \cap S_2 \cap S_m$. Each element is contained in three sets. Therefore, the corresponding element-agent has a positive preference towards three set-agents, which ``contain'' the element. We indicate this through dashed arcs which have weight~$1$. Similarly, a dashed arc from a set-agent has weight~$1$ and goes to element-agents ``contained'' in this set.} 
  			\label{fig:MWA_delta_stars}
  		\end{figure}
  
      The seat graph~$G$ consists of~$n$ stars with four leaves each and $2m-2n$ isolated vertices.
      Since we can assume that each element is contained in only three sets and each set contains three elements, it holds $\maxoutdeg = 4$. 		
      
      It remains to show that there exists a cover of size $n$ if and only if there is an arrangement $\sigma$ with $\wel \geq 8n$. 
      
      Given a cover $J \subseteq [m]$ of $(U, (S_j)_{j \in [m]})$, we assign for each $j \in J$, agent $q_j^1$ to the center of a star and $q_j^2$ together with the three element-agents contained in $S_j$ to its leaves. 
      For this arrangement it holds $\wel = 8n$. 
      
      Conversely, given an arrangement $\sigma$ with $\wel = 8n$ we can observe that the utility of each non-isolated agent is positive. 
      Otherwise we cannot achieve the desired welfare in a seat graph with $4n$ edges and binary preferences.
      From this we can conclude that no set-agent $q_j^2$ or element-agent~$p_i$ is assigned to the center of a star, because the number of agents which have a positive preference towards $q_j^2$ (resp.\ $p_i$) is one (resp.\ three).
      This means that $n$ different set-agents $q_j^1$ are assigned to the star centers. 
      Moreover, the three element-agents contained in the corresponding sets together with the second set-agent are assigned to the leaves. 
    \end{proof}}
  
  The previous hardness results motivate us to study the combined parameters $(\nonisolated, \maxoutdeg)$ for symmetric preferences. We conclude the section with another color-coding based algorithm which is more involved than the one for \cref{thm:MWA_k_path-cycle+stars} since it works for arbitrary seat graphs.	
  
  \newcommand{\mwakdeltasymmfpt}{%
    For symmetric preferences, \MWA is \fpt\ wrt.\ $\nonisolated+\maxoutdeg$.
  }
\statementarxiv{theorem}{thm:MWA_k+delta_symm}{\mwakdeltasymmfpt}
  \appendixproofwithstatement{thm:MWA_k+delta_symm}{\mwakdeltasymmfpt}{
    \begin{proof}		
      Let $I = (\agents, (\sat)_{p\in \agents}, G)$ be an \MWA instance with
      symmetric preferences. Then, for each vertex, every in-neighbor in the preference graph is also an out-neighbor.
      Hence, we consider the underlying preference as an undirected graph. 
      This means that each non-isolated agent has bounded in- and out-degrees. Hence, we can use random separation to separate the non-isolated agents from their isolated neighbors.
      The approach is as follows: Color every agent independently with color~$r$ or $b$, each with probability $1/2$.
      We say that a coloring $\chi\colon \agents \to \{r, b\}$ is \myemph{successful} if there exists an optimal arrangement~$\sigma$ such that
      \begin{compactenum}[(i)]
        \item $\chi(p)=r$ for each $p \in W$ and
        \item $\chi(p)=b$ for each $p \in (\agents \setminus W)\cap 
        \Neigh{\prefgraph}(W)$, where $W\coloneqq\{p \in \agents~|~\degr{G}(\sigma(p)) \geq 1\}$ denotes the set of agents which are non-isolated in $\sigma$. %
      \end{compactenum}
      Since the seat graph has $\nonisolated$ non-isolated vertices and the 
      degree of each agent in the preference graph is bounded 
      by~$\maxoutdeg$, we infer $|W|+|(\agents \setminus W)\cap 
      \Neigh{\prefgraph}(W)| \leq 
      \nonisolated (1+\maxoutdeg)$. 
      Hence, the probability that a random coloring is successful is at least 
      $2^{-\nonisolated (1+\maxoutdeg)}$.
      
      Let $\redcomp$ be the subset of agents colored with~$r$ in~$\chi$.
      By definition of a successful coloring, we know for each connected component $C$ of $\redcomp$, that the agents in $C$ are all assigned to either isolated or non-isolated vertices. 
      Therefore, the size of each component is bounded by $\nonisolated$. 
      It remains to decide which non-isolated component to take and how to 
      assign them. 

      Let the vertices of the seat graph $G$ be denoted by
      $\{1, 2, \ldots, \nonisolated\}$. We design a simple algorithm using 
      color-coding as follows. We color the red agents $\redcomp$ uniformly 
      at random with colors $[\nonisolated]$. The $\nonisolated$ colors 
      one-to-one correspond to the $\nonisolated$ non-isolated seats in the 
      seat graph. 
      Let $\sigma$ be a hypothetical optimal arrangement and $\chi'\colon 
      \redcomp \to [\nonisolated]$ a coloring.
      We say~$\chi'$ is \myemph{good} (wrt.\ $\sigma$) if the agent at the $i$th vertex of~$G$ is colored~$i$, i.e., $\chi'(\sigma^{-1}(i))=i$, for each $i \in [\nonisolated]$. Note that given a solution $\sigma$, the probability that the $\nonisolated$ non-isolated agents are colored with pairwise distinct colors is at least~$e^{-k}$~\cite{cygan15}.	
      
      Since for each component all agents are assigned to either non-isolated or isolated vertices, we first check for each component, if each color appears at most one.	
      If there are two agents $p_1, p_2$ with $\chi'(p_1) = \chi'(p_2)$, then this component is assigned to isolated vertices in a good coloring. 
      Since there are $\Oh(n)$ components, this can be done in time $\Oh(k^2 \cdot n)$.
      
      For each of the remaining components we individually compute the welfare this component contributes if it is selected. 
      In this regard, we observe that the utility of an agent only depends on the agents inside the same component as all preferences between two components is zero. Hence, we compute for each agent $p$ with $\chi'(p) = i$ in a component $C$ his utility
      $\util[p]{\sigma} =\sum_{q\in 
        C\setminus\{p\},\chi'(q)=j,\{i,j\}\in 
        E(G)} \sat(q)$.

      Finally, we use dynamic programming (DP) to select from the remaining components those whose colors match the seats, sizes sum up to $k$, and which maximize the welfare. 
      Let $C_1,C_2,\ldots ,C_m$ be the remaining weakly connected components.
      We define a DP table where an entry $T[S,i]$ returns the maximum welfare of a partial arrangement assigning the agents of the first~$i$ components, where 
      \begin{inparaenum}[(i)]
        \item no two agents are colored with the same color and
        \item each color in~$S$ is used once.
      \end{inparaenum}
      
      We initiate our table for the first component as follows:
      \begin{align*}		
        T[S,1] = \begin{cases}
          0 & \text{ if } S = \emptyset, \\
          \sum_{p \in C_1} \util{\sigma} & \text{ if } S = \bigcup_{p \in C_i} \{\chi(p)\}, \\
          - \infty & \text{ else.}
        \end{cases}
      \end{align*}
      
      When considering a component, it is assigned completely to either non-isolated or isolated vertices. 
      Hence, we obtain the following recursion:
      \begin{align*}
        T[S,i] = &\max\Bigg\{ T[S,i-1], \\
                 &T\left[S\setminus\bigcup_{p \in C_i} \{\chi(p)\}, i-1\right] + \sum_{p \in C_i} \util{\sigma} \Bigg\}
      \end{align*}
      We return $T[[\nonisolated], m]$ as the maximum welfare of this instance. 
      The entries of this table can be computed in time $2^\nonisolated\cdot n^{\Oh(1)}$ as $S \subseteq [\nonisolated]$. 
      The probability of a successful and good coloring is $2^{-\nonisolated(1+\maxoutdeg)}e^{-\nonisolated}$. 
      Hence, after repeating this algorithm $2^{\nonisolated(1+\maxoutdeg)}e^{\nonisolated}$ times we obtain a solution with high probability.				
      We can get a de-randomized \fpt\ algorithm using set families and perfect hash functions~\cite{cygan15} to replace the two coloring steps respectively.
    \end{proof}
  }
  
  \section{\MUAf}\label{sec:mua}
  \appendixsection{sec:mua}        
  In this section, we focus on maximizing the minimum utility.
  Under non-negative preferences, \fpt\ algorithms (wrt.\ $k$) exist since any non-trivial instance has at most $k$ agents:
  \newcommand{\obsmua}{%
    For non-negative preferences, \MUA is \fpt\ wrt.\ $\nonisolated$, and becomes polynomial-time solvable if the seat graph is a clique-graph.%
  }
\statementarxiv{proposition}{obs:MUA_nonneg-clique-k}{\obsmua}
  \appendixproofwithstatement{obs:MUA_nonneg-clique-k}{\obsmua}{
    \begin{proof} 
      For non-negative preferences, if the seat graph contains isolated vertices then $\egal = 0$. 
      Otherwise, $\nonisolated = |\agents|$, i.e., we can bound the input size by $\nonisolated$. 
      Hence, we obtain fixed-parameter tractability since we have a linear problem kernel.
      
      For clique-graphs, if there are isolated vertices, then $\egal = 0$ for every arrangement $\sigma$. 
      Otherwise, since the seat graph is a clique, the seat of every agent is indifferent. 
      Hence, every arrangement is an optimal solution.
    \end{proof}
  }  
  
The presence of negative preferences excludes any \fpt\ algorithm (wrt.\ $\nonisolated$) for \MUA\ since it generalizes \textsc{Clique}: %
  \newcommand{\muakcliquewhard}{%
    \MUA is \Wh\ wrt.\ $\nonisolated$, even for a clique-graph and for symmetric or strict preferences. %
  }
\statementarxiv{theorem}{thm:MUA_k_clique_symm+strict}{\muakcliquewhard}
  \appendixproofwithstatement{thm:MUA_k_clique_symm+strict}{\muakcliquewhard}{
    \begin{proof}		
      \textbf{Symmetric preferences.} 
      We provide a parameterized reduction from the \Wh\ problem \Clique, 
      parameterized by the solution size.
      Let $(\hat{G}, h)$ denote an instance of \Clique.
      For each vertex~$v_i \in V(\hat{G})$ we create one agent $p_i$.
      We set $\sat[p_i](p_j)=\sat[p_j](p_i)=0$ if and only if $\{v_i,v_j\} \in E(\hat{G})$. 
      Otherwise, the preference value is set to~$-1$. 
      The seat graph $G$ consists of a clique of size~$\nonisolated\coloneqq h$ and $|V(\hat{G})|-h$ isolated vertices.
      
      It is straightforward to verify that $\hat{G}$ has a size-$h$ clique if and only if there is an arrangement~$\sigma$ with $\egal \geq 0$.
      
      \paragraph*{Strict preferences.} 
      Now, we turn to strict preferences and modify the above reduction.
      The only differences are the preferences.
      For each neighbor~$v_j$ of $v_i$ in $\hat{G}$ we set $\sat[p_i](p_j)$ to a 
      positive value from $[\hat{n}]$, where $\hat{n}\coloneqq |V(\hat{G})|$.
      For each non-neighbor $v_j$ of $v_i$ we set $\sat[p_i](p_j)$ to a negative 
      value from $\{-\hat{n}(\hat{n}+1), \ldots, -\hat{n}^2\}$. 
      In any case, the values can be assigned arbitrary, but each preference value can appear only at most once.
      Again, it is straightforward to verify that~$\hat{G}$ has a size-$h$ clique if and only if there is an arrangement~$\sigma$ with 
      $\egal > 0$.		
    \end{proof}
  }

  For simple seat graphs, we can again combine color-coding with dynamic programming to obtain \fpt\ algorithms for the single parameter~$\nonisolated$:
  \newcommand{\muafptkpath}{%
    For stars-, path-, and cycle-graphs, \MUA is \fpt\ wrt.\ $\nonisolated$.%
  }
\statementarxiv{theorem}{thm:MUA_k_path-cycle+stars}{\muafptkpath}
  \appendixproofwithstatement{thm:MUA_k_path-cycle+stars}{\muafptkpath}{
    \begin{proof}
      Let $I = (\agents, (\sat)_{p\in \agents}, G)$ be an \MUA instance and the vertices of the seat graph $G$ be denoted by $\{1, 2, \dots, \nonisolated\}$.
      
      We iterate over the possible value~$\beta$ of the minimum utility of an optimal solution; there are polynomially many.
      We design a simple algorithm using color-coding as follows. 
      We color the agents $\agents$ uniformly at random with colors 
      $[\nonisolated]$. 
      The $\nonisolated$ colors one-to-one correspond to the $\nonisolated$ 
      non-isolated seats in the seat graph.
      Let $\sigma$ be a hypothetical optimal arrangement and $\chi\colon 
      \agents\to[\nonisolated]$ a coloring.
      We say $\chi$ is \myemph{good} (wrt.\ $\sigma$) if the agent at the $i$th vertex of $G$ is colored $i$, i.e., $\chi(\sigma^{-1} (i)) = i$, for each $i \in [\nonisolated]$. 
      Note that given a solution $\sigma$, the probability that the $\nonisolated$ non-isolated agents are colored with pairwise distinct colors is at least $e^{-\nonisolated}$~\cite{cygan15}.
      
      \paragraph*{Stars-graphs.} 	
      Given a coloring $\chi\colon P \to [\nonisolated]$, we check whether there exists an arrangement~$\rho$ such that 
      \begin{inparaenum}[(i)]
        \item $\chi(\rho^{-1} (i)) = i$ and
        \item for each $p \in P$, it holds that $\util{\sigma}\geq \beta$.
      \end{inparaenum}
      Since agents in different stars do not interact with each other, we will check for each star if there is an arrangement, where each agent assigned to this star has utility at least~$\beta$.
      Let $k_c$ be the center of a star and $k_1, \ldots, k_\ell$ its leaves.
      Towards this, for each possible star center vertex $p$, %
      and each leaf $k_i \in \{k_1, \ldots, k_\ell\}$ we greedily choose an agent~$q$ with maximum $\sat(q)$ such that $\sat[q](p) \geq \beta$ and $\chi(q)=k_i$, i.e., we check
      \begin{align*}
        \bigvee_{\substack{p \in \agents,\\ \chi(p)=k_c}} \left[\sum_{i = 
        1}^{\ell} \left( \max_{\substack{q \in \agents\setminus\{p\}, \\ 
        \chi(q)=k_i, \sat[q](p)\geq \beta}} \sat (q) \right) \geq \beta \right].
      \end{align*}
      If we find such an arrangement $\rho$, return it; otherwise we return no.
      Since the number of stars as well as the number of leaves of each star is $\Oh(\nonisolated)$, the computation time is $\Oh(k^2 \cdot n^2)$.

      \paragraph*{Path-graphs.}
      The approach is very similar to the color-coding approach for finding a 
      $k$-path~\cite{AYZ95colorcoding}.
      We define a dynamic programming (DP) table, where an entry $T[S,p,q]$ with $S \subseteq [\nonisolated], p,q \in \agents$, is true if there is an 
      arrangement 
      assigning agents on a $|S|$-path, where 
      \begin{compactenum}[(i)]
        \item no two agents are colored with the same color,
        \item each color in $S$ is used once,
        \item agent $p$ is assigned to the last vertex on the path,
        \item agent $q$ is assigned next to $p$, i.e., $q$ is the penultimate agent on the path, and
        \item the utility of each agent on the path up to $q$ (without $p$) is at least~$\beta$. 
      \end{compactenum}
      
      In the base case we have to check $\util[q]{\sigma} = \sat[q](p) \geq 
      \beta$ since the path consists only of a single edge.
      Hence, for each pair of agents ${p,q}$ and each $S \subseteq [k]$ with $|S| 
      \leq 2$, we define
      \begin{align*}
        T[S,p,q] = \textit{true} \Leftrightarrow S = \{\chi(p), \chi(q)\} &\wedge 
                                                                            \chi(p) \neq \chi(q) \\
                                                                          &\wedge \sat[q](p) \geq \beta.
      \end{align*}
      
      \noindent For $|S| > 2$, the following recurrence holds:
      \begin{align*}
        &T[S,p,q] = \bigvee_{\substack{r \in \agents, \\ \sat[q](p) + \sat[q](r) \geq 
        \beta}} T[S\setminus \{\chi(p)\}, q, r].
      \end{align*}
      Observe that the above recurrence correctly computes a path on $\nonisolated$ 
      agents since no color is repeated and hence no agent is repeated. 
      Specifically, it is a path ending at agent $p$ where each agent has utility at 
      least $\beta$.
      To decide whether there exists an arrangement~$\sigma$ with $\egal \geq \beta$ we 
      check whether $\bigvee_{p,q \in \agents, \sat[p](q)\geq \beta} T[[k], p,q]$ is true.

      Since this DP table stores binary values for 
      all subsets of $\nonisolated$ colors and two agents $p,q$ and each binary value 
      can be 
      computed by looking up up to $\nonisolated\cdot n$ previous entries, the table 
      can be 
      computed in time $2^\nonisolated\cdot n^{\Oh(1)}$. 
      
      \paragraph*{Cycle-graphs.} 
      In this case we can extend the DP table for path-graphs by adding a start agent~$s$, which is assigned to the first vertex in the path and an adjacent second agent~$r$. 
      We can check whether there exists an arrangement~$\sigma$ with $\egal \geq \beta$ by
      \begin{align*}
        \bigvee_{\substack{p,q,s,r \in \agents \\ \sat[p](q) + \sat[p](s) \geq \beta \\ 
        \sat[s](r) + \sat[s](p) \geq \beta}} T'[[\nonisolated], p,q,s,r].
      \end{align*}
      Again, given a coloring $\chi$ the time needed to fill this DP table's entries is $2^\nonisolated \cdot n^{\Oh(1)}$. 
    \end{proof}
  }

  Unfortunately, \MUA cannot be solved in polynomial time even if each agent knows a few other agents and the preferences are symmetric, unless \PP{}={}\NPB.
  
  \newcommand{\muadelta}{%
    For symmetric preferences and for each of the following restrictions, \MUA remains \NPh\ for constant~$\maxoutdeg$:
    \begin{inparaenum}[(i)]
      \item a clique-graph,
      \item a path-graph and non-negative preferences,
      \item a cycle-graph and binary preferences,
      \item a stars-graph where each star has two leaves and binary preferences.
    \end{inparaenum}%
  }
\statementarxiv{theorem}{thm:MUA_delta_clique+star+path-cycle_symm}{\muadelta}
  \appendixproofwithstatement{thm:MUA_delta_clique+star+path-cycle_symm}{\muadelta}{
    \begin{proof}
      \begin{compactenum}[(i)]
        \item For clique-graphs we provide a polynomial-time reduction from \IS, which is \NPh\ even on cubic graphs~\cite{garey1979}.
        \ifarxiv
        	\decprob{\IS}
        	{An undirected graph~$\hat{G}$ and a non-negative integer~$h$.}
        	{Does~$\hat{G}$ admit a size-$h$ independent set, i.e., a subset of vertices where no two agents are adjacent?}
        \fi 
        
        Let $(\hat{G},h)$ be an instance of \IS. 	
        For each vertex $v_i \in V(\hat{G})$ we create one agent named~$p_i$ and set $\sat[p_i](p_j)=\sat[p_j](p_i) = -1$ if and only if $\{v_i,v_j\} \in E(\hat{G})$. 
        Otherwise, the preference value is set to zero.
        Hence, the preferences are symmetric.
        Since~$\hat{G}$ is cubic, we obtain $\maxoutdeg = 3$. 
        The graph induced by the non-isolated vertices is a clique of size $\nonisolated\coloneqq h$ and $|V(\hat{G})|-h$ isolated vertices.
        
        It is straightforward that $\hat{G}$ admits a size-$h$ independent set iff.\ there is an  arrangement 
        with $\egal \geq 0$.
        
        \item For a path-graph we give a polynomial-time reduction from \HamPath. 
        This problem remains \NPh\ even if the two endpoints of the path are specified~\cite{garey1979}. 
        Given an instance $(\hat{G},s,t)$ with $s,t \in V(\hat{G})$ being the specified endpoints of \HamPath we create for each vertex $v_i \in V(\hat{G})$ one agent named $p_i$ and set $\sat[p_i](p_j) = 1$ if and only if $\{v_i,v_j\} \in E(\hat{G})$. 
        Moreover, we add two additional agents $q_s, q_t$ with $\sat[p_s](q_s) = \sat[q_s](p_s) = \sat[p_t](q_t) = \sat[q_t](p_t) = 2$.
        All non-mentioned preferences are set to zero (see \cref{fig:MUA_delta_path} for the corresponding preference graph).	
        The graph induced by the non-isolated vertices in the seat graph~$G$ consists of a single path with $|V(\hat{G})|+2$ vertices.        
        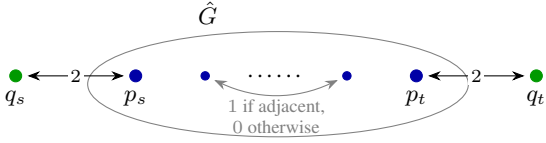
\begin{figure}[t]
        	\centering
        	\begin{tikzpicture}[>=stealth', shorten <= 2pt, shorten >= 2pt]
        	\def \xs {9ex}
        	\def \ys {14ex}
        	\def \ss {2.5ex}
        	
        	\node[nodeW] (x1) {};
        	\node[nodeU, right = \xs of x1] (p1) {};
        	\node[nodeU, right = 2.5*\xs of p1] (p2) {};
        	\node[nodeW, right = \xs of p2] (x2) {};

        	\begin{footnotesize}
	        	\node[below = 0pt of x1] {$q_s$};
	        	\node[below = 0pt of p1] {$p_{s}$};
	        	\node[below = 0pt of x2] {$q_t$};
	        	\node[below = 0pt of p2] {$p_{t}$};
	        	\path (p1) -- node[auto=false]{\ldots\ldots} (p2);
	        \end{footnotesize}

        	\node[ellipse, draw = gray, minimum width = 32ex, minimum height = 1.4cm, label={above left:{\footnotesize$\hat{G}$}}] (ell) at (\xs*2.45,-0.05*\ys) {};

        	\begin{scriptsize}
        	\foreach \x in {1,2} {
        		\draw[Stealth-Stealth] (p\x) edge  node[pos=0.5, fill=white, inner sep=1pt] {$2$} (x\x);
        	}
        	
        	\node[nodeU, right = \xs-2ex  of p1] (pi) {};
        	\node[nodeU, left = \xs-2ex  of p2] (pj) {};
        	\draw[Stealth-Stealth] (pi) edge[bend right=25, gray] node[below, inner sep=1pt] {\begin{tabular}{c}$1$ if adjacent, \\ $0$ otherwise \end{tabular}} (pj);
        	\end{scriptsize}
        	\end{tikzpicture}
        	\caption{Preference graph for \cref{thm:MUA_delta_clique+star+path-cycle_symm} for path-graphs.}
        	\label{fig:MUA_delta_path}
        \end{figure}
        Clearly, the preferences are binary, and hence non-negative. Since \HamPath remains \NPh\ even on cubic graphs, we obtain $\maxoutdeg = 4$.
        
        Before we show the correctness, i.e., there exists a Hamiltonian path in 
        $\hat{G}$ starting in $s$ and ending in $t$ iff.\
        if there is an arrangement~$\sigma$ with $\egal \geq 2$, we observe the following: 
        In every arrangement~$\sigma$ with $\egal \geq 2$, agents $p_s, q_s$ resp.\ $p_t, q_t$ are assigned consecutively. 
        Moreover, $q_s$ and $q_t$ are assigned the two endpoints of the path since otherwise we can find an agent $p \in 
        \agents$ with $\util{\sigma} \leq 1$.
        
        With this observation the correctness follows directly.
        
        \item Since \HamCycle is \NPh\ even on cubic graphs, we can conclude from a result by Bodlaender et al.~\shortcite{bodlaender2020}  that \MUA remains \NPh\ even for a cycle-graph and binary and symmetric preferences with $\maxoutdeg = 3$.
        
        \item For stars-graphs we provide a polynomial-time reduction from \PathPartition which is \NPh\ even on subcubic graphs~\cite[also follows from \cite{dyer85}]{bevern17}. 
        \decprob{\PathPartition}
        {An undirected graph~$\hat{G}$.}
        {Can $V(\hat{G})$ be partitioned into $|V(\hat{G})|/3$ mutually disjoint vertex subsets $V_1, \ldots, V_{|V(\hat{G})|/3}$, such that each subgraph $\hat{G}[V_i]$ is a $P_3$, i.e., a path with three vertices?}
        
        Given an instance $\hat{G}$ of \PathPartition we create for each vertex $v_i \in V(\hat{G})$ one agent named~$p_i$ and set $\sat[p_i](p_j)=\sat[p_j](p_i)=1$ if and only if $\{v_i,v_j\} \in E(\hat{G})$. 
        Otherwise, the preference value is set to zero.
        The seat graph $G$ consists of $|V(\hat{G})|/3$ $P_3$'s.        
        Since each vertex in $\hat{G}$ has at most three neighbors, we constructed an \MUA instance with symmetric preferences and $\maxoutdeg = 3$.
        
        It remains to verify that there exists a partition of $\hat{G}$ into~$P_3$'s if and only if there is an arrangement~$\sigma$ with $\egal \geq 1$. 
        The ``only if'' part is straightforward. 
        For the ``if'' part let~$\sigma$ be an arrangement with $\egal \geq 1$. 
        Then for each agent $p_i$ assigned to the endpoint and~$p_j$ assigned to the middle vertex in a $P_3$ it holds that the vertex~$v_i$ is adjacent to~$v_j$. 
        Since the preferences are symmetric, $v_j$ is also adjacent to $v_i$, i.e., each vertex has a neighbor in the partition. 
        From this we can conclude that there exists a partition of~$\hat{G}$ into~$P_3$'s. \qedhere
      \end{compactenum}
    \end{proof}
  }	

  \MUA admits a polynomial-size problem kernel for the combined parameters~$(\nonisolated, \maxoutdeg)$.
  \newcommand{\muakdelta}{%
    \MUA is \fpt\ wrt.\ $\nonisolated+\maxoutdeg$.
  }
  \begin{theorem}%
    \label{thm:MUA_k+delta}
    \muakdelta
  \end{theorem}
  \begin{proof}
  The idea is to obtain a polynomial-sized problem kernel, i.e., an equivalent instance of size $(\nonisolated+\maxoutdeg)^{\Oh(1)}$, or solve the problem in polynomial time.  First, if the seat graph has no isolated vertices, then $ |\agents| = \nonisolated$ and we have a linear-sized kernel.
  Otherwise, $\egal \leq 0$ holds for every arrangement~$\sigma$.
  Further, we observe that in every directed graph with maximum out-degree~$\maxoutdeg$, there is always a vertex with in-degree bounded by $\maxoutdeg$.
  That is, the sum of in- and out-degrees of this vertex is at most $2\maxoutdeg$.
  Hence, we iteratively select an agent~$p$ with minimum in-degree in the preference graph, put him to our solution \newE{$S$}, and delete all in- and out-neighbors of~$p$. 
  If, after~$\nonisolated$ steps, we can find a set~$S$ of $\nonisolated$ ``independent'' agents (they do not have arcs towards each other), then we can assign them arbitrarily to the non-isolated vertices. 
  Since the preference between each two agents in $S$ is zero, the utility of each agent in~$S$ is also zero. 
  Hence, we found an arrangement~$\sigma$ with $\egal = 0$. 
  If we could not find $\nonisolated$ independent agents, then the instance has at most $\nonisolated(1+2\maxoutdeg)$ agents since in each step we deleted at most $1+2\maxoutdeg$ agents.
  It is straightforward that the approach above runs in polynomial time.
  \end{proof}
  \looseness=-1
  \section{\EFAf}\label{sec:efa}
  \appendixsection{sec:efa}
  In this section, we consider envy-freeness.
  First, we observe the following for non-negative preferences and we will use it extensively in designing both algorithms and reductions. 
  
  \newcommand{\obsefanonneg}{%
    Let $p$ be an agent with non-negative preferences.
    Then, for each envy-free arrangement it holds that if~$p$ is isolated, then every~$q$ with $\sat(q)\!>\!0$ is isolated as well.
  }
  \begin{obs}\label{lem:EFA_non-neg_cond}
    \obsefanonneg
  \end{obs}
  \appendixproofwithstatement{lem:EFA_non-neg_cond}{\obsefanonneg}{
    \begin{proof}
      Given an envy-free arrangement $\sigma$ and agent $p \in \agents$ with non-negative preferences, suppose this statement is not satisfied, i.e., $p$ is isolated and there is a non-isolated agent $q \in \agents$ with $\sat(q)>0$. 
      Clearly, $\util[p]{\sigma}=0$, i.e., $p$ envies the agent assigned to a neighbor of $\sigma(q)$.
    \end{proof}}

  \noindent  By \cref{lem:EFA_non-neg_cond}, deciding envy-freeness is easy for clique-graphs when the preferences are additionally symmetric.
    
  \newcommand{\efacliquenonnegsymm}{%
    For clique-graphs, and non-negative and symmetric preferences, \EFA is 
    polynomial-time solvable. %
  }
\statementarxiv{proposition}{thm:EFA_clique_nonneg+symm}{\efacliquenonnegsymm}
  \appendixproofwithstatement{thm:EFA_clique_nonneg+symm}{\efacliquenonnegsymm}{
    \begin{proof}
      Let $I = (\agents, (\sat)_{p\in \agents}, G)$ be an instance of \EFA, where the seat graph $G$ consists of a clique of size $\nonisolated$ and isolated vertices.       
      Since the agent's preferences are symmetric, each weakly connected component of the preference graph~$\prefgraph$ is a strongly connected component. 
      Therefore, in the rest of this proof we will refer to a weakly and strongly connected component as a component of~$\prefgraph$. 
      
      Because all seats in the clique are indifferent and the preferences are non-negative, only an isolated agent can envy a non-isolated agent.  
      By \cref{lem:EFA_non-neg_cond}, we obtain for an \efArr that the agents in each component are all assigned to either isolated or non-isolated vertices. 
      
      Hence, finding an \efArr of $I$ reduces to finding a subset of the components of~$\prefgraph$, where the size of the components sums up to exactly $\nonisolated$. 
      This can be solved using dynamic programming. 
      
      Let $C_1, \ldots, C_m$ be an arbitrary ordering of the components of~$\prefgraph$.
      We define a DP table~$T$ where an entry $T[i, k']$ is true if there is $S \subseteq[i]$ where the sum of the sizes of the components of~$S$ is exactly~$k'$. 
      We start filling the table for $i = 1$:
      \begin{align*}
        T[1, k'] = \begin{cases}
          \textit{true} & \text{if } k' = |C_1| \text{ or } k' = 0, \\
          \textit{false} & \text{otherwise}.
        \end{cases}
      \end{align*}
      
      In each further step the considered component $C_i$ is completely assigned to either isolated or non-isolated vertices.
      Therefore, the following recurrence holds:
      \begin{align*}
        T[i, k'] = T[i-1, k'-|C_i|] \vee T[i-1, k'].
      \end{align*}
      We return yes if $T[m, k]$ is true, no otherwise.
      Since the preference graph has $\Oh(n)$ components and $k \leq n$, 
      computing this table's entries can be done in $\Oh(n^2)$ time. 
    \end{proof}
  }

  By \cref{obs:xp-k}, \EFA\ is polynomial-time solvable for constant~$\nonisolated$.
  In the next two theorems, we show that this result cannot be improved to obtain \fpt\ algorithms by providing a parameterized reduction from either \Clique\ or \IS (wrt.\ the solution size~$h$).
 We introduce a novel all-or-nothing gadget (see \cref{fig:EFA_k_path-cycle_symm} for an example) to enforce that only $f(h)$ many copies of the all-or-nothing gadgets can be non-isolated, which correspond to a solution of size $h$. 
\ifarxiv
 \newE{As the name suggests, either ``all'' agents in the gadget are assigned to non-isolated seats or ``nothing'' from the gadget is assigned to non-isolated seats.}
\fi 
 \todoE{review: give more intuition about all-or-nothing gadget? Update: added a sentence in arxiv version}
  \newE{In addition, \cref{thm:EFA_k_matching+clique+path-cycle_bin} corrects an error by Bodlaender et al.~\shortcite{bodlaender2020tech} and shows that \EFA\ remains \NPh\ for matching-graphs and strict preferences.
  A crucial observation in this setting is that not every non-isolated agent needs to be matched with his most preferred agent.}
  
  \newcommand{\efakbinstrictwhard}{%
    For each considered seat graph class, \EFA is \Wh\ wrt.\ $\nonisolated$ even if the preferences are binary or strict. %
  }
\statementarxiv{theorem}{thm:EFA_k_matching+clique+path-cycle_bin}{\efakbinstrictwhard}
  \appendixproofwithstatement{thm:EFA_k_matching+clique+path-cycle_bin}{\efakbinstrictwhard}{
    \begin{proof}
      \noindent   \textbf{Binary preferences.}		First, we consider the case with matching-graphs. %
      We prove this via a parameterized reduction from the \Wh\ problem \Clique, parameterized by the solution size.
      
      Let $(\hat{G}, h)$ denote an instance of \Clique.
      Without loss of generality, we assume that $h$ is even. 
      We construct an instance~$I$ of \EFA in the following way (see \cref{fig:EFA_k_matching_bin} for two adjacent vertices $v_i, v_j \in V(\hat{G})$):
      \begin{compactitem}[--]
        \item For each vertex~$v_i \in V(\hat{G})$, create $h^2$ \myemph{vertex-agents} $p_i^1, \ldots, p_i^{h^2}$, and set $\sat[p_i^z](p_i^{z+1})=\sat[p_i^{z+1}](p_i^z)=1$ for all~$z \in [h^2\!-\!1]$.
        \item For each edge~$e_\ell \in E(\hat{G})$, create one \myemph{set-agent} $q_{\ell}$ and set $\sat[p_i^1](q_{\ell}) = 1$ if $v_i \in e_\ell$.
        \item The non-mentioned preferences are set to zero.
      \end{compactitem}

      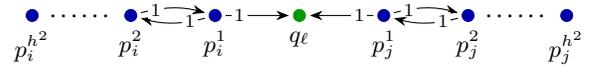
\begin{figure}[t]
      	\centering
      	\begin{tikzpicture}[>=stealth', shorten <= 1.5pt, shorten >= 1.5pt]
      	\def \xs {6ex}
      	\def \ys {4ex}

      	\node[nodeW] (q) {};
      	\node[below = 0pt of q] {\footnotesize$q_\ell$};
      	
      	\foreach \x / \pos in {i/left,j/right} {
      		\node[nodeU, \pos = \xs of q] (p\x1) {};
      		\node[below = 0pt of p\x1] {\footnotesize$p_\x^1$};
      		
      		\node[nodeU, \pos = \xs of p\x1] (p\x2) {};
      		\node[below = 0pt of p\x2] {\footnotesize$p_\x^2$};
      		
      		\node[nodeU, \pos = 1.2*\xs of p\x2] (p\x3) {};
      		\node[below = 0pt of p\x3] {\footnotesize$p_\x^{h^2}$};
      		
      		\path (p\x2) -- node[auto=false]{\ldots\ldots} (p\x3);
      	}

      	\begin{scriptsize}
      	\foreach \x in {i,j} {
      		\draw[-Stealth] (p\x1) edge  node[pos=0.25, fill=white, inner sep=1pt] {$1$} (q);
      		
      		\draw[-Stealth] (p\x1) edge[bend left=15]  node[pos=0.25, fill=white, inner sep=1pt] {$1$} (p\x2); 
      		\draw[-Stealth] (p\x2) edge[bend left=15]  node[pos=0.25, fill=white, inner sep=1pt] {$1$} (p\x1); 
      	}
      	\end{scriptsize}
      	\end{tikzpicture}
      	\caption{Preference graph for \cref{thm:EFA_k_matching+clique+path-cycle_bin} for each considered seat graph class, where $e_{\ell}=\{v_i, v_j\}$.}
      	\label{fig:EFA_k_matching_bin}
      \end{figure}
      The seat graph $G$ consists of $\nonisolated/2 \coloneqq  (h^3 + \binom{h}{2}) /2$ disjoint edges and $|V(\hat{G})|h^2+|E(\hat{G})|-k$ isolated vertices.
      It is straightforward that the preferences are binary and the construction can be done in polynomial time.
      
      It remains to show the correctness, i.e., $\hat{G}$ admits a size-$h$ clique~$C$ if and only if~$I$ admits an \efArr. 
      For the ``only if'' part, let $C$ be a size-$h$ clique of $\hat{G}$. 
      For each vertex $v_i \in C$ and $\ell \in [\frac{h^2}{2}]$, assign agents $p_i^{2\ell-1}, p_i^{2\ell}$ to the endpoints of the same edge in the seat graph. 
      After this step $h^3$ endpoints of edges are assigned. 
      Moreover, for every edge $e_\ell \in C$ we assign~$q_{\ell}$ to the remaining endpoints of edges in~$G$.
      The remaining agents are assigned to isolated vertices.			
      This arrangement $\sigma$ is envy-free because of the following:
      \begin{compactenum}[(i)]
        \item Every edge-agent $q_{\ell}$ is envy-free since $\sat[q_{\ell}](p) = 0$ for each $p \in \agents\setminus\{q_{\ell}\}$. 
        
        \item Every vertex-agent $p_i^z$ with $v_i \in C$, $z \in [h^2]$, is envy-free because $p_i^z$ has his maximum possible utility of one. 
        
        \item The remaining vertex-agents $p_i^z$ with $v_i \notin C$ and $z \in [h^2]$, have zero utility. 
        However, there is also no non-isolated agent towards which $p_i^z$ has a positive preference.
        Therefore, also all agents corresponding to vertices in $V(\hat{G})\setminus C$ are envy-free.
      \end{compactenum}
      
      For the ``if'' part, let~$\sigma$ be an \efArr of~$I$. 		
      By \cref{lem:EFA_non-neg_cond} it follows that for each vertex $v_i \in V(\hat{G})$, agents $p_i^1, \ldots, p_i^{h^2}$ are all assigned to either isolated or non-isolated vertices.
      Together with $h^3 + \binom{h}{2} < h^3 + h^2$, 
      we can conclude that at most $h$ different sets of vertices are assigned on the matching-edges. 
      
      Moreover, if an agent~$q_{\ell}$ for $e_\ell = \{v_i,v_j\}$ is non-isolated, then again by \cref{lem:EFA_non-neg_cond}, agents~$p_i^1$ and~$p_j^1$ (and all agents of $v_i$ and~$v_j$) are non-isolated. 
      Since $\binom{h}{2}+1$ edges are incident to at least $h+1$ vertices, at most $\binom{h}{2}$ edge-agents can be assigned to matching-edges. 
      Therefore, exactly~$h$ groups of vertex- and $\binom{h}{2}$ edge-agents are assigned to matching-edges, which have to form a clique in~$\hat{G}$. 
      
      Since matching-graphs are stars-graphs, we also obtain hardness for this graph class. 
      For a clique-, path-, and cycle-graph and binary preferences we can use the same reduction as above, but create a clique, path resp.\ cycle of size $h^3+\binom{h}{2}$ and $|V(\hat{G})|h^2+|E(\hat{G})|$ isolated vertices as seat graph. 
      The reasoning for the correctness works analogously. %
      
      \paragraph*{Strict preferences.} %
      We modify the reduction above. 
      Define $\hat{n} = |V(\hat{G})|$. 
      Again, we create for each vertex $v_i \in V(\hat{G})$ in total $h^2$ agents $p_i^1, \ldots, p_i^{h^2}$ and for each edge $e_\ell \in E(\hat{G})$ two agents $q_{\ell}^1, q_{\ell}^2$. 
      For each edge $e_\ell \in E(\hat{G})$ define the preferences as follows:
      \begin{compactitem}[--]
        \item $\sat[q_{\ell}^1](q_{\ell}^2) = \sat[q_{\ell}^2](q_{\ell}^1) = 0$.
        \item For each  $p \in \agents\setminus\{q_{\ell}^1,q_{\ell}^2\}$, set $\sat[q_{\ell}^1](p) = \sat[q_{\ell}^2](p)$ to a distinct negative value from $\{-1, \ldots, -\hat{n}h^2-2|E(\hat{G})|\}$. 
      \end{compactitem}
      For each vertex $v_i \in V(\hat{G})$ define the following preferences: 
      \begin{compactitem}[--]
        \item For each edge $e_\ell$ incident to $v_i$, set $\sat[p_i^1](q_{\ell})$ to a distinct positive value from $[\hat{n}-1]$.
        \item $\sat[p_i^{1}](p_i^2) = \sat[p_i^{2}](p_i^1) = \hat{n}$.
        \item For each $z \in [h^2/2]\setminus\{1\}$, set $\sat[p_i^{2z-1}](p_i^{2z}) = \sat[p_i^{2z}](p_i^{2z-1}) = 2$.
        \item For each $z \in [h^2/2-1]$, set $\sat[p_i^{2z}](p_i^{2z+1}) = \sat[p_i^{2z+1}](p_i^{2z}) = 1$.
        \item For each $z \in [h^2]$ and remaining agent $p$, set $\sat[p_i^{z}](p)$ to a distinct negative value from $\{-1, \ldots, -\hat{n}h^2-2|E(\hat{G})|\}$.
      \end{compactitem}
      For a path- and cycle-graph (resp.\ clique-graph) we construct an \EFA instance~$I$ as above but increase the non-negative preferences by $\hat{n}h^2+2|E(\hat{G})|$ (resp.\ $(h-1)(\hat{n}h^2+2|E(\hat{G})|)$).
      In any case, the values can be assigned arbitrary, but each preference value can appear at most once.       
      The seat graph induced by the non-isolated vertices consists of $\left(h^3+\binom{h}{2}\right)/2$ disjoint edges resp.\ a path, cycle, or clique with $h^3+2\binom{h}{2}$ vertices and with $\hat{n}h^2+|E(\hat{G})|$ isolated vertices.
      
      The correctness, i.e., $\hat{G}$ has a size-$h$ clique if and only if there is an \efArr, can be shown in a similar fashion as for the binary case. 
    \end{proof}
  }

  \newcommand{\thmefakhard}{%
    For clique-, stars-,  path-, and cycle-graphs, and for symmetric preferences, \EFA is \Wh\ wrt.\ $\nonisolated$.
  }
  \newcommand{\clmefaksymm}{
    Every \efArr~$\sigma$ satisfies:
    \begin{compactenum}[(i)]
      \item For each edge $e_\ell \in E(\hat{G})$, agents $q_{\ell}^2, \ldots, q_{\ell}^{h^4}$ are always assigned to isolated vertices.
      
      \item\label{efasymk-allornothing} If a vertex-agent~$p_i^z$ with $v_i \in V(\hat{G}), z \in [h(h-1)]$ is non-isolated, then all agents from $\{p_i^1, \ldots, p_i^{\newH{h(h-1)}}\}$ are non-isolated.   
      Moreover, if such a set of agents is non-isolated, then the seats of $p_i^{sh+z}$ and $p_i^{sh+z+1}$ for $s\in \{0\}\cup [h-2]$, $z\in [h-1]$ are adjacent. 
      
      \item\label{efasymk-edgepair} If $q_{\ell}^0$ or $q_{\ell}^1$ is non-isolated for some $e_\ell \in E(\hat{G})$, then~$q_{\ell}^0$ and $q_{\ell}^1$ are non-isolated and their seats are adjacent. 
      
      \item\label{efasymk-edgevertexset} If $q_{\ell}^0$ and $q_{\ell}^1$ are non-isolated for $e_\ell  = \{v_i,v_j\}\in E(\hat{G})$, then both~$q_{\ell}^0$ and~$q_{\ell}^1$ are adjacent in~$\sigma$ to vertex-agents corresponding to $v_i$ or $v_j$. 
      In particular, $q_{\ell}^0$ and~$q_{\ell}^1$ are adjacent to $p_i^{sh+1}$ or $p_j^{sh+1}$ with $s \in \{0\} \cup [h-2]$, where the next $h$ seats on the path (resp.\ cycle) are assigned to $p_i^{sh+1}, \ldots, p_i^{(s+1)h}$ or $p_j^{sh+1}, \ldots, p_j^{(s+1)h}$. 
    \end{compactenum}
  }

\statementarxiv{theorem}{thm:EFA_k_clique+stars+path-cycle_symm}{\thmefakhard}

\ifarxiv
  \begin{proof} 
  	\textbf{Path- and cycle-graphs. }
  	We provide
\else
  \begin{proof}[Proof sketch]
    We only show the case with path- and cycle-graphs via
\fi 		
    a parameterized reduction from
    \ifarxiv the \Wh\ problem \Clique, parameterized by the solution size~\cite{DF13}.
    \else
    \Clique, parameterized by the solution size~\cite{DF13}. 
    \fi
    Let $(\hat{G},h)$ be an instance of \Clique. 
    For each \mbox{$v_i \in V(\hat{G})$,} create~$h(h-1)$ \myemph{vertex-agents} $p_i^1, \ldots, p_i^{h(h-1)}$. 
    For each $e_{\ell} \in E(\hat{G})$, create $h^4+1$ \myemph{edge-agents} $q_{\ell}^0, q_\ell^1, \ldots, q_\ell^{h^4}$.
    
    Since the preferences will be symmetric, we only specify one value for each pair of agents (see \cref{fig:EFA_k_path-cycle_symm} for the corresponding preference graph).
    For each~$(e_{\ell},v_i)\in E(\hat{G}) \times V(\hat{G})$, we do the following:
    Set $\sat[p_i^{sh+z}](p_i^{sh+z+1})$ = $z$ for each $(z,s)\in [h-1]\times \{0, \ldots, h-2\}$, and if $s\neq h-2$, then set $\sat[p_i^{sh+2}](p_i^{(s+1)h+2}) = 1$. %
    For each $z \in \{0\}\cup [h^4-1]$, set $\sat[q_{\ell}^{z}](q_{\ell}^{z+1}) = 1$ and $\sat[q_{\ell}^{0}](q_{\ell}^{2}) = \sat[q_{\ell}^{3}](q_{\ell}^{h^4}) = 1$.
    For each~$s\in \{0\}\cup\! [h-2]$, set $\sat[q_{\ell}^{2}](p_i^{sh+2})\!=\!-1$ if $v_i\in e_{\ell}$. %
    The non-mentioned preferences are set to zero.
	\begin{figure}[t]
		\centering
		\begin{tikzpicture}[>=stealth', shorten <= 2pt, shorten >= 2pt]
			\hspace{-1.5mm}
			\def \xs {4ex}
			\def \ys {8ex}

			\node[nodeW] (q2) {};
			\node[nodeW, below left = 0.8*\ys and 0.5*\xs of q2] (q0) {};
			\node[nodeW, below right = 0.8*\ys and 0.5*\xs of q2] (q1) {};
			\node[nodeW, above = 0.65*\ys of q2] (q3) {};
			\node[nodeW, above right = 0.3*\ys and 1.3*\xs of q3] (q4) {};
			\node[nodeW, above left = 0.3*\ys and 1.3*\xs of q3] (q5) {};
			
			\path (q4) -- node[auto=false]{\ldots\ldots} (q5);
			
			\foreach \i / \pos / \v in {0/below/0, 1/below/1,2/{above right}/2, 3/{right}/3, 4/right/4, 5/left/h^4}{
				\node[\pos = -1pt of q\i] {\scriptsize $q_\ell^{\v}$};
			}
			
			\foreach \i / \pos / \l / \r in {i/left/h-1/2, j/right/2/h-1}{
				\node[nodeU, above \pos = 0.65*\ys and 1.7*\xs of q2] (p\i0h1) {};
				\node[below = 0pt of p\i0h1] {\scriptsize$p_\i^1$};
				
				\node[nodeU, below = 0.8*\ys of p\i0h1] (p\i1h1) {};
				\node[below = 0pt of p\i1h1] {\scriptsize$p_\i^{h+1}$};
				
				\node[nodeU, below = 0.9*\ys of p\i1h1] (p\i2h1) {};
				\node[below = 0pt of p\i2h1] {\scriptsize$p_\i^{(h-2)h+1}$};
				
				\foreach \x / \n / \m in {0h/2/h, 1h/h+2/2h, 2h/(h-2)h+2/h(h-1)}{
					\node[nodeU, \pos = 1.2*\xs of p\i\x1] (p\i\x2) {};

					\node[nodeU, \pos = 2.4*\xs of p\i\x2] (p\i\x3) {};
					\node[below = 1pt of p\i\x3] {\scriptsize$p_\i^{\m}$};
					
					\path (p\i\x2) -- node[auto=false]{\color{gray}$\overset{\l, \ldots, \r}{\ldots\ldots\ldots}$} (p\i\x3); 
				}
				\node[below \pos = 1pt and -3pt of p\i0h2] {\scriptsize$p_\i^{2}$};
				\node[below = 1pt of p\i1h2] {\scriptsize$p_\i^{h+2}$};
				\node[below \pos = 1pt and -12pt of p\i2h2] {\scriptsize$p_\i^{(h-2)h+2}$};
			
				\path (p\i1h2) -- node[pos=0.55, auto=false]{\vdots} (p\i2h2);
			}

			\begin{scriptsize}
			\draw[Stealth-Stealth] (q0) edge[bend left = 35]  node[pos=0.5, fill=white, inner sep=1pt] {$1$} (q1); 
			\foreach \x / \y / \b in {1/2/0, 2/3/0, 3/4/10, 3/5/-10, 0/2/0} {
				\draw[Stealth-Stealth] (q\x) edge[bend left=\b] node[pos=0.5, fill=white, inner sep=1pt] {$1$} (q\y); 
			}
			
			\foreach \i / \b in {i/right, j/left}{					
				\draw[Stealth-Stealth] (q2) edge[bend \b=50, myRed] node[pos=0.5, fill=white, inner sep=1pt] {$-1$} (p\i0h2);
				\draw[Stealth-Stealth] (q2) edge[bend \b=17, myRed] node[pos=0.5, fill=white, inner sep=1pt] {$-1$} (p\i1h2);
				\draw[Stealth-Stealth] (q2) edge[bend \b=-10, myRed] node[pos=0.60, fill=white, inner sep=1pt] {$-1$} (p\i2h2);
			}
			
			\foreach \i in {i, j}{					
				\foreach \x in {0h,1h,2h}{	
					\draw[Stealth-Stealth] (p\i\x1) edge node[pos=0.5, fill=white, inner sep=1pt] {$1$} (p\i\x2);
				}
				\draw[Stealth-Stealth] (p\i0h2) edge node[pos=0.5, fill=white, inner sep=1pt] {$1$} (p\i1h2);
			}
			\end{scriptsize}
		\end{tikzpicture}
		\caption{Preference graph for \cref{thm:EFA_k_clique+stars+path-cycle_symm} for path-\,or\,cycle-graphs, where $e_{\ell}=\{v_i, v_j\}$. 
		The vertex-agents sets $\{p_i^1, \ldots, p_i^{h(h-1)}\}$ and $\{p_j^1, \ldots, p_j^{h(h-1)}\}$ are all-or-nothing gadgets, see \cref{claim:EFA_k_clique+stars+path-cycle_symm}\eqref{efasymk-allornothing}.}
		\label{fig:EFA_k_path-cycle_symm}
	\end{figure}
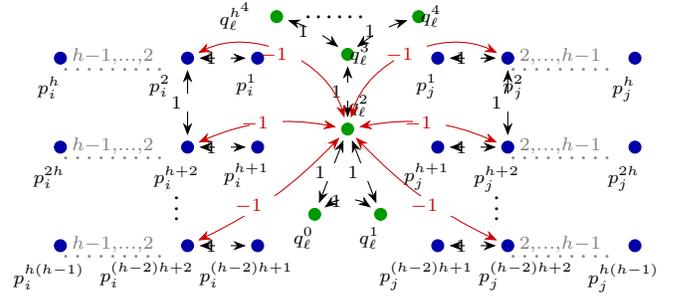
    The seat graph~$G$ consists of a path (resp.\ cycle) with $\nonisolated\coloneqq h^2(h-1)+h(h-1)$ vertices and $|V(\hat{G})|h(h-1)+|E(\hat{G})|(h^4+1)-\nonisolated$ isolated vertices.
    
    It remains to show that $\hat{G}$ admits a size-$h$ clique if and only if~$I$ admits an \efArr.	For the ``only if'' part, let~$C$ be a size-$h$ clique. 
    For the path we begin assigning at one of the endpoints; for the cycle we can begin at any non-isolated vertex.
    For each edge $e_\ell = \{v_i,v_j\} \in C$ and some (not yet used) $z,z' \in [h-1]$, assign agents $p_i^{zh}, \ldots, p_i^{(z-1)h+1}, q_{\ell}^0, q_\ell^1, p_j^{(z'-1)h+1}, \ldots, p_j^{z'h}$ in this order to the path (resp.\ cycle). 
    Since each vertex $v_i \in C$ is incident to exactly $h-1$ edges in~$C$,
    we can find such $z,z' \in [h-1]$. 
    The remaining agents are assigned to isolated vertices.			
    This arrangement is envy-free because: %
    \begin{inparaenum}[(i)]
      \item Every vertex-agent $p_i^z$ with $v_i \in C, z \in [h(h-1)]$ has his maximum possible utility.
      \item Vertex-agents $p_i^z$ with $v_i \notin C, z \in [h(h-1)]$ are envy-free since there is no non-isolated agent towards which~$p_i^z$ has a positive preference.
      \item All edge-agents $q_\ell^0, q_{\ell}^1$ with $e_\ell \in C$ are envy-free since they have maximum possible utility one.
      \item The remaining edge-agents assigned to isolated vertices either have no non-isolated agent towards which they have a positive preference, or the corresponding agents $q_{\ell}^0, q_{\ell}^1$ are assigned in such a way, that $q_{\ell}^2$ does not envy his neighbors. 
    \end{inparaenum}
    
    Before we prove the ``if'' part, we observe the following.
    \begin{claim}[\appendixsymb]\label{claim:EFA_k_clique+stars+path-cycle_symm}
      \clmefaksymm
    \end{claim}
    \appendixproofwithstatement{claim:EFA_k_clique+stars+path-cycle_symm}{\clmefaksymm}
    {\begin{proof}[Proof of \cref{claim:EFA_k_clique+stars+path-cycle_symm}]
        \renewcommand{\qedsymbol}{$\diamond$}
        \begin{compactenum}[(i)]
          \item Suppose there is an agent $q_{\ell}^z$ for $e_\ell \in E(\hat{G})$, $z \in [h^4]\setminus\{1\}$assigned to the path (resp.\ cycle). 
          Then, by \cref{lem:EFA_non-neg_cond} all agents $q_{\ell}^2, \ldots, q_{\ell}^{h^4}$ have to be assigned to the path (resp.\ cycle).
          However, since $\nonisolated < h^4-1$ for $h>1$, there are not enough vertices in the path (resp.\ cycle).
          
          \item By the previous statement, we know that for each edge $e_\ell \in E(\hat{G})$, agent~$q_\ell^2$ is isolated. 
          Hence, the first part of this statement follows from \cref{lem:EFA_non-neg_cond}.
          If the set of agents for a vertex $v_i \in V(\hat{G})$ is non-isolated, then the seats of $p_i^{zh+s}, p_i^{zh+s+1}$ for $z\in \{0\}\cup [h-2]$, $s\in [h-1]$ are adjacent because $\sat[p_i^{zh+s}](p_i^{zh+s+1})$ is the largest positive preference value of $p_i^{zh+s}$. 
          Hence, if~$p_i^{zh+s}$ is not assigned next to $p_i^{zh+s+1}$, he will have an envy.
          
          \item This follows from \cref{lem:EFA_non-neg_cond}. 
          
          \item By the previous statement we know that agents $q_{\ell}^0$ and $q_{\ell}^1$ are both either isolated or non-isolated and in the latter case, their seats are adjacent. 
          If $q_{\ell}^0$ (resp.\ $q_{\ell}^1$) is not adjacent to a vertex-agent, then he has to be adjacent to an edge-agent, which implies that $q_\ell^2$ envies $q_{\ell}^0$ (resp.\ $q_{\ell}^1$). 
          To make $q_\ell^2$ envy-free, the next agent but one from $q_{\ell}^0$ (resp.\ $q_{\ell}^1$) has to be some $p_i^{zh+2}$ or $p_j^{zh+2}$. 
          Combining this with the second statement of this claim proves this statement. \qedhere
        \end{compactenum}
      \end{proof}
    }
    \noindent By \cref{claim:EFA_k_clique+stars+path-cycle_symm}\eqref{efasymk-allornothing}, at most $h$ different groups of vertex-agents can be assigned to the path (resp.\ cycle), i.e., non-isolated.
    \newH{By our bound on~$\nonisolated$, at least $2\binom{h}{2}$ many edge-agents are non-isolated. 
    By \cref{claim:EFA_k_clique+stars+path-cycle_symm}\eqref{efasymk-edgepair}, at least $\binom{h}{2}$ many pairs of edge-agents of the form $\{q_{\ell}^0, q_{\ell}^1\}$ are non-isolated.
    Let $E'\subseteq E(\hat{G})$ denote the set of edges that correspond to the non-isolated edge-agents. 
    By \cref{claim:EFA_k_clique+stars+path-cycle_symm}\eqref{efasymk-edgevertexset}, all vertex-agents that are ``incident'' to the edges in~$E'$ must also be non-isolated.
    Since only $h$ different groups of vertex-agents can be non-isolated, this corresponds to a clique of size $h$ in $\hat{G}$.
  }
    \appendixcontinue{thm:EFA_k_clique+stars+path-cycle_symm}{}{\thmefakhard}{
    \iflong
      We continue with the proof.
    \fi 
      
      \paragraph*{Clique-graphs.} %
      We provide a parameterized reduction from the \Wh\ problem \IS parameterized by the solution size. 
      Let $(\hat{G}, h)$ be an instance of \IS. 		
      
      For each vertex $v_i \in V(\hat{G})$, we create one agent named $p_i$ and set $\sat[p_i](p_j)=\sat[p_j](p_i) = -1$ if $\{v_i,v_j\} \in E(\hat{G})$; otherwise the preference values are set to zero.
      The seat graph~$G$ consists of a clique of size $\nonisolated\coloneqq h$ and $|\hat{V}|-h$ isolated vertices.
      It is straightforward to verify that~$\hat{G}$ has a size-$h$ independent set if and only if there is an \efArr~$\sigma$.
      
      \paragraph*{Stars-graphs.}
      We also reduce from \IS and use the fact that \IS remains \Wh, parameterized by the solution size $h$ even on $r$-regular graphs~\cite{Cai2008}, where $r > h$. %
      For each vertex $v_i \in V(\hat{G})$ (resp.\ each edge $e_{\ell} \in E(\hat{G})$), create one vertex-agent $p_i$ (resp.\ edge-agent $q_{\ell}$). 
      Finally, create a special agent~$x$ which can only be assigned to the star center.
      Summarizing, the agent set is $\agents=\{x\}\cup \{p_i\mid v_i\in V(\hat{G})\}\cup \{q_{\ell}\mid e_{\ell}\in E(\hat{G})\}$.
      
      Now, we define the symmetric preferences. %
      For each edge $e_{\ell} \in E(\hat{G})$ and each agent~$u \in 
      \agents\setminus\{q_\ell\}$, we set $\sat[q_{\ell}](u) = \sat[u](q_{\ell}) = h-1$ if $u$ corresponds to the endpoint of $e_{\ell}$, %
      and $\sat[q_{\ell}](u) = \sat[u](q_{\ell}) = -1$ otherwise. 
      The non-mentioned preferences are set to zero.		
      The seat graph consists of a single star with~$h$ leaves, i.e., $\nonisolated\coloneqq h+1$, and $|V(\hat{G})|+|E(\hat{G})|-k$ isolated vertices. 
      
      Before we continue with the correctness, we observe that in every \efArr the agents assigned to the leaves of the star have a non-negative preference towards the center-agent. 
      Otherwise they envy every isolated agent.
      Since for each edge-agent $q_{\ell}$ there are only two other agents with non-negative preferences towards $q_{\ell}$, no edge-agent can be assigned to the star center in an \efArr. 
      Moreover,  since each vertex in $\hat{G}$ has more than~$h$ neighbors, no vertex-agent can be assigned to the center in an \efArr as there are only $h$ leaves.
      Hence, in every \efArr agent~$x$ has to be assigned to the center. 
      Furthermore, only vertex-agents can be assigned to leaves because $\sat[q_{\ell}](x)<0$ for each $e_{\ell} \in E(\hat{G})$.
      
      This brings us to the following: 
      There exists an \efArr if and only if each edge-agent has at most one incident vertex-agent in the star's leaves, which is equivalent to $\hat{G}$ having a size-$h$ independent set. 
      \iflong \hfill \qed \fi }	
  \end{proof}

  \newcommand{\efadeltabinnphard}{%
    For each considered graph class, \EFA remains \NPh\ for binary preferences with $\maxoutdeg = 3$. %
  }
  
  The next two theorems show that achieving envy-freeness remains hard for bounded maximum out-degree and binary (resp.\ symmetric) preferences.
  The proofs are similar to the ones for \cref{thm:EFA_k_matching+clique+path-cycle_bin,thm:EFA_k_clique+stars+path-cycle_symm}.
  
  \statementarxiv{theorem}{thm:EFA_delta_matching+clique+path-cycle_bin}{\efadeltabinnphard}
  \appendixproofwithstatement{thm:EFA_delta_matching+clique+path-cycle_bin}{\efadeltabinnphard}{
    \begin{proof}
      We provide a polynomial-time reduction from \Clique for each graph class of the seat graph.
      Let $(\hat{G},h)$ be an instance of \Clique and $\hat{n}\coloneqq|V(\hat{G})|$. 
      We construct an instance~$I$ of \EFA in the following way: 
      \begin{compactitem}[--]
        \item For each vertex $v_i \in V(\hat{G})$, create $\hat{n}^2$ \myemph{vertex-agents} $p_i^1, \ldots, p_i^{\hat{n}^2}$ and set $\sat[p_i^z](p_i^{z+1})=\sat[p_i^{z+1}](p_i^z)=1$ for $z \in [\hat{n}^2-1]$.		

        \item For each edge $e_\ell = \{v_i,v_j\} \in E(\hat{G})$, create one \myemph{set-agent} $q_{\ell}$ and set $\sat[p_i^j](q_{\ell}) = \sat[p_j^i](q_{\ell}) = 1$. 
        
        \item The non-mentioned preferences are set to zero. 
      \end{compactitem}
      \cref{fig:EFA_delta_path-cycle_bin} shows this construction for two adjacent vertices $v_i, v_j \in V(\hat{G})$. 		
  		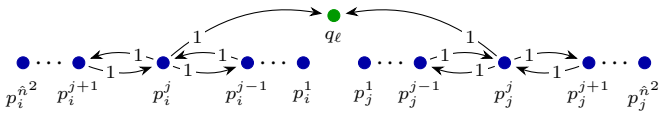
\begin{figure}[t]
  			\centering
  			\begin{tikzpicture}[>=stealth', shorten <= 1.5pt, shorten >= 2pt]
  			\def \xs {6ex}
  			\def \ys {3.2ex}

  			\node[nodeW] (q) {};
  			\node[below = 0pt of q] {\scriptsize$q_\ell$};
  			
  			\foreach \x / \pos in {i/left,j/right} {
  				\node[nodeU, below \pos = \ys and 0.3*\xs of q] (p\x1) {};
  				\node[below = 1pt of p\x1] {\scriptsize$p_\x^1$};
  				
  				\node[nodeU, \pos = 0.6*\xs of p\x1] (p\x2) {};
  				\node[below = 1pt of p\x2] {\scriptsize$p_\x^{j-1}$};
  				
  				\node[nodeU, \pos = \xs of p\x2] (p\x3) {};
  				\node[below = 1pt of p\x3] {\scriptsize$p_\x^{j}$};
  				
  				\node[nodeU, \pos = \xs of p\x3] (p\x4) {};
  				\node[below = 1pt of p\x4] {\scriptsize$p_\x^{j+1}$};
  				
  				\node[nodeU, \pos = 0.6*\xs of p\x4] (p\x5) {};
  				\node[below = 1pt of p\x5] {\scriptsize$p_\x^{\hat{n}^2}$};
  				
  				\path (p\x1) -- node[auto=false]{\ldots} (p\x2);
  				\path (p\x4) -- node[auto=false]{\ldots} (p\x5);
  			}

  			\begin{scriptsize}
  			\foreach \x / \pos in {i/left,j/right} {
  				\draw[-Stealth] (p\x3) edge[bend \pos=25] node[pos=0.2, fill=white, inner sep=1pt] {$1$} (q);
  				
  				\draw[-Stealth] (p\x2) edge[bend \pos=-20] node[pos=0.25, fill=white, inner sep=1pt] {$1$} (p\x3); 
  				\draw[-Stealth] (p\x3) edge[bend \pos=-20] node[pos=0.25, fill=white, inner sep=1pt] {$1$} (p\x2); 
  				
  				\draw[-Stealth] (p\x3) edge[bend \pos=-20] node[pos=0.25, fill=white, inner sep=1pt] {$1$} (p\x4); 
  				\draw[-Stealth] (p\x4) edge[bend \pos=-20] node[pos=0.35, fill=white, inner sep=1pt] {$1$} (p\x3); 
  			}
  			\end{scriptsize}
  			\end{tikzpicture}
  			\caption{Preference graph for \cref{thm:EFA_delta_matching+clique+path-cycle_bin} for each considered seat graph class, where $e_{\ell}=\{v_i, v_j\}$.}
  			\label{fig:EFA_delta_path-cycle_bin}
  		\end{figure}
      The seat graph~$G$ consists of a matching (i.e., stars with one leaf), clique, path resp.\ cycle with $\nonisolated \coloneqq h\hat{n}^2 + \binom{h}{2}$ vertices and $\hat{n}^3+|E(\hat{G})|-k$ isolated vertices. %
      Clearly, the preferences are binary and this construction of $I$ can be done in polynomial time. 
      Moreover, since between every pair of vertices in~$\hat{G}$ there is at most one incident edge, we obtain $\maxoutdeg \leq 3$. 
      
      It remains to show the correctness, i.e., $\hat{G}$ admits a size-$h$ clique $C$ if and only if $I$ admits an \efArr. 
      For the ``only if'' part, let $C$ be a size-$h$ clique of $\hat{G}$. 
      For each vertex $v_i \in C$, assign agents $p_i^1, \ldots, p_i^{\hat{n}^2}$ to non-isolated vertices in the following way:
      \begin{compactitem}[--]
        \item For a clique-graph, we can assign them arbitrarily. 
        \item For a matching-graph, assign for each $z \in 			[\hat{n}^2/2]$ the two agents~$p_i^{2z-1}$ and~$p_i^{2z}$ to the endpoints of the same edge.
        \item For a path- and cycle-graph assign the agents consecutively on the path (resp.\ cycle), i.e., $p_i^{z+1}$ is assigned next to~$p_i^{z}$ for $z \in [\hat{n}^2-1]$.
      \end{compactitem}
      After this step $h\hat{n}^2$ non-isolated vertices are assigned. 
      Moreover, for every edge $e_\ell \in C$ we assign $q_{\ell}$ to an arbitrary remaining non-isolated vertex in $G$. 
      All remaining agents are assigned to isolated vertices.			
      This arrangement $\sigma$ is envy-free because of the following:
      \begin{compactenum}[(i)]
        \item Every edge-agent $q_{\ell}$ is envy-free since $\sat[q_{\ell}](p) = 0$ for each $p \in \agents\setminus\{q_{\ell}\}$. 
        
        \item Every vertex-agent $p_i^z$ with $v_i \in C$, $z \in [\hat{n}^2]$, is envy-free because $p_i^z$ has his maximum possible utility. 
        
        \item The remaining vertex-agents $p_i^z$ with $v_i \notin C$ and $z \in [\hat{n}^2]$, have zero utility. 
        However, there is also no non-isolated agent towards which $p_i^z$ has a positive preference. 
        Therefore, also all agents corresponding to vertices in $V(\hat{G})\setminus C$ are envy-free.
      \end{compactenum}
      
      For the ``if'' part, let~$\sigma$ be an \efArr of~$I$. 
      By \cref{lem:EFA_non-neg_cond} we can observe that for each $v_i \in V(\hat{G})$, agents $p_i^1, \ldots, p_i^{\hat{n}^2}$ are all assigned to either isolated or to non-isolated vertices.
      Combining this with $h\hat{n}^2 + \binom{h}{2} < h\hat{n}^2 + 
      \hat{n}^2$, 
      we can conclude that at most $h$ different sets of vertices are assigned to non-isolated vertices. 
      
      Moreover, if agent~$q_{\ell}$ for an edge $e_\ell = \{v_i, v_j\} \in E(\hat{G})$ is non-isolated, then again by \cref{lem:EFA_non-neg_cond}, $p_i^j$ and~$p_j^i$ (and all created agents of $v_i$ and~$v_j$) are assigned to non-isolated vertices. 
      Since $\binom{h}{2}+1$ edges are incident to at least $h+1$ vertices, 
      we know that at most $\binom{h}{2}$ edge-agents can be non-isolated. 
      Therefore, exactly~$h$ different sets of vertex-agents and $\binom{h}{2}$ edge-agents are assigned to non-isolated vertices, which have to form a clique in~$\hat{G}$. 
    \end{proof}
  }

  \looseness=-1
  \newcommand{\efadeltasymmnphard}{%
    \EFA remains \NPh\ even for a clique-, \mbox{path-,} \mbox{cycle-,} or stars-graph and symmetric preferences with constant~$\maxoutdeg$. %
  }

\statementarxiv{theorem}{thm:EFA_delta_clique+path-cycle+stars_symm}{\efadeltasymmnphard}
  \appendixproofwithstatement{thm:EFA_delta_clique+path-cycle+stars_symm}{\efadeltasymmnphard}
  {
    \begin{proof}
      \textbf{Clique-graphs.} 
      We provide a polynomial-time reduction from \IS, which is \NPh\ even on cubic graphs~\cite{garey1979}. 
      Given an instance $(\hat{G},h)$ of \IS, we create for each vertex $v_i 
      \in V(\hat{G})$ one agent named~$p_i$. 
      We set $\sat[p_i](p_j)=\sat[p_j](p_i)=-1$ if and only if $\{v_i,v_j\} \in E(\hat{G})$. 
      Otherwise, the preference value is set to zero.
      The seat graph~$G$ consists of a clique of size $\nonisolated \coloneqq h$ and $|V(\hat{G})|-h$ isolated vertices. 
      Clearly, the preferences are symmetric and since $\hat{G}$ is cubic, it holds $\maxoutdeg = 3$.
      
      Regarding the correctness we observe that an arrangement of $I$ is 
      envy-free 
      if and only if each agent has zero utility, 
      which is equivalent to $\hat{G}$ having a size-$h$ independent set.

      \paragraph*{Path- and cycle-graphs.} 
      We provide a polynomial-time reduction from the \NPc\ problem \Clique. 
      Let $(\hat{G},h)$ be an instance of \Clique\ and 
      $\hat{n}\coloneqq|V(\hat{G})|$. 
      For each vertex $v_i \in V(\hat{G})$, we create~$h\hat{n}$ \myemph{vertex-agents} named
      $p_i^1, \ldots, p_i^{h\hat{n}}$.
      For each edge $e_\ell \in E(\hat{G})$, we create~$\hat{n}^4+1$ \myemph{edge-agents} named $q_\ell^0,q_\ell^1, \ldots, q_\ell^{\hat{n}^4}$. 
      
      Since we will have symmetric preferences, for each pair of agents we 
      only specify one value (see \cref{fig:EFA_delta_path-cycle_symm} 
      for the corresponding preference graph). 
      For each edge $e_\ell \in E(\hat{G})$, each vertex $v_i \in V 
      (\hat{G})$, do the following:
      \begin{compactitem}[--]
        \item For each $z \in [h-1]$ and $s \in \{0\}\cup [\hat{n}-1]$, set 
        $\sat[p_i^{sh+z}](p_i^{sh+z+1}) = 1$.
        
        \item For each $s \in [\hat{n}-1]$, set $\sat[p_i^{(s-1)h+2}](p_i^{sh+2}) =  1$. 
        
        \item For each $z \in [\hat{n}^4]$, set $\sat[q_{\ell}^{z}](q_{\ell}^{z-1}) = 1$,  $\sat[q_{\ell}^{0}](q_{\ell}^{2}) = 1$ and $\sat[q_{\ell}^{3}](q_{\ell}^{\hat{n}^4}) = 1$.
        
        \item Set $\sat[q_{\ell}^{2}](p_i^{(j-1)h+2}) = -1$ if $v_i \in e_\ell$.
        
        \item The non-mentioned preferences are set to zero.
      \end{compactitem}
  		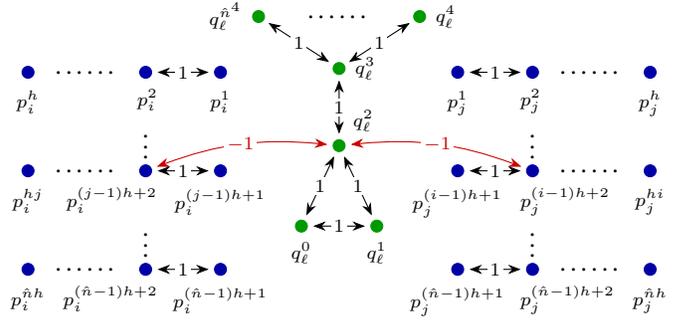
\begin{figure}[t]
  			\centering
  			\begin{tikzpicture}[>=stealth', shorten <= 2pt, shorten >= 2pt]
  			\def \xs {4ex}
  			\def \ys {6ex}

  			\node[nodeW] (q2) {};
  			\node[nodeW, below left = \ys and 0.6*\xs of q2] (q0) {};
  			\node[nodeW, below right = \ys and 0.6*\xs of q2] (q1) {};
  			\node[nodeW, above = 0.9*\ys of q2] (q3) {};
  			\node[nodeW, above right = 0.6*\ys and 1.5*\xs of q3] (q4) {};
  			\node[nodeW, above left = 0.6*\ys and 1.5*\xs of q3] (q5) {};
  			
  			\path (q4) -- node[auto=false]{\ldots\ldots} (q5);
  			
  			\foreach \i / \pos / \v in {0/below/0, 1/below/1, 3/right/3, 4/right/4, 5/left/\hat{n}^4}{
  				\node[\pos = 0pt of q\i] {\scriptsize$q_\ell^{\v}$};
  			}
  			\node[above right = 0pt of q2] {\scriptsize$q_\ell^{2}$};
  			
  			\foreach \i /\j / \pos in {i/j/left, j/i/right}{
  				\node[nodeU, above \pos = 0.9*\ys and 2.3*\xs of q2] (p\i0h1) {};
  				\node[below = 1pt of p\i0h1] {\scriptsize$p_\i^1$};
  				
  				\node[nodeU, below = 1.2*\ys of p\i0h1] (p\i1h1) {};
  				\node[below = 1pt of p\i1h1] {\scriptsize$p_\i^{(\j-1)h+1}$};
  				
  				\node[nodeU, below = 1.2*\ys of p\i1h1] (p\i2h1) {};
  				\node[below = 1pt of p\i2h1] {\scriptsize$p_\i^{(\hat{n}-1)h+1}$};
  				
  				\foreach \x / \n / \m in {0h/2/h, 1h/(\j-1)h+2/h\j, 2h/(\hat{n}-1)h+2/\hat{n}h}{
  					\node[nodeU, \pos = 1.3*\xs of p\i\x1] (p\i\x2) {};
  					\node[below \pos = 1pt and -10pt of p\i\x2] {\scriptsize$p_\i^{\n}$};
  					
  					\node[nodeU, \pos = 2.2*\xs of p\i\x2] (p\i\x3) {};
  					\node[below = 1pt of p\i\x3] {\scriptsize$p_\i^{\m}$};
  					
  					\path (p\i\x2) -- node[auto=false]{\ldots\ldots} (p\i\x3); 
  				}
  				\path (p\i1h2) -- node[pos=0.7, auto=false]{\vdots} (p\i2h2);
  				\path (p\i0h2) -- node[pos=0.7, auto=false]{\vdots} (p\i1h2);
  			}

  			\begin{scriptsize}
  			\foreach \x / \y in {0/1, 1/2, 2/3, 3/4, 3/5, 0/2} {
  				\draw[Stealth-Stealth] (q\x) edge  node[pos=0.5, fill=white, inner sep=1pt] {$1$} (q\y); 
  			}
  			
  			\foreach \i / \b in {i/right, j/left}{ 
  				\draw[Stealth-Stealth] (q2) edge[bend \b=13, myRed] node[pos=0.5, fill=white, inner sep=1pt] {$-1$} (p\i1h2);
  			}
  			
  			\foreach \i in {i, j}{					
  				\foreach \x in {0h,1h,2h}{	
  					\draw[Stealth-Stealth] (p\i\x1) edge node[pos=0.5, fill=white, inner sep=1pt] {$1$} (p\i\x2);
  				}
  			}
  			\end{scriptsize}
  			\end{tikzpicture}
  			\caption{Preference graph for \cref{thm:EFA_delta_clique+path-cycle+stars_symm} for path- and cycle-graphs, where $e_{\ell}=\{v_i, v_j\}$.}
  			\label{fig:EFA_delta_path-cycle_symm}
  		\end{figure}
      
      The seat graph consists of a path (resp.\ cycle) with $\nonisolated\coloneqq \hat{n}h^2+h(h-1)$ vertices and $h\hat{n}^2+ |E(\hat{G})|(\hat{n}+1)-k$ isolated vertices.
      Clearly, this instance $I$ of \EFA can be constructed in polynomial time. 
      Since between every pair of vertices in~$\hat{G}$ there is at most one edge, it holds $\maxoutdeg = 5$.
      
      It remains to show that $\hat{G}$ admits a size-$h$ clique if and only 
      if $I$ admits an \efArr. 		
      We start proving the ``only if'' part. 
      Let $C$ be a size-$h$ clique. 
      For the path we begin assigning at one of the endpoints, for the cycle 
      we can begin at any non-isolated vertex.
      For each edge $e_\ell = \{v_i,v_j\} \in C$ assign agents $p_i^{jh}, \ldots, p_i^{(j-1)h+1}, q_{\ell}^0, q_\ell^1, p_j^{(i-1)h+1}, \ldots, p_j^{ih}$ in this order to the path (resp.\ cycle). 
      Next, for each vertex $v_i \in V(\hat{G})$ and $z \in [n]\setminus \{j| 
      v_j \in C\}$ assign $p_i^{(z-1)h+1}, \ldots, p_i^{zh}$ in this order 
      consecutively to the path (resp.\ cycle). 
      The remaining agents are assigned to isolated vertices.	
      This arrangement is envy-free because of the following:
      \begin{compactenum}[(i)]
        \item The vertex-agents $p_i^{zh+s}$ for $v_i \in 
        C, z \in \{0\}\cup[\hat{n}-1]$, $s \in [h]$ have their 
        maximum possible utility. 
        Thus, they are envy-free.
        
        \item A vertex-agent $p_i^z, z \in [\hat{n}h]$ with $v_i 
        \notin C$ is envy-free since there is no non-isolated agent 
        towards which they have a positive preference.
        
        \item For each edge $e_\ell \in C$, agents$q_{\ell}^0, q_{\ell}^1$ are envy-free since they have their maximum possible utility. 
        
        \item The remaining edge-agents assigned to isolated vertices either have no non-isolated agent towards which they have a positive preference, or the corresponding agents $q_{\ell}^0, q_{\ell}^1$ are assigned in such a way, that $q_{\ell}^2$ does not envy his neighbors. 
      \end{compactenum}	
      
      Before we prove the ``if'' part, we observe the following properties an \efArr has to satisfy.
      \begin{claim}\label{claim:EFA_delta_clique+stars+path-cycle_symm}
        Every envy-free arrangement~$\sigma$ satisfies:
        \begin{compactenum}
          \item For each edge $e_\ell \in E(\hat{G})$ agents 
          $q_{\ell}^2, \ldots, q_{\ell}^{\hat{n}^4}$ are assigned to isolated vertices.
          
          \item\label{efa-delta-symm:clique+stars+path-cycle-vertex-gadget} If a vertex-agent $p_i^z$ with $v_i \in V(\hat{G}), z \in [\hat{n}h]$ is assigned to the path (resp.\ cycle), then all agents from $\{p_i^1, \ldots, p_i^{\hat{n}h}\}$ are assigned to the path (resp.\ cycle).
          
          \item If $q_{\ell}^0$ or $q_{\ell}^1$ is non-isolated for some $e_\ell \in E(\hat{G})$, then both agents $q_{\ell}^0, q_{\ell}^1$ are non-isolated and their seats are adjacent.

          \item If $q_{\ell}^0$ or $q_{\ell}^1$ is non-isolated 
          for some 
          $e_\ell = \{v_i, v_j\} \in E(\hat{G})$, then all agents from $\{p_i^1, \ldots, 
          p_i^{h(h-1)}\}$ 
          and from $\{p_j^1, \ldots, p_j^{h(h-1)}\}$ are non-isolated. 
        \end{compactenum}
      \end{claim}
      \begin{proof}[Proof of 
        \cref{claim:EFA_delta_clique+stars+path-cycle_symm}]
        \renewcommand{\qedsymbol}{$\diamond$}
        \begin{compactenum}
          \item Suppose an agent $q_{\ell}^z$ for some $e_\ell \in E(\hat{G}), z \in [n^4]\setminus\{1\}$ is assigned to the path (resp.\ cycle). 
          Then, by \cref{lem:EFA_non-neg_cond} all agents $q_{\ell}^2, \ldots, q_{\ell}^{\hat{n}^4}$ have to be assigned to the path (resp.\ cycle).
          However, since $\nonisolated < \hat{n}^4-1$ for $\hat{n}>1$, there are not enough vertices in the path (resp.\ cycle).
          
          \item By the previous statement we know for each edge $e_\ell \in E(\hat{G})$ that agent $q_\ell^2$ is isolated. 
          Hence, this follows now from \cref{lem:EFA_non-neg_cond}.
          
          \item This follows from \cref{lem:EFA_non-neg_cond}. 
          
          \item By the previous statement we know that agents $q_{\ell}^0$ and $q_{\ell}^1$ are both either isolated or non-isolated, and in the latter case, their seats are adjacent. 
          If $q_{\ell}^0$ (resp.\ $q_{\ell}^1$) is not adjacent to a vertex-agent, then he has to be adjacent to an edge-agent, which implies that $q_\ell^2$ envies $q_{\ell}^0$ (resp.\ $q_{\ell}^1$). 
          To make $q_\ell^2$ envy-free, the next agent but one from 
          $q_{\ell}^0$ (resp.\ $q_{\ell}^1$) has to be $p_i^{(j-1)h+2}$ 
          or $p_j^{(i-1)h+2}$ for $e_\ell = \{v_i, v_j\}$. 
          Combining this with Statement~\eqref{efa-delta-symm:clique+stars+path-cycle-vertex-gadget} proves this statement. \qedhere
        \end{compactenum}
      \end{proof}
      From the second statement of this claim it follows that at most $h$ different sets of vertex-agents can be assigned to the path (resp.\ cycle). 
      
      Moreover, for each pair of non-isolated edge-agents $q_{\ell}^0, q_{\ell}^1$ the two vertex-agents $p_i^{(j-1)h+2}, p_j^{(i-1)h+2}$ need to be assigned to the two sides of $q_{\ell}^0, q_{\ell}^1$.
      Since  $\binom{h}{2} + x$ edges are incident to at least $h+1$ vertices, we can assign at most $\binom{h}{2}$ edge-agents to the path (resp.\ cycle). 
      Hence, we have to assign exactly $h$ different sets of vertex-agents to the 
      non-isolated vertices, which have to be incident to the $\binom{h}{2}$ 
      non-isolated edge-agents.
      This means the corresponding vertices and edges in $\hat{G}$ form a 
      clique. 
      
      \paragraph*{Stars-graphs.} 
      We provide a similar reduction from \Clique. 
      For each vertex $v_i \in V(\hat{G})$ create $3\hat{n}^2$ \myemph{vertex-agents} $p_i^1, \ldots, p_i^{3\hat{n}^2}$. 
      For each edge $e_\ell \in E(\hat{G})$, create $\hat{n}^4$ \myemph{edge-agents} $q_{\ell}^1, \ldots, q_{\ell}^{\hat{n}^4}$. 
      Since we will have symmetric preferences, for each pair of agents we only specify one value. 
      For each $e_\ell \in E(\hat{G})$ and $v_i \in V(\hat{G})$, do the following:
      \begin{compactitem}[--]
        \item For each $z \in \{0\}\cup[\hat{n}^2-1]$, $s \in [2]$, set 
        $\sat[p_i^{3z+s}](p_i^{{3z+s+1}}) = 1$.
        \item For each $z \in [\hat{n}^2-1]$, set 
        $\sat[p_i^{3(z-1)+2}](p_i^{{3z+2}}) = 1$.
        \item For each $z \in [\hat{n}^4-1]\setminus\{2\}$, set 
        $\sat[q_{\ell}^{z}](q_{\ell}^{z+1}) = 1$, and 
        $\sat[q_{\ell}^2](q_{\ell}^4) = \sat[q_{\ell}^3](q_{\ell}^6) = 
        \sat[q_{\ell}^{\hat{n}^4}](q_{\ell}^7) = 1$, and  
        $\sat[q_{\ell}^1](q_{\ell}^3) = \sat[q_{\ell}^1](q_{\ell}^5)= 
        -1$.
        \item $\sat[q_{\ell}^6](p_{i}^{3j}) = \sat[q_{\ell}^6](p_j^{3i})= 
        -1$ if $e_\ell = \{v_i, v_j\}$.
        \item The non-mentioned preferences are set to zero. 
      \end{compactitem}
      \cref{fig:EFA_delta_stars_symm} %
      shows this construction for two adjacent vertices $v_i, v_j \in V(\hat{G})$.
  		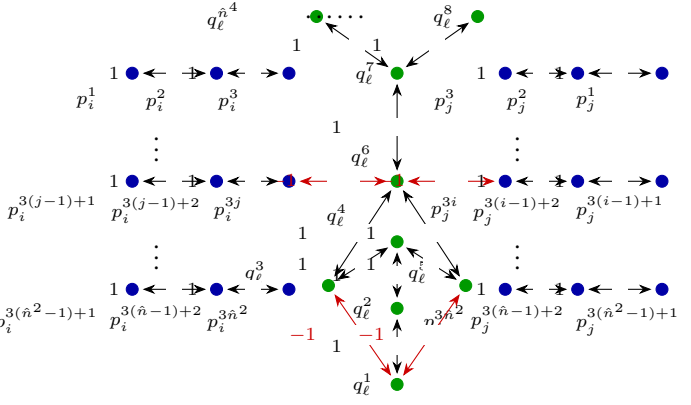
\begin{figure}[t]
  			\centering
  			\begin{tikzpicture}[>=stealth', shorten <= 1.5pt, shorten >= 1.5pt]
  			\hspace{-8mm}
  			\def \xs {5ex}
  			\def \qxs {6ex}
  			\def \ys {8ex}
  			\def \qys {4ex}

  			\node[nodeW] (q6) {};
  			\node[nodeW, below = 2*\ys of q6] (q1) {};
  			\node[nodeW, below = 1.2*\ys of q6] (q2) {};
  			\node[nodeW, below left = \ys and \xs of q6] (q3) {};
  			\node[nodeW, below = 0.5*\ys of q6] (q4) {};
  			\node[nodeW, below right = \ys and \xs of q6] (q5) {};
  			\node[nodeW, above = \ys of q6] (q7) {};
  			\node[nodeW, above right = \qys and \qxs of q7] (q8) {};
  			\node[nodeW, above left = \qys and \qxs of q7] (q9) {};
  			
  			\path (q8) -- node[auto=false]{\ldots\ldots} (q9);
  			
  			\foreach \i / \pos / \v / \l in {1/right/1/0, 3/above left/3/-6, 4/above/4/0, 5/above right/5/-6, 7/right/7/1, 8/right/8/0, 9/left/\hat{n}^4/0}{
  				\node[\pos = \l pt of q\i] {\scriptsize $q_\ell^{\v}$};
  			}
  			\node[above right = 0pt of q6] {\scriptsize $q_\ell^{6}$};
  			\node[right = 0pt of q2] {\scriptsize $q_\ell^{2}$};
  			
  			\foreach \i /\j / \pos / \rr in {i/j/left/20, j/i/right/340}{
  				\node[nodeU, \pos = 1.6*\xs of q6] (p\i1h1) {};
  				\node[below = 0pt of p\i1h1] {\scriptsize$p_\i^{3\j}$};
  				
  				\node[nodeU, above = \ys of p\i1h1] (p\i0h1) {};
  				\node[below = 0pt of p\i0h1] {\scriptsize$p_\i^3$};
  				
  				\node[nodeU, below = 1*\ys of p\i1h1] (p\i2h1) {};
  				\node[below = 0pt of p\i2h1] {\scriptsize$p_\i^{3\hat{n}^2}$};
  				
  				\foreach \x / \n / \m in {0h/2/1, 1h/3(\j-1)+2/3(\j-1)+1, 2h/3(\hat{n}-1)+2/3(\hat{n}^2-1)+1}{
  					\node[nodeU, \pos = \xs of p\i\x1] (p\i\x2) {};
  					\node[below = 0pt of p\i\x2] {\scriptsize$p_\i^{\n}$};
  					
  					\node[nodeU, \pos = 1.2*\xs of p\i\x2] (p\i\x3) {};
  					\node[below \pos = 0pt and -15pt of p\i\x3] {\scriptsize$p_\i^{\m}$}; 
  				}
  				\path (p\i1h2) -- node[pos=0.65, auto=false]{\vdots} (p\i2h2);
  				\path (p\i0h2) -- node[pos=0.65, auto=false]{\vdots} (p\i1h2);
  			}

  			\begin{scriptsize}
  			\foreach \x / \y in {1/2, 2/4, 3/4, 3/6, 4/5, 5/6, 6/7, 7/8, 7/9} {
  				\draw[Stealth-Stealth] (q\x) edge node[pos=0.5, fill=white, inner sep=1pt] {$1$} (q\y); 
  			}
  			\foreach \x / \y in {1/3, 1/5} {
  				\draw[Stealth-Stealth] (q\x) edge[myRed] node[pos=0.5, fill=white, inner sep=1pt] {$-1$} (q\y); 
  			}
  			
  			\foreach \i / \b in {i/right, j/left}{ 
  				\draw[Stealth-Stealth] (q6) edge[bend \b=0, myRed] node[pos=0.5, fill=white, inner sep=1pt] {$-1$} (p\i1h1);
  			}
  			
  			\foreach \i in {i, j}{					
  				\foreach \x in {0h,1h,2h}{	
  					\draw[Stealth-Stealth] (p\i\x1) edge node[pos=0.5, fill=white, inner sep=1pt] {$1$} (p\i\x2);
  					\draw[Stealth-Stealth] (p\i\x2) edge node[pos=0.5, fill=white, inner sep=1pt] {$1$} (p\i\x3);
  				}
  			}
  			\end{scriptsize}
  			\end{tikzpicture}
  			\caption{Preference graph for \cref{thm:EFA_delta_clique+path-cycle+stars_symm} for stars-graphs, where $e_{\ell}=\{v_i, v_j\}$.}
  			\label{fig:EFA_delta_stars_symm}
  		\end{figure}  
      The seat graph consists of $h(h-1)$ stars with three leaves, $h(\hat{n}^2-h+1)+\binom{h}{2}$ $P_3$'s (i.e., $\nonisolated \coloneqq 3h\hat{n}^2+5\binom{h}{2}$) and $\hat{n}^3+|E(\hat{G})|\hat{n}^4-k$ isolated vertices.  
      The construction of this \EFA instance $I$ can clearly be done in polynomial time and it holds $\maxoutdeg = 5$.
      
      It remains to show the correctness, i.e., $\hat{G}$ admits a size-$h$ clique if and only if $I$ has an \efArr. 		
      For the ``only if'' part let $C$ be a size-$h$ clique of $\hat{G}$. We obtain an \efArr by the following:
      \begin{compactitem}[--]
        \item For each edge $e_\ell = \{v_i, v_j\} \in C$, assign each agent from $\{p_i^{3(j-1)+1}, p_i^{3(j-1)+2}, p_i^{3j}, q_{\ell}^3\}$ resp.\ from $\{p_j^{3(i-1)+1}, p_j^{3(i-1)+2}, p_j^{3i}, q_{\ell}^5\}$ to a star with three leaves, where $p_i^{3(j-1)+2}$ resp.\ $p_j^{3(i-1)+2}$ is assigned to the center. 
        Moreover, assign $q_{\ell}^1, q_{\ell}^2, q_{\ell}^4$ in this order to a $P_3$. 
        
        \item For each $v_i\in C, z \in [\hat{n}^2]\setminus \{j|v_j \in C\}$ assign $p_i^{3(z-1)+1}, p_i^{3(z-1)+2}, p_i^{3z}$ in this order to a $P_3$. 
        
        \item All remaining agents are assigned to isolated vertices.
      \end{compactitem}
      This arrangement is envy-free because of the following:
      \begin{compactenum}[(i)]
        \item The vertex-agents assigned to leaves have their maximum possible utility. 
        The vertex-agents assigned to a center of a star have utility two, but cannot increase this in this arrangement. 
        Thus, they are envy-free.
        
        \item The vertex-agents $p_i^z$ with $v_i \notin C, z \in [3\hat{n}^2]$ and edge-agents $q_\ell^s$ with $e_\ell \notin C, s \in [\hat{n}^4]$ are envy-free since there is no non-isolated agent towards which they have a positive preference.

        \item The edge-agents corresponding to an edge in the clique are assigned in such a way that they are envy-free.
      \end{compactenum}			
      
      Conversely, let $\sigma$ be an \efArr of $I$. 
      We observe the following properties of an \efArr:
      \begin{claim} \label{claim:EFA_delta_stars_symm}
        Every \efArr satisfies:
        \begin{compactenum}[(i)]
          \item\label{EFA_delta_stars_symm:no-ql} For each $e_\ell \in E(\hat{G})$, agents $q_{\ell}^6, \ldots, q_{\ell}^{\hat{n}^4}$ are assigned to isolated vertices.
          
          \item\label{EFA_delta_stars_symm:all-or-nothing-v} For each $v_i \in V(\hat{G})$, the corresponding vertex-agents are all assigned to either non-isolated or isolated vertices.
          If they are non-isolated, then for each $z \in [\hat{n}^2]$, the seats assigned to agents $p_i^{3(z-1)+1}, p_i^{3(z-1)+2}, p_i^{3z}$ are adjacent and $p_i^{3(z-1)+2}$ is assigned to the center of a star. 
          
          \item\label{EFA_delta_stars_symm:clique} At most $h$ different sets of vertex-agents and at most $\binom{h}{2}$ different sets of edge-agents are assigned to the stars.
        \end{compactenum}
      \end{claim}
      \begin{proof}[Proof of \cref{claim:EFA_delta_stars_symm}]
        \renewcommand{\qedsymbol}{$\diamond$}
        Let~$\sigma$ be an \efArr.
        
        Statement~\eqref{EFA_delta_stars_symm:no-ql}: This follows from \cref{lem:EFA_non-neg_cond} and $k < \hat{n}^4-5$ for $\hat{n}>3$. 
        
        Statement~\eqref{EFA_delta_stars_symm:all-or-nothing-v}: By the previous statement, we know that for each vertex-agent no agents with negative preferences are assigned to the stars. 
        Hence, by \cref{lem:EFA_non-neg_cond} it follows that all agents corresponding to the same vertex are either isolated or non-isolated.
        This proves the first part of Statement~\eqref{EFA_delta_stars_symm:all-or-nothing-v}.
        
        Next, suppose, towards a contradiction, that the second part of Statement~\eqref{EFA_delta_stars_symm:all-or-nothing-v} does not hold. Then the utility of  $p_i^{3(z-1)+1}$ or $p_i^{3z}$ is zero and this agent has an envy, a contradiction.
        
        Statement~\eqref{EFA_delta_stars_symm:clique}: By Statement~\eqref{EFA_delta_stars_symm:all-or-nothing-v}, at most $h$ different sets of vertex-agents are assigned to the stars. 
        
        For each edge $e_\ell = \{v_i,v_j\} \in E(\hat{G})$, if agent $q_{\ell}^3$ resp.\ $q_{\ell}^5$ is non-isolated, then by the previous statement, he is assigned to a leaf of a star, where the center agent is $p_i^{3(j-1)+2}$ resp.\ $p_j^{3(i-1)+2}$ to make $q_\ell^6$ envy-free. 				
        Furthermore, by \cref{lem:EFA_non-neg_cond}, if one of the agents $q_{\ell}^3$ or $q_{\ell}^5$ is non-isolated, then  $q_{\ell}^1, q_{\ell}^2, q_{\ell}^4$ are also non-isolated and their assigned seats form a $P_3$. 
        
        Combining these statements we obtain that none of the agents $q_{\ell}^1, q_{\ell}^2, q_{\ell}^3$ can be assigned to the center of a star with four leaves since there are no other agents which can be assigned to the remaining leaves.
        This means that the remaining candidates for the star centers with four leaves are vertex-agents. 
        
        For these $h(h-1)$ stars the fourth leaf has to be assigned to an incident edge-agent. 
        Since $\binom{h}{2}+1$ edges are incident to at least $h+1$ vertices, we can conclude that no more than $\binom{h}{2}$ different sets of edge-agents are assigned to the stars. 				
      \end{proof}	
      From \cref{claim:EFA_delta_stars_symm} we obtain that $h$ different sets of vertex-agents and $\binom{h}{2}$ edge-agents are assigned to the seats. 
      Moreover, these non-isolated edge-agents have to be incident to the non-isolated vertex-agents. 
      
      Summing up, we obtain that~$I$ admits an \efArr if and only if $\hat{G}$ admits a size-$h$ clique~$C$. 
    \end{proof}
  }
  
  Under non-negative or symmetric preferences, we obtain an \fpt\ algorithm for $\nonisolated+\maxoutdeg$, which is based on random separation and dynamic programming. 
  
  \newcommand{\efakdeltanonnegsymm}{%
  	For non-negative or symmetric preferences, \EFA is \fpt\ wrt.\ $\nonisolated+\maxoutdeg$. %
  }
\statementarxiv{theorem}{thm:EFA_k+delta_nonneg}{\efakdeltanonnegsymm}

\ifarxiv
\begin{proof} 
	\textbf{Non-negative preferences. }
\else
\ifshort
	\begin{proof}[Proof sketch]
		We only consider the case with non-negative preferences. 
\else
	\begin{proof}[Proof sketch]
		We only consider the case with non-negative preferences. 
		The case with symmetric preferences can be found in the appendix.
		
\fi 		
\fi   	
    Let $I=(\agents, (\sat)_{p\in \agents}, G)$ be an \EFA instance with non-negative preferences. %
    Then, no non-isolated agent in an \efArr\ can have any in-arcs from isolated agents by the contrapositive of \cref{lem:EFA_non-neg_cond}.
    This means that each non-isolated agent has bounded in- and out-degree.
    Hence, we can use random separation to separate the non-isolated agents from their isolated neighbors.	
    The approach is as follows: Color every agent independently with~\textit{r} or \textit{b}, each with probability $1/2$. 
    We say that a coloring~$\chi\colon \agents \to \{r,b\}$ is \myemph{successful} if there exists an \efArr~$\sigma$ such that
    \begin{compactenum}[(i)]
      \item $\chi(p)=r$ for each $p \in W$ and
      \item $\chi(p)=b$ for each $p \in (\agents \setminus W)\cap \Nout{\prefgraph}(W)$, 
    \end{compactenum}
    where $W\coloneqq \{p \in \agents \mid \degr{G}(\sigma(p)) \geq 1\}$ denotes the set of~non-isolated agents in $\sigma$. 
    Note that by the above reasoning, if \mbox{$p \in W$,} then it holds $\Nin{\prefgraph}(p) \subseteq W$ and by the first condition, all agents in $\Nin{\prefgraph}(p)$ are colored with~$r$.   
    Since the seat graph has~$\nonisolated$ non-isolated vertices and the out-degree of each agent in the preference graph is bounded by $\maxoutdeg$, we infer that $|W|+|(\agents \setminus W)\cap \Nout{\prefgraph}(W)| \leq \nonisolated (1+\maxoutdeg)$.
    Hence, the probability that a random coloring is successful is at least $2^{-\nonisolated (1+\maxoutdeg)}$.
    
    Let $\redcomp$ be the subset of agents colored with~$r$.
    We already know for each weakly connected component $C$ of $\redcomp$ that the agents in $C$ are all assigned to either isolated or non-isolated vertices (see \cref{lem:EFA_non-neg_cond}). 
    Therefore, the size of each component is bounded by $\nonisolated$. 
    It remains to decide which component to assign to non-isolated vertices and how to assign them.
    
    Let the vertices of the seat graph $G$ be denoted by $\{1, 2, \dots, \nonisolated\}$.
    We design a simple algorithm using color-coding as follows. 
    We color the red agents $\redcomp$ uniformly at random with colors     $[\nonisolated]$. 
    The $\nonisolated$ colors one-to-one correspond to the $\nonisolated$ non-isolated seats in the seat graph.  
    Let $\sigma$ be a hypothetical \efArr and $\chi'\colon \redcomp \to [\nonisolated]$ a coloring.  
    We say $\chi'$ is \myemph{good} (wrt.\ $\sigma$) if the agent at the $i$th vertex of~$G$ is colored $i$, i.e., $\chi'(\sigma^{-1} (i)) = i$, for each $i \in [\nonisolated]$.  
    Note that given a solution $\sigma$, the probability that the $\nonisolated$ non-isolated agents are colored with pairwise distinct colors is at least $e^{-\nonisolated}$~\cite{cygan15}. 
    Since for each component all agents are assigned to either non-isolated or to isolated vertices, we can first check for each component, if each color appears at most once. 
    If there are two agents~$p_1$ and~$p_2$ with $\chi'(p_1)=\chi'(p_2)$, then this component will be assigned to isolated vertices in a good coloring.
    Since there are $\Oh(n)$ components, this can be done in time $\Oh(\nonisolated^2 \cdot n)$.
    
    For each remaining component, we check whether the given arrangement is envy-free. 
    In this regard, we observe that the utility of an agent only depends on the agents inside the same component as all preferences between two components are zero.
    Hence, we compute for each agent~$p$ with $\chi'(p)= i$ in a component 
    $C$ his utility
    \mbox{$\util[p]{\rho} = \sum_{\substack{q \in C\setminus\{p\}, \chi'(q)=j, 
        \{i,j\} \in E(G)}} \sat[p](q)$.}
    Similarly, we can compute the utility of $p$ in a swap-arrangement, where $p$ and the agent assigned to seat $j\neq i$ swap their seats and determine whether $p$ is envy-free. 
    If not, we can assign this component to isolated vertices. 
    Since each component has $\Oh(\nonisolated)$ agents, this step can be done in time $\Oh(\nonisolated^2\cdot n)$. 

    Finally, we use dynamic programming (DP) to select from the remaining envy-free components those whose colors match the seats and sizes sum up to~$\nonisolated$.
    Let $C_1, C_2, \ldots, C_m$ be the remaining weakly connected components.
    We define a DP table where an entry $T[S,i]$ is true if there is a partial arrangement assigning the first $i$ weakly connected components, where
    \begin{inparaenum}[(i)]
      \item no two non-isolated agents are colored with the same color and
      \item each color in $S$ is used once.
    \end{inparaenum}
    
    We start filling our table for $i=1$ as follows:
    
    {\centering
      $T[S,1] =$ \textit{true} $\Leftrightarrow |S|=0 \vee S = \bigcup_{p 
        \in C_1} 
      \{\chi'(p)\}$. 
      \par}    
  	{\noindent 
    Each component $C_i$ is either non-isolated or isolated. 
    Therefore, the following recurrence holds:
	}

    {\centering
      $T[S,i] = T[S,i-1] \vee T\left[S\setminus\bigcup_{p \in C_i} 
        \{\chi'(p)\}, i-1\right]$.
      \par}
  	{\noindent      
    Since $S \subseteq [\nonisolated]$, the entries of this table can be computed in time $2^\nonisolated\cdot n^{\Oh(1)}$.   
	}
    Summing up, the above algorithm runs in time $2^{\nonisolated}\cdot n^{\Oh(1)}$. 
    The probability of a successful and good coloring is at least \mbox{$2^{-\nonisolated(1+\maxoutdeg)}\cdot e^{-\nonisolated}$.} 
    Hence, after repeating this algorithm $2^{\nonisolated(1+\maxoutdeg)}\cdot e^{\nonisolated}$ times we obtain a solution with high probability. 
    Finally, we can de-randomize the random separation and color-coding approaches while maintaining fixed-parameter tractability~\cite{cygan15}.
    \appendixcontinue{thm:EFA_k+delta_nonneg}{}{\efakdeltanonnegsymm}{
    \ifarxiv
    \paragraph*{Symmetric preferences.} 
    \fi 
    We continue with the proof for the case with symmetric preferences and provide a proof sketch for that.
    Let $I=(\agents, (\sat)_{p\in \agents}, G)$ be an \EFA instance. 
    Since the preferences are symmetric, we assume that the corresponding preference graph~$\prefgraph$ is undirected. 
    We denote the non-isolated vertices of the seat graph~$G$ by $\niv_1, \dots, \niv_\nonisolated$. 
    	
    \paragraph*{Guess the structure of a solution.} 
    	For a given \efArr~$\sigma$ of~$I$, we observe the following about the preference structure in~$\sigma$. 
        Obviously exactly~$\nonisolated$ agents are non-isolated in~$\sigma$ which can have preferences among each other. 
    	Hence, the preference graph induced by these non-isolated agents can consist of several connected components.     	
    	Furthermore, we know that the~$\nonisolated$ non-isolated agents in~$\sigma$ have at most $\nonisolated\maxoutdeg$ non-zero preferences towards isolated agents. 
    	When constructing a solution, these are the isolated agents which can have an envy, since the others have only zero preferences towards the non-isolated agents. 
    	    	
        Hence, in this step we, guess (by branching over all possibilities of) the preference structure of a solution.
    	This means if we look at the preference graph induced by the non-isolated agents, we guess the seats of which of them belong to the same connected component. 
    	Hence, we partition the non-isolated vertices into disjoint subsets, where each two non-isolated agents assigned to a seat in a subset~$S$ are connected through a path in the preference graph using agents assigned to seats in~$S$. 
    	We denote such a partition of the non-isolated vertices by $S_1 \uplus \ldots \uplus S_x = \{\niv_1, \dots, \niv_\nonisolated\}$. 
    	The number of all possible partitions of a set of size $\nonisolated$ can be upper bounded by $\Oh(\nonisolated^{\nonisolated})$. 
    	
    	Next, we guess the number~$t$ of isolated agents in the solution which have a non-zero preference towards a non-isolated agent (i.e., are adjacent to a non-isolated agent in the preference graph). 
    	Since each agent has at most $\maxoutdeg$ neighbors in the preference graph and there are $\nonisolated$ non-isolated vertices, it holds $t \leq \nonisolated\maxoutdeg$. 
    	We denote the seats of these isolated agents as $\iv_1, \ldots, \iv_{t}$.
    	
    	Furthermore, we guess the edges in~$\prefgraph$ which are incident to a non-isolated and an isolated agent:
        Since we have $\nonisolated$ non-isolated which are incident to ~$t$ isolated agents, the number of guesses in this step is bounded by $\Oh(2^{\nonisolated\maxoutdeg})$.
    	We denote this bipartite graph by $H_{\colr\colb}$, where $V(H_{\colr\colb}) = \{\niv_1, \ldots, \niv_{\nonisolated}\} \cup \{\iv_1, \ldots, \iv_{t}\}$ and there is an edge if we guess that the two agents have a non-zero preference towards each other. 
    	
    \paragraph*{Random separation.}
    	In the previous step we guessed a structure of the solution. 
        In this step, we use random separation to separate the non-isolated agents from their isolated neighbors.
    	The approach is as follows: Color every agent independently with colors $\colr_1, \ldots, \colr_{\nonisolated}, \colb_1, \ldots, \colb_{t}$, each with probability $1/(\nonisolated+t)$. 
    	We say that a coloring~$\chi\colon \agents \to \{\colr_1, \ldots, \colr_{\nonisolated}, \colb_1, \ldots, \colb_{t}\}$ is \myemph{successful} if there exists an \efArr~$\sigma$ such that
    	\begin{compactenum}[(i)]
    		\item $\chi(p) \in \{\colr_1, \ldots, \colr_{\nonisolated}\}$ for each $p \in W$, and
    		\item $\chi(p) \in \{\colb_1, \ldots, \colb_{t}\}$ for each $p \in (\agents \setminus W)\cap \Neigh{\prefgraph}(W)$, 
    	\end{compactenum}
    	where  $W\coloneqq \{p \in \agents \mid \degr{G}(\sigma(p)) \geq 1\}$ denotes the set of non-isolated agents in~$\sigma$.  
    	Since the seat graph has~$\nonisolated$ non-isolated vertices and the degree of each agent in the preference graph is bounded by $\maxoutdeg$, we infer that $|W|+|(\agents \setminus W)\cap \Neigh{\prefgraph}(W)| \leq \nonisolated (1+\maxoutdeg)$.
    	Hence, the probability that a random coloring is successful is at least $(\nonisolated+t))^{-\nonisolated (1+\maxoutdeg)}\geq(\nonisolated(1+\maxoutdeg))^{-\nonisolated (1+\maxoutdeg)}$. 
    	
    	We introduce some notation and notions for the sake of brevity.
        Given a coloring $\chi$, we call agents colored with a color from $\{\colr_1, \ldots, \colr_{\nonisolated}\}$ \myemph{red agents}. 
    	A \myemph{blue agent} is an agent colored with a color from $\{\colb_1, \ldots, \colb_{t}\}$ and which has a non-zero preference to some red agent.
    	We denote the set of red agents by~\myemph{$\redcomp$} and the set of blue agents by~\myemph{$\bluecomp$}. 
    	Let $\prefgraph[\redcomp]$ be the preference graph induced by the red agents.
    	A \myemph{red component} is a connected component in $\prefgraph[\redcomp]$. 
    	A blue agent~$p$ is \myemph{adjacent to a red component} $C$ if $p$ is adjacent to some red agent in $C$ in the preference graph~$\prefgraph$. 

        In the following, we do a sanity check to rule out irrelevant red components and blue agents. 

    	\noindent \textit{Red components (wrt.\ agents):} 	    	
    	If one of the following criteria is not fulfilled for a red component~$C$, we assign all agents in this component to isolated vertices. 
    	\begin{compactenum}[(RE1)]
    		\item $C$ has at most $\nonisolated$ agents. 
    		
    		\item There exists $y \in [x]$ such that it holds $\{j \colon \chi(p)=r_j \text{ for } p \in \redcomp\} = \{j \colon \niv_j \in S_y\}$ and $|C|=|S_y|$.
    		
    		\item $C$ has no two blue neighbors with the same color. 
    		
    		\item Each agent in~$C$ does not envy any other non-isolated agent when $C$ is assigned to non-isolated vertices respecting their coloring.

    	\end{compactenum}

    	In the following we describe each of the above criteria in more detail:
        
    	(RE1) By definition of a successful coloring, we know for each red component $C$ that the agents in $C$ are all assigned to either isolated or non-isolated vertices. 
    	Therefore, the size of each red component is bounded by $\nonisolated$. 
    	
    	(RE2) It remains to decide which of these components to assign to non-isolated vertices and how to assign them.
    	In this regard we note that the colors $\colr_1, \ldots, \colr_{\nonisolated}$ one-to-one correspond to the non-isolated seats $\niv_1, \ldots, \niv_{\nonisolated}$ in the seat graph. 
    	Hence, we check for each red component if the contained colors match a component in the guessed structure $S_1, \ldots, S_x$; otherwise, we assign all agents in this component to isolated vertices. 
    	Since each component has $\Oh(\nonisolated)$ agents, this step can be done in time $\Oh(\nonisolated^2 \cdot n)$. 
    	
    	(RE3) Moreover, we check that no red component has two blue neighbors with the same color; otherwise, we assign all agents in this component to isolated vertices since this cannot match the guessed structure. 
    	This can be done in time $\Oh((\nonisolated\maxoutdeg)^2 \cdot n)$.     	
    	
    	(RE4) For each remaining component we check whether the given arrangement is envy-free. 
    	In this regard, we observe that the utility of an agent only depends on the agents inside the same component as all preferences between two components are zero.
    	Hence, we compute for each agent $p$ with $\chi(p)= \colr_i$ in a component $C$ his utility
    	$\util[p]{\rho} = \sum_{\substack{q \in C\setminus\{p\}, \chi(q)=\colr_j, \{\niv_i,\niv_j\} \in E(G)}} \sat[p](q)$.
    	Similarly, we can compute the utility of $p$ in a swap-arrangement, where $p$ and some agent colored with $\colr_j\neq \colr_i$ swap their seats and determine whether $p$ is envy-free. 
    	If $p$ is not envy-free, we can assign all agents in the component~$C$ to the isolated vertices. 
    	Since each component has $\Oh(\nonisolated)$ agents, this step can again be done in time $\Oh(\nonisolated^2\cdot n)$. 
    	    	    	
    	\noindent \textit{Blue agents:}  
    	We do a sanity check for each blue agent~$p$ with $\chi(p)=b_i$ in two steps:
    	\begin{compactenum}
    		\item[(BL1)] The colors of the red neighbors of~$p$ contain all colors indicated the neighbors of $\iv_i$ in the guessed structure. 
    		
    		\item[(BL2)] $p$ has at least one \emph{compatible} combination of red-components. 
    	\end{compactenum}
    
        In the following we describe each of the above steps in more detail:      	  	
    	
    	(BL1) The colors $\colb_1, \ldots, \colb_{t}$ one-to-one correspond to the~$t$ isolated agents $\iv_1, \ldots, \iv_{t}$ guessed in the first step.     	
    	Hence, we check now whether our coloring~$\chi$ matches the guessed structure. 
    	This means for every blue agent $p \in \bluecomp$ with $\chi(p)=b_i$, the non-zero preferences towards red agents should match $H_{\colr\colb}$. 
    	In particular, for every agent~$p \in \bluecomp$ colored $\colb_i$ it should hold:
    	\begin{align*}
    		\{j\colon \exists q \in \redcomp \cap \Neigh{\prefgraph}(p) \text{ s.t. } \chi(q) = \colr_j \} = 
    		\{j\colon \niv_j \in \Neigh{H_{\colr\colb}}(\iv_i) \} .
    	\end{align*}
    	If the left-hand side (LHS) above is a superset of the right-hand side (RHS), 
        then for each index~$j$ not contained RHS, we assigned all adjacent red agents (and the whole red component in which they are contained) colored with $\colr_j$ and check the above equality again. 
        If LHS is a strict subset of RHS,
        then we assign all adjacent red agents (and the whole red component in which they are contained) and~$p$ to isolated vertices. 
    	
    	(BL2) After the previous step, every blue agent in $\bluecomp$ is left with at most~$\maxoutdeg$ red components where the colors of each component ``match'' with a set in $S_1, \ldots, S_x$. %
    	Next, we check for each blue agent if he is a possible candidate for the isolated seats $\iv_1, \ldots, \iv_{t}$. 
    	To this end, let $\mathcal{R} = \{R_1, ..., R_{n'}\}$ be the connected red components in $\prefgraph[\redcomp]$, and 	
    	we say a subfamily of red components $\mathcal{R'} \subseteq \mathcal{R}$ of size at most $\maxoutdeg$ is \myemph{compatible} with a blue agent~$p$ colored $\chi(p) = b_i$ if 
    	\begin{compactenum}[(i)]
    		\item for each $u_j\in N_{H_{rb}}(v_i)$ there exists exactly one agent $q \in \bigcup_{R \in \mathcal{R'}}R$ with color $r_j$ i.e., $\chi(q) = r_j$,
    		\item $p$ is envy-free when among all adjacent red components of~$p$ exactly those from $\mathcal{R'}$ are assigned to non-isolated and~$p$ to isolated vertices.
    	\end{compactenum} 
	    The utility of a blue agent is always zero because he will be assigned to an isolated vertex and his utility in a swap-arrangement depends only on the agents in the red components. 
    	Therefore, we can check (similar to the red agents) if assigning a specific combination of red components to non-isolated vertices makes a specific blue agent envious in time $\Oh(\nonisolated^2)$.
    	Next, we check for each blue agent if there exists at least one compatible combination of red-components. 
    	If this is not the case for some $p \in \bluecomp$, then we assign all the adjacent red components and $p$ to isolated vertices. 
    	Since each blue agent has $\Oh(\maxoutdeg)$ adjacent red components, there are $2^{\maxoutdeg}$ possible combinations of them. 
    	Since there are $\Oh(n)$ blue agents, this step can be done in time $\Oh(n\cdot 2^{\maxoutdeg}\cdot\nonisolated^2)$. 
    	    	
   	\paragraph*{Construct the compatibility graph.}
    	After we have filter out irrelevant red components and blue agents, we need to decide which remaining red components are good for the corresponding non-isolated vertices. 
    	However, we cannot decide on each red component independently since choosing a red component may have an influence on the choice of another red component since otherwise a blue neighbor will be envious.
    	In other words, we need to decide which blue agent~$p$ is a neighbor of a non-isolated agent and which adjacent red components of~$p$ shall be non-isolated. 
    	To achieve this, we construct a so-called compatibility graph where for each blue agent there are multiple copies of non-adjacent vertices and check whether the graph contains an induced subgraph which corresponds to the guessed structure.  
    	
    	In the guessed structure if some $\iv_i, \iv_j$ are adjacent to the same component $S_{y}, y \in [x]$, then also the blue agents colored with $\colb_i, \colb_j$ need to be compatible for at least one combination. 
    	Hence, for two blue agents $p,q \in \bluecomp$ with $\chi(p) = \colb_i$, $\chi(q) = \colb_j$, and $\mathcal{R}_p, \mathcal{R}_q \subseteq \mathcal{R}$ we say that $(p,\mathcal{R}_p)$ is \myemph{compatible} with $(q,\mathcal{R}_q)$ if it holds:
    	\begin{compactenum}[(i)]
          \item $\mathcal{R}_p$ is compatible with $p$ and $\mathcal{R}_q$ is compatible with $q$, 
          \item $\mathcal{R}_p \cap \mathcal{R}_q \neq \emptyset$, and
          \item for each $S_y, y \in [x]$ which is adjacent to~$\iv_i$ and~$\iv_j$ there exists exactly one $R \in \mathcal{R}_p \cap \mathcal{R}_q$ such that
          \begin{align*}
            \{\ell\colon \niv_\ell \in S_y\} = \{\ell\colon \chi(q') = \colr_{\ell} \text{ for } q' \in R\}.
          \end{align*}
        \end{compactenum}
        Together with the two previous sanity checks, the third condition above implies that the common red components of $p$ and $q$ should one-to-one correspond to guessed structure restricted to the common neighbors of $\iv_i$ and $\iv_j$.
        Note that the number of blue colors is $\Oh(\nonisolated\maxoutdeg)$,
        the number of compatible combinations of red components for a blue agent is $\Oh(2^{\maxoutdeg})$, and the number of components in a combination is $\Oh(\maxoutdeg)$. %
    	
    	To check whether two blue agents can be ``used'' for the solution, we create a colored compatibility graph $\compgraph$. 
    	We add a vertex $w_{p, \mathcal{R}_p}$ (corresponding to agent~$p$) for each blue agent $p \in \bluecomp$ and each compatible combination of red components $\mathcal{R}_p$. 
    	The vertex $w_{p, \mathcal{R}_p}$ is colored with $\chi(p)$. 
    	There is an edge between two vertices $w_{p, \mathcal{R}_p}$ and $w_{q, \mathcal{R}_q}$ if $(p, \mathcal{R}_p)$ is compatible with $(q, \mathcal{R}_q)$. 
    	The number of vertices of $\compgraph$ is $\Oh(n2^{\maxoutdeg})$.
    	If two vertices are adjacent, then $p$ and $q$ are adjacent to the same red component at least once. 
    	Since each agent is adjacent to at most $\Oh(\maxoutdeg)$ red components, the number of neighbors of a red component is bounded by $\Oh(\nonisolated\maxoutdeg)$, and the number of partitions for each agent is $\Oh(2^{\maxoutdeg})$, we obtain that the maximum degree of $\compgraph$ is bounded by $\Oh(\maxoutdeg\cdot\nonisolated\maxoutdeg \cdot 2^{\maxoutdeg})$. 
    	
    \paragraph*{Construct a solution from the compatibility graph.}
    	Finally, to find an envy-free arrangement, we check whether the compatibility graph~$\compgraph$
    	contains an induced subgraph isomorphic to the to the guessed structure. 
    	To this end, we create an incidence graph $\incgraph$ of the guessed structure. 
    	The set of vertices $V(\incgraph)$ consists of $\{\iv_1, \ldots, \iv_{t}\}$ and two vertices $\iv_i, \iv_j \in V(\incgraph)$ are adjacent in $\incgraph$ if $\Neigh{H_{\colr\colb}}(\iv_i) \cap \Neigh{H_{\colr\colb}}(\iv_j) \neq \emptyset$, i.e., $\iv_i$ and $\iv_j$ have a common red neighbor in the guessed structure $H_{\colr\colb}$. 
    	
    	To check if~$\compgraph$ has an induced subgraph isomorphic to~$\incgraph$ we need the next claim.
    	    	
    	\begin{claim} \label{claim:indsubgraphisomorph}
    		Let $\smallgraph$ be a graph where the vertices are denoted by $\{v_1, \ldots, v_{\tilde{n}}\}$. %
    		Furthermore, let $\largegraph$ be a colored graph with maximum degree $\Delta$ and colored with colors from $\{\colb_1, \ldots, \colb_{\tilde{n}}\}$. 
    		Then, we can decide in time $f(\tilde{n}, \Delta) \cdot |V(\largegraph)|^{\Oh(1)}$ whether~$\largegraph$ contains a \myemph{colorful induced subgraph isomorphic} to $\smallgraph$, i.e., 
               whether the exists a subset of vertices $W \subseteq V(\largegraph)$ and a bijection $\rho\colon W \to V(\smallgraph)$ such that
    		\begin{compactenum}[(1)]
    			\item $\{w,w'\} \in E(G[W])$ iff. $\{\rho(w), \rho(w')\} \in E(\smallgraph)$, and
    			\item for each $w \in W$ with color $\gamma(w) = b_j$ it holds that $\rho(w) = v_j$. 
    		\end{compactenum}
    	\end{claim}
    	\begin{proof}[Proof of \cref{claim:indsubgraphisomorph}]
    		\renewcommand{\qedsymbol}{$\diamond$}
    		We prove the statement using random separation:
    		Color every vertex of $\largegraph$ independently with colors~$y$ and~$g$, each with probability $1/2$. 
    		We say a random coloring~$\psi\colon V(\largegraph) \to \{y,g\}$ is \myemph{successful} if there exists a colorful induced subgraph $\largegraph[W]$ isomorphic to $\smallgraph$ such that %
    		\begin{compactenum}[(i)]
    			\item $\psi(v) = y$ for each $v \in W$, and
    			\item $\psi(v) = g$ for each $v \in (V(\largegraph) \setminus W)\cap \Neigh{\largegraph}(W)$.
    		\end{compactenum}  
    		Since graph $\smallgraph$ has~$\tilde{n}$ vertices and the degree of each vertex in $\largegraph$ is bounded by $\Delta$, we infer that $|W|+|(V(\largegraph) \setminus W)\cap \Neigh{\largegraph}(W)| \leq \tilde{n} (1+\Delta)$.
    		Hence, the probability that a random coloring is successful is at least $2^{-\tilde{n} (1+\Delta)}$. 
    		We say a vertex in $\largegraph$ is \myemph{yellow} if he is colored with $y$; otherwise we say he is \myemph{green}. 
    		
    		After having colored the vertices of $\largegraph$, we construct a bipartite graph on $A \uplus B$:
    		Set~$A$ consists of a vertex for each connected component in~$\smallgraph$. 
    		Set~$B$ consists of a vertex for each connected component in the graph~$\largegraph$ induced by the yellow vertices.     		
    		There is an edge between a vertex~$a$ from~$A$ and a vertex~$b$ from~$B$ if and only if the following holds:
    		\begin{compactenum}[(i)]
    			\item the colors of $C_b$ match with the vertices in $C_a$, i.e., 
    			\begin{align*}
    				\{j \colon \gamma(p)=\colb_j \text{ for } p \in C_b \} = \{j \colon \iv_j \in C_a\}, 
    			\end{align*}
    			\item $\smallgraph[C_a]$ is isomorphic to $\largegraph[C_b]$,
    		\end{compactenum}
    		where $C_a$ resp.\ $C_b$ is the set of vertices in the component in~$\smallgraph$ (resp.~$\largegraph$) corresponding to $a$ (resp.~$b$).
			The construction of this graph can be done in time $\Oh((\tilde{n}|V(\largegraph)|)^2)$. 
    		
    		Finally, we need to find a maximum cardinality matching, which can be done in time~$\Oh((\tilde{n}+|V(\largegraph)|)^3)$.     		
    		If the size of the matching is equal to $|A|$, then we can construct a solution by taking the components which correspond to a matched vertex.     		
    		Otherwise, there is no solution for the current guess and coloring.
    		
    		Summing up, this runs in \fpt\ (wrt.\ $\nonisolated+\maxoutdeg$) time.%
    	\end{proof}
    	
    	By \cref{claim:indsubgraphisomorph} we can decide in time $f(\nonisolated,\maxoutdeg) \cdot (n2^{\maxoutdeg})^{\Oh(1)}$ if $\incgraph$ is isomorphic to an induced subgraph of~$\compgraph$.     	
    	If the answer is yes, then we obtain an \efArr by doing the following: 
    	Let $S$ be the vertices of $\compgraph$ which are isomorphic to $\incgraph$. 
    	Since we have a colored induced subgraph, we know for each pair of vertices $(p, \mathcal{R}_p), (q, \mathcal{R}_q) \in S$ that $p \neq q$. 
    	Hence, for each selected agent exactly one combination of components is selected. 
    	Furthermore, by the definition of the edges in $\compgraph$, we know that if $p,q$ are adjacent in $\compgraph$ then they are also compatible. 
    	Hence, each color in the combination of components can appear at most once. 
    	If a color does not appear, then the component of this color is not adjacent to an isolated vertex in the guessed structure. 
    	In this case, $\redcomp$ also needs to contain such a red component if the guesses are all correct. 
    	Therefore, by taking the subsets which are contained in the vertices of $S$ we obtain the set of non-isolated agents and by their coloring an envy-free arrangement. 
    	\iflong \qed \fi 
    } 
  \end{proof}
  
  For matching-graphs we also obtain fixed-parameter tractability even for arbitrary preferences.
  \newcommand{\efakdeltamatchingfpt}{%
   For matching-graphs, \EFA is \fpt\ wrt.\ \mbox{$\nonisolated+\maxoutdeg$}. %
  }
\statementarxiv{theorem}{thm:EFA_k+delta_matching}{\efakdeltamatchingfpt}
  \appendixproofwithstatement{thm:EFA_k+delta_matching}{\efakdeltamatchingfpt}{
    \begin{proof}
      Let $I = (\agents, (\sat)_{p \in \agents} ,G)$ be an \EFA instance where the seat graph is a matching-graph.
      
      Then, no non-isolated agent in an \efArr can have any in-arcs from the isolated agents in $\posprefgraph$. This means that each non-isolated agent has 
      bounded in- and out-degrees in $\posprefgraph$. Hence, we can use random 
      separation to separate the non-isolated agents from their isolated 
      neighbors in the positive preference graph $\posprefgraph$.
      The approach is as follows: Color every agent independently
      with one of the two colors $r$ and $b$, each with probability $1/2$.
      We say that a coloring $\chi\colon \agents\to\{r, b\}$ is 
      \myemph{successful} if there exists an \efArr $\sigma$ such that
      \begin{compactenum}[(i)]
        \item $\chi(p)=r$ for each $p \in W$, and
        \item $\chi(p)=b$ for each $p \in (\agents \setminus W)\cap
        \Noutpos{\prefgraph}(W)$,
      \end{compactenum}
      where $W\coloneqq \{p \in 
      \agents \mid \degr{G}(\sigma(p)) \geq 1\}$ denotes the set of agents which are non-isolated in $\sigma$.
      Note that for a successful coloring (wrt.\ $\sigma$) we restrict to only the positive 
      preference graph $\posprefgraph$ and no non-isolated agent can have a color-$b$ in-neighbor in~$\posprefgraph$. 
      Since the seat graph has $\nonisolated$ non-isolated vertices and the 
      (positive) out-degree of each agent in the preference 
      graph is bounded by~$\maxoutdeg$, we infer $|W|+|(\agents \setminus 
      W)\cap \Noutpos{\prefgraph}(W)| \leq 
      \nonisolated (1+\maxoutdeg)$.
      Hence, the probability that a random coloring is successful is at least 
      $2^{-\nonisolated (1+\maxoutdeg)}$.
      
      Let $\redcomp$ be the subset of agents colored with~$r$.
      For each weakly connected component $C$ of $\redcomp$ in $\posprefgraph$ the agents in $C$ are all assigned to either isolated or non-isolated vertices (see \cref{lem:EFA_non-neg_cond}).
      Therefore, the size of each such component is bounded by~$\nonisolated$. 
      It remains to decide which non-isolated component to take and how to assign them. 
      
      Let the vertices of the seat graph~$G$ be denoted by $\{1, 2, \dots, \nonisolated\}$.
      We design an \fpt\ algorithm using color-coding as follows: 
      We iterate over all possible partitions of the non-isolated vertices, which are $\Oh(k^{k})$ many. 
      Each set in a partition will be assigned to exactly one weakly connected component. 
      We color the red agents $\redcomp$ uniformly at random with colors $[\nonisolated]$. 
      The $\nonisolated$ color one-to-one correspond to the $\nonisolated$ non-isolated seats in the seat graph.
      Let $\sigma$ be a hypothetical \efArr and $\chi'\colon \redcomp \to [\nonisolated]$ a coloring.  
      We say~$\chi'$ is \myemph{good} if the agent at the $i$th vertex of~$G$ is colored~$i$, i.e., $\chi'(\sigma^{-1} (i)) = i$, for each $i \in [\nonisolated]$.
      
      Note that given a solution $\sigma$, the probability that the $\nonisolated$ non-isolated agents are colored with pairwise distinct colors is at least $e^{-k}$~\cite{cygan15}. 
      Since for each component all agents are all assigned to either non-isolated or isolated vertices, we can first check for each component, if each color appears at most one. 
      If there are two agents $p_1, p_2$ with $\chi'(p_1) = \chi'(p_2)$, then this component will be assigned to isolated vertices for a good coloring. 
      Moreover, we check for each component if its colors match a set in the partition~$P$. 
      If this is not the case, then this component will be assigned to isolated vertices. 
      Since there are $\Oh(n)$ components, this can be done in time $O(k^2\cdot n)$.

      For the remaining components we will check individually if the given 
      arrangement is envy-free or not. 
      In this regard, we observe that the utility of each agent has to be 
      non-negative and each agent with positive preferences has to be matched 
      to his best preference inside the component as all preferences between 
      two components is at most zero.
      If this is not the case, we can assign this component to isolated 
      vertices. 
      Since each component has $\Oh(k)$ agents, this step can be done in time 
      $\Oh(k^2 \cdot n)$. 
      
      Let $C_1, C_2, \ldots, C_m$ be the remaining weakly connected 
      components.
      In each of these components there can be some agents, which have to be matched with agents of other components which should have zero preferences towards each other. 
      To select the right components, we construct a multicolored independent set instance $\hat{G}$: 
      The color-classes of this instance are the sets in the partition $P$. 
      There is a vertex~$v_i$ for each component $C_i$ in a color-class if the set of colors in $C_i$ is equal to the set of the partition. 
      There is an arc $(v_i, v_j)$ from a vertex $v_i$ to $v_j$, if there are two agents $p \in C_i, q \in C_j$ and $\ell \in [k/2]$ with
      \begin{inparaenum}[(i)]
        \item $\chi'(p), \chi'(q) \in \{2\ell-1, 2\ell\}, 
        \chi'(p)\neq\chi(q)$ and 
        \item $\sat[p](q)<0$ or $\sat[q](p)<0$. 
      \end{inparaenum}
      
      Since the size of each component is bounded by $\nonisolated$, we obtain that in this independent set instance, the out-degree of each vertex is bounded by $\nonisolated\maxoutdeg$. 
      We obtain a solution by the following branching tree:
      \begin{compactitem}[--]
        \item Select a vertex $v \in V(\hat{G})$ with minimum in-degree. 
        \item Branch into either putting this vertex $v$ into the solution or one of its in- or out-neighbors. 
        \item Delete $v$ and all its in- or out-neighbors and the color-class of the selected vertex. 
        \item Decrease the parameter $\nonisolated$ by one.
      \end{compactitem} 
      We observe that in every digraph with maximum out-degree $k\maxoutdeg$, there is always a vertex with in-degree bounded by $k\maxoutdeg$. 
      In other words, there is a vertex with sum of in- and out-degrees bounded by $2\nonisolated\maxoutdeg$.
      Hence, in each step we branch into at most $2\nonisolated\maxoutdeg+1$ possibilities. 
      Since the color-classes are partitions of $\nonisolated$, the depth of this branching tree is also bounded by $\nonisolated$, i.e., the size of this tree is bounded by $(2\nonisolated\maxoutdeg+1)^\nonisolated$. 
      
      The following claim shows that if the guess is correct, we can always find a solution with this approach:
      \begin{claim}\label{claim:EFA_k+delta_matching}
        If there exists a colorful $\nonisolated$-independent set, then there exists also one which contains either a vertex $v$ of minimum in-degree or one of its in- or out-neighbors.
      \end{claim}
      \begin{proof}[Proof of \cref{claim:EFA_k+delta_matching}]
        \renewcommand{\qedsymbol}{$\diamond$}
        Let $I$ be a colorful $\nonisolated$-independent set of $\hat{G}$ and $v$ a vertex of minimum in-degree. 
        If $\Neigh{\hat{G}}(v)\cap I \neq \emptyset$, then clearly the statement holds. 
        However, if $\Neigh{\hat{G}}(v)\cap I = \emptyset$, then we can safely exchange $v$ with the vertex $w$ of the same color-class as $v$. 
      \end{proof}
      Therefore, if the guesses are correct, then a minimum in-degree vertex $v$ or one of its neighbors will appear in the branching tree at some step and we will find a $\nonisolated$-colorful independent set. 
      The vertices in this set correspond to the components which we can assign to the edges in matching-graph in an \efArr.
      
      Summing up, we obtain that this algorithm runs in time $\Oh(k^k(2\nonisolated\maxoutdeg+1)^k)$.
      The probability of a successful and good coloring is $2^{-\nonisolated(1+\maxoutdeg)}e^{-k}$. 
      Hence, after repeating this algorithm $2^{\nonisolated(1+\maxoutdeg)}\cdot e^{k}$ times we obtain a solution with high probability. 
      Finally, we can de-randomize the random separation and color-coding approaches while maintaining fixed-parameter tractability~\cite{cygan15}.  
    \end{proof}
  }
  
  \newH{ The same approach for matching-graphs does not hold for other graphs as implied by the next two theorems.
  The proofs are once again similar to the ones for \cref{thm:EFA_k_matching+clique+path-cycle_bin,thm:EFA_k_clique+stars+path-cycle_symm}.}
  \newcommand{\efakdeltacliquestarwhard}{%
    \EFA is \Wh\ wrt.\ $\nonisolated$, even if $\maxoutdeg = 3$ and the seat graph is a clique- or stars-graph. %
  }
  
  \statementarxiv{theorem}{thm:EFA_k+delta_clique+stars}{\efakdeltacliquestarwhard}
  \appendixproofwithstatement{thm:EFA_k+delta_clique+stars}{\efakdeltacliquestarwhard}{
    \begin{proof}
      We provide a parameterized reduction from the \Wh\ problem \Clique parameterized by the solution size for both classes of seat graphs. 
      Let $(\hat{G}, h)$ be an instance of \Clique.
      
      \paragraph*{Clique-graphs.} 
      For each vertex $v_i \in V(\hat{G})$, we create~$h^2$ \myemph{vertex-agents} named $p_i^1, \ldots, p_i^{h^2}$. 
      For each edge $e_\ell \in E(\hat{G})$, we create $h^4$ \myemph{edge-agents} $q_{\ell}^1, \ldots, q_{\ell}^{h^4}$. 
      For each $e_\ell \in E(\hat{G})$ and $v_i \in V(\hat{G})$, we define the preferences as follows:
      \begin{compactitem}[--]
        \item For each $z \in [h^2-1]$, set $\sat[p_i^z](p_i^{z+1}) = 1$ and $\sat[p_i^{h^2}](p_i^{1}) = 1$.
        \item For each $z \in [h^4-1]\setminus\{1\}$, set 
        $\sat[q_{\ell}^{z+1}](q_{\ell}^{z}) = 1$ and 
        $\sat[q_{\ell}^3](q_{\ell}^{h^4}) = 1$ and $\sat[q_{\ell}^2](q_{\ell}^{1}) = 2$. 
        \item If $v_i \in e_\ell$, then set $\sat[q_{\ell}^{2}](p_i^1) = -1$. 
        \item The non-mentioned preferences are set to zero. 
      \end{compactitem}
      \cref{fig:EFA_k+delta_clique} shows this construction for two adjacent 
      vertices $v_i, v_j \in V(\hat{G})$.
      \begin{figure}[t]
      	\centering
      	\begin{tikzpicture}[>=stealth', shorten <= 1.5pt, shorten >= 1.5pt]
      	\def \xs {6ex}
      	\def \ys {4ex}
      	\def \qys {3ex}

      	\node[nodeW] (q2) {};			
      	\node[nodeW, below = \ys of q2] (q1) {};
      	\node[nodeW, above = \ys of q2] (q3) {};
      	\node[nodeW, above right = \qys and \xs of q3] (q4) {};
      	\node[nodeW, above left = \qys and \xs of q3] (q5) {};
      	
      	\path (q4) -- node[auto=false, yshift=2pt]{\ldots \ldots} (q5);
      	
      	\foreach \i / \pos / \v in {1/right/1, 3/above/3, 4/right/4, 5/left/h^4}{
      		\node[\pos = 0pt of q\i] {\scriptsize$q_\ell^{\v}$};
      	}
      	
      	\node[above right = 0pt of q2] {\scriptsize$q_\ell^{2}$};
      	
      	\foreach \i / \pos in {i/left, j/right}{
      		\node[nodeU, \pos = 2*\xs of q2] (p\i1) {};
      		\node[below = 2pt of p\i1] {\scriptsize$p_\i^{1}$};
      		
      		\node[nodeU, above \pos = \ys and \xs of p\i1] (p\i2) {};
      		\node[\pos = 0pt of p\i2] {\scriptsize$p_\i^2$};
      		
      		\node[nodeU, below \pos = \ys and \xs of p\i1] (p\i3) {};
      		\node[\pos = 0pt of p\i3] {\scriptsize$p_\i^{h^2}$};

      		\path (p\i2) -- node[pos=0.5, auto=false]{\vdots} (p\i3);
      	}

      	\begin{scriptsize}
      	\foreach \x / \y / \v in {1/2/2, 2/3/1, 3/4/1, 5/3/1} {
      		\draw[-Stealth] (q\y) edge node[pos=0.25, fill=white, inner sep=1pt] {$\v$} (q\x); 
      	}
      	\foreach \i / \b in {i/right, j/left}{ 
      		\draw[-Stealth] (q2) edge[bend \b=0, myRed] node[pos=0.25, fill=white, inner sep=1pt] {$-1$} (p\i1);
      		
      		\foreach \s / \t in {1/2, 3/1}{ 
      			\draw[-Stealth] (p\i\s) edge node[pos=0.3, fill=white, inner sep=1pt] {$1$} (p\i\t);
      		}
      	}
      	\end{scriptsize}
      	\end{tikzpicture}
      	\caption{Preference graph for \cref{thm:EFA_k+delta_clique+stars} for clique-graphs, where $e_{\ell}=\{v_i, v_j\}$.}
      	\label{fig:EFA_k+delta_clique}
      \end{figure}
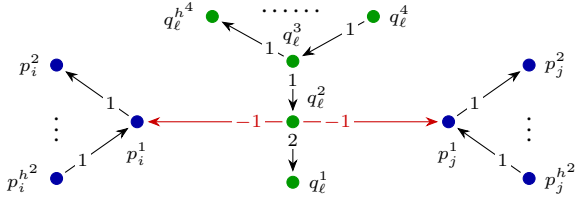

      The seat graph consists of a clique of size $\nonisolated \coloneqq 
      h^3+\binom{h}{2}$ and $|V(\hat{G})|h^2+|E(\hat{G})|h^4-k$ isolated 
      vertices.
      The construction of this \EFA instance $I$ can clearly be done in polynomial time and it holds $\maxoutdeg = 3$.
      
      It remains to show the correctness, i.e., $\hat{G}$ admits a size-$h$ clique $C$ if and only if $I$ admits an \EFA. 
      For the ``only if'' part, let $C$ be a size-$h$ clique of $\hat{G}$. 
      For each $v_i \in C$ and $\ell \in [{h^2}]$, assign agent~$p_i^{\ell}$ to some non-isolated vertex. 
      Moreover, for every $v_i, v_j \in C$ assign $q_{i,j}^1$ to some remaining non-isolated vertex. 
      The remaining agents are assigned to isolated vertices.
      This arrangement is envy-free because
      \begin{compactenum}[(i)]
      	\item for each $v_i \in C, \ell \in [h^2]$, agent $p_i^\ell$ has his maximum possible utility,
      	\item for each $v_i, v_j \in C$, agent $q_{i,j}^\ell$ for $\ell \in [h^4]\setminus \{2\}$ has no positive preference towards a non-isolated agent. 
      	Since $p_i^1, p_j^1$ are both non-isolated, also $q_{i,j}^2$ is envy-free, and
      	\item each agent corresponding to a vertex or edge in $\hat{G}\graphminus C$ has no positive preference towards a non-isolated agent. 
      \end{compactenum}		
  
  	  For the ``if'' part, we conclude by \cref{lem:EFA_non-neg_cond}, that the vertex-agents corresponding to the same vertex $v_i \in V(\hat{G})$ are all assigned to either non-isolated or isolated vertices in an \efArr. 
      This implies that at most~$h$ different sets of vertex-agents can be assigned to non-isolated vertices. 
      
      Furthermore, because $\nonisolated < h^4-1$ for $h>1$ and \cref{lem:EFA_non-neg_cond}, agents $q_{\ell}^2, \ldots, q_{\ell}^{h^4}$ for $e_\ell \in E(\hat{G})$ are always assigned to isolated vertices. 
      Since we can assign at most $h$ different sets of vertex-agents to non-isolated vertices, the remaining seats are filled with edge-agents.
      However, if some $q_{\ell}^1$ is non-isolated, then also $p_i^1$ is non-isolated for $v_i \in e_\ell$. 
      Otherwise, $q_{\ell}^2$ has an envy.      
      Therefore, we can conclude that~$\hat{G}$ admits~a size-$h$ clique if and only if $I$ admits an \efArr.
      
      \paragraph*{Stars-graphs.} 
      We provide a reduction from \Clique which is similar to the one given for the clique-graphs. 
      For each vertex $v_i \in V(\hat{G})$, create three \myemph{vertex-agents} $p_i^1, p_i^2, p_i^3$. 
      For each edge $e_\ell \in E(\hat{G})$, create $h^4$ \myemph{edge-agents} $q_{\ell}^1, \ldots, q_{\ell}^{h^4}$.
      For each $e_\ell \in E(\hat{G})$ and $v_i \in V (\hat{G})$, define the preferences as follows:
      \begin{compactitem}[--]
        \item For each $z \in [2]$, set $\sat[p_i^z](p_i^{z+1}) = 
        \sat[p_i^{z+1}](p_i^{z})= 1$.
        \item For each $z \in [h^4-1]\setminus[7]$, set 
        $\sat[q_{\ell}^{z}](q_{\ell}^{z+1}) = 1$ and 
        $\sat[q_{\ell}^1](q_{\ell}^2) = \sat[q_{\ell}^2](q_{\ell}^1)= 
        \sat[q_{\ell}^2](q_{\ell}^3) = \sat[q_{\ell}^2](q_{\ell}^4)= 
        \sat[q_{\ell}^5](q_{\ell}^3) = 
        \sat[q_{\ell}^6](q_{\ell}^4)= \sat[q_{\ell}^7](q_{\ell}^5) = 
        \sat[q_{\ell}^7](q_{\ell}^6)= \sat[q_{\ell}^{h^4}](q_{\ell}^7) = 1$. 
        \item For $e_\ell = \{v_i, v_j\}$ with $i<j$, set $\sat[q_{\ell}^{5}](p_i^1) = -1$ and $\sat[q_{\ell}^{6}](p_j^1) = -1$.
        \item The non-mentioned preferences are set to zero. 
      \end{compactitem}
      \cref{fig:EFA_k+delta_stars} shows this construction for two adjacent 
      vertices $v_i, v_j \in V(\hat{G})$.
      \begin{figure}[t]
      	\centering
      	\begin{tikzpicture}[>=stealth', shorten <= 2pt, shorten >= 2pt]
      	\def \xs {5ex}
      	\def \ys {5ex}
      	\def \qys {3.5ex}

      	\node[nodeW] (q1) {};			
      	\node[nodeW, above = \ys of q1] (q2) {};
      	\node[nodeW, above left = 0*\qys and \xs of q2] (q3) {};
      	\node[nodeW, above right = 0*\qys and \xs of q2] (q4) {};
      	\node[nodeW, above = \ys of q3] (q5) {};
      	\node[nodeW, above = \ys of q4] (q6) {};
      	
      	\node[nodeW, above = 2*\qys of q2] (q7) {};
      	\node[nodeW, above right = \qys and \xs of q7] (q8) {};
      	\node[nodeW, above left = \qys and \xs of q7] (q9) {};
      	
      	\path (q8) -- node[auto=false]{\ldots\ldots} (q9);
      	
      	\foreach \i / \pos / \v in {1/right/1, 2/above/2, 3/left/3, 4/right/4, 5/left/5, 6/right/6, 7/below/7, 8/right/8, 9/left/h^4}{
      		\node[\pos = 0pt of q\i] {\scriptsize$q_\ell^{\v}$};
      	}
      	
      	\foreach \i / \pos in {i/left, j/right}{
      		\node[nodeU, \pos = 3*\xs of q7] (p\i1) {};					
      		\node[nodeU, below = \ys of p\i1] (p\i2) {};					
      		\node[nodeU, below = \ys of p\i2] (p\i3) {};
      		
      		\foreach \j in {1,2,3}{
      			\node[\pos = 0pt of p\i\j] {\scriptsize$p_\i^{\j}$};
      		}
      	}

      	\begin{scriptsize}
      	\foreach \x / \y / \b in {1/2/25, 2/1/25, 3/2/0, 4/2/0, 3/5/0, 4/6/0, 5/7/10, 6/7/-10, 8/7/0, 7/9/0} {
      		\draw[-Stealth] (q\y) edge[bend right = \b] node[pos=0.25, fill=white, inner sep=1pt] {$1$} (q\x); 
      	}
      	\foreach \i / \b / \j in {i/left/5, j/right/6}{ 
      		\draw[-Stealth] (q\j) edge[bend \b=-50, myRed] node[pos=0.3, fill=white, inner sep=1pt] {$-1$} (p\i1);
      		
      		\foreach \s / \t in {1/2, 2/3}{ 
      			\draw[-Stealth] (p\i\s) edge[bend \b = 30] node[pos=0.3, fill=white, inner sep=1pt] {$1$} (p\i\t);
      			\draw[-Stealth] (p\i\t) edge[bend \b = 30] node[pos=0.3, fill=white, inner sep=1pt] {$1$} (p\i\s);
      		}
      	}
      	\end{scriptsize}
      	\end{tikzpicture}
      	\caption{Preference graph for \cref{thm:EFA_k+delta_clique+stars} for stars-graphs, where $e_{\ell}=\{v_i, v_j\}$ with $i<j$.}
      	\label{fig:EFA_k+delta_stars}
      \end{figure}
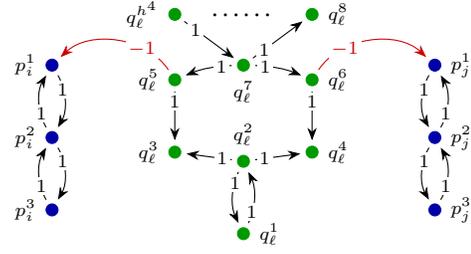
  
      The seat graph consists of $h$ stars, each with~$h$ leaves and $\binom{h}{2}$ disjoint edges, i.e., $\nonisolated\coloneqq 2h^2$, and $3|V(\hat{G})|+|E(\hat{G})|h^4-k$ isolated vertices. 
      The construction of this \EFA instance $I$ can clearly be done in polynomial time and it holds $\maxoutdeg = 3$.
      
      It remains to show the correctness, i.e., $\hat{G}$ admits a size-$h$ clique $C$ if and only if $I$ admits an \EFA. 
      For the ``only if'' part, let $C$ be a size-$h$ clique of $\hat{G}$. 
      For each $v_i \in C$, assign agent $p_i^{2}$ to the center of a star with $h$ leaves and $p_i^1, p_i3$ to its leaves. 
      Moreover, for every $v_i, v_j \in C$ with $i<y$, assign $q_{i,j}^3$ to the leaves of the star with center $p_i^2$ and $q_{i,j}^4$ to the leaves of the star with center $p_j^2$. 
      Furthermore, assign for every $v_i, v_j \in C$, agents $q_{i,j}^1,q_{i,j}^2$ to the endpoints of the same matching-edge. 
      The remaining agents are assigned to isolated vertices.
      It can be argued analogously to the case with clique-graphs that this arrangement is envy-free. 
      
      Before we show the ``if'' part we observe the following properties every \efArr has to satisfy:  
      By \cref{lem:EFA_non-neg_cond} it follows that the three vertex-agents of a vertex $v_i \in V(\hat{G})$ are all assigned to either non-isolated or isolated vertices.
      If they are non-isolated, then they are assigned consecutively, i.e., $p_i^2$ is assigned to the center of a star and $p_i^1, p_i^3$ to its leaves. 
      Since there are $h$ stars with at least two leaves, we can conclude that at most~$h$ different sets of vertex-agents are non-isolated in an \efArr.
      \begin{claim}\label{claim:EFA_k+delta_clique}
        In every \efArr exactly $h$ different sets of vertex-agents are assigned to the centers of stars with $h$ leaves.
      \end{claim}
      \begin{proof}[Proof of \cref{claim:EFA_k+delta_clique}]
        \renewcommand{\qedsymbol}{$\diamond$}
        Let~$\sigma$ be an \efArr and $e_\ell \in E(\hat{G})$ an edge. 
        We already know that we can assign at most $h$ different sets of vertex-agents to the stars. 
        
        From \cref{lem:EFA_non-neg_cond} and $2h^2 < h^4-5$ for $h>1$ it follows that agents $q_{\ell}^5, \ldots, q_{\ell}^{h^4}$ are assigned to isolated vertices.
        Moreover, if agent $q_{\ell}^3$ (resp.\ $q_{\ell}^4$) is non-isolated, then he can only be in the leaf of a star, where the center agent is~$p_i^2$ (resp.~$p_j^2$) with $e_\ell = \{v_i, v_j\}$. 
        Otherwise $q_{\ell}^5$ (resp.\ $q_{\ell}^6$) envies $p_i^2$ (resp.~$p_j^2$). 
        Furthermore, if one of the two agents $q_{\ell}^1$ or~$q_{\ell}^2$ is non-isolated, then their assigned seats are adjacent. 
        
        Combining these statements we obtain that none of the agents $q_{\ell}^1$ nor $q_{\ell}^2$ can be assigned to the center of a star with~$h$ leaves since there are not enough agents to assign to the remaining~$h$ leaves.
        This means that the remaining candidates for the star centers with $h$ leaves are the vertex-agents. 
        Hence, exactly $h$ different sets of vertex-agents are assigned to the centers of stars with $h+1$ leaves. 
      \end{proof}
      For each star with center-agent $p_i^2$, we can assign to the remaining $h-1$ leaves some incident edge-agents $q_{\ell}^3, q_{\ell}^4$. 
      Each non-isolated agent $q_{\ell}^3, q_{\ell}^4$ implies that the pair 
      $q_{\ell}^1, q_{\ell}^2$ is assigned to some matching-edge. 
      Since we have $h(h-1) = 2\binom{h}{2}$ we can choose at most 
      $\binom{h}{2}$ different sets of agents for the matching-edges. 
      Hence, the vertices of the corresponding non-isolated vertex-agents form a clique in $\hat{G}$	
    \end{proof}
  }

  \newcommand{\efakdeltapathcyclewhard}{%
    \EFA is \Wh\ wrt.\ $\nonisolated+\maxoutdeg$, even for a path- or 
    cycle-graph.  %
  }  
  \statementarxiv{theorem}{thm:EFA_k+delta_path-cycle}{\efakdeltapathcyclewhard}
  \appendixproofwithstatement{thm:EFA_k+delta_path-cycle}{\efakdeltapathcyclewhard}{
    \begin{proof}
      We prove this via a parameterized reduction from the \Wh\ problem 
      \Clique, parameterized by the solution size. 
      Let $(\hat{G},h)$ be an instance of \Clique. 
      
      For each vertex $v_i \in V(\hat{G})$, create $h(h-1)$ \myemph{vertex-agents} named $p_i^1, \ldots, p_i^{h(h-1)}$. 
      For each edge $e_\ell \in E(\hat{G})$, create~$h^4+1$ \myemph{edge-agents} $q_\ell^0, q_\ell^1, \ldots, q_\ell^{h^4}$. 
      For each $e_\ell \in E(\hat{G})$ and $v_i \in V (\hat{G})$, define the preferences as follows (see \cref{fig:EFA_k+delta_path-cycle} for the corresponding preference graph):
      \begin{compactitem}[--]
        \item For each $z \in \{0\} \cup [h-2]$ and $s \in [h-1]$, set 
        $\sat[p_i^{zh+s}](p_i^{zh+s+1}) = 
        \sat[p_i^{zh+s+1}](p_i^{zh+s}) = s$.
        
        \item For each $z \in [h-2]$, set $\sat[p_i^{(z-1)h+2}](p_i^{zh+2}) 
        = 
        \sat[p_i^{zh+2}](p_i^{(z-1)h+2}) = 1$.
        
        \item For each $z \in [h^4]\setminus[3]$, set 
        $\sat[q_{\ell}^{z-1}](q_{\ell}^{z}) = 1$ and 
        $\sat[q_{\ell}^{h^4}](q_{\ell}^{3}) = 
        \sat[q_{\ell}^{1}](q_{\ell}^{0}) = 
        \sat[q_{\ell}^{2}](q_{\ell}^{0}) = 1$. For each $z \in [3]$, set 
        $\sat[q_{\ell}^{z}](q_{\ell}^{z-1}) = 1$. 
        
        \item For each $v_i \in e_\ell$, $z \in \{0\}\cup [h-2]$, set 
        $\sat[q_{\ell}^{2}](p_i^{zh+2}) = -1$.			
        \item The non-mentioned preferences are set to zero.
      \end{compactitem}
  		
  		\begin{figure}[t]
  			\centering
  			\begin{tikzpicture}[>=stealth', shorten <= 2pt, shorten >= 2pt]
  			\def \xs {3ex}
  			\def \ys {6ex}
  			\def \qys {3ex}
  			\def \qxs {6ex}

  			\node[nodeW] (q2) {};
  			\node[nodeW, below left = 1.7*\ys and \xs of q2] (q0) {};
  			\node[nodeW, below right = 1.7*\ys and \xs of q2] (q1) {};
  			\node[nodeW, above = \ys of q2] (q3) {};
  			\node[nodeW, above right = \qys and \qxs of q3] (q4) {};
  			\node[nodeW, above left = \qys and \qxs of q3] (q5) {};
  			
  			\path (q4) -- node[auto=false]{\ldots\ldots} (q5);
  			
  			\foreach \i / \pos / \v in {0/below/0, 1/below/1, 3/above/3, 4/right/4, 5/left/h^4}{
  				\node[\pos = 0pt of q\i, inner sep=0.5pt] {\scriptsize$q_\ell^{\v}$};
  			}
  			\node[above right = -3pt and 0pt of q2] {\scriptsize$q_\ell^{2}$};
  			
  			\foreach \i / \pos / \l / \r in {i/left/h-1/2, j/right/2/h-1}{
  				\node[nodeU, above \pos = 0.5*\ys and 2.2*\xs of q2] (p\i0h1) {};
  				\node[below = 1pt of p\i0h1] {\scriptsize $p_\i^1$};
  				
  				\node[nodeU, below = \ys of p\i0h1] (p\i1h1) {};
  				\node[below = 1pt of p\i1h1] {\scriptsize $p_\i^{h+1}$};
  				
  				\node[nodeU, below = 1.35*\ys of p\i1h1] (p\i2h1) {};
  				\node[below = 1pt of p\i2h1] {\scriptsize $p_\i^{(h-2)h+1}$};
  				
  				\foreach \x / \n / \m in {0h/2/h, 1h/h+2/2h, 2h/(h-2)h+2/h(h-1)}{
  					\node[nodeU, \pos = 2*\xs of p\i\x1] (p\i\x2) {};
  					\node[below \pos = 2pt and -2pt of p\i\x2] {\scriptsize $p_\i^{\n}$};
  					
  					\node[nodeU, \pos = 3.5*\xs of p\i\x2] (p\i\x3) {};
  					\node[below = 2pt of p\i\x3] {\scriptsize $p_\i^{\m}$};
  					
  					\path (p\i\x2) -- node[auto=false]{\color{gray}$\overset{\l, \ldots, \r}{\ldots\ldots\ldots}$} (p\i\x3); 
  				}
  				\path (p\i1h2) -- node[pos=0.42, auto=false]{\vdots} (p\i2h2);
  				\path (p\i1h2) -- node[pos=0.58, auto=false]{\vdots} (p\i2h2);
  			}

  			\begin{scriptsize}
  			\foreach \x / \y / \b in {0/1/15, 1/0/15, 2/1/0, 3/2/0, 3/4/-10, 5/3/-10, 2/0/0} {
  				\draw[-Stealth] (q\x) edge[bend left = \b]  node[pos=0.25, fill=white, inner sep=1pt] {$1$} (q\y); 
  			}
  			
  			\foreach \i / \b in {i/right, j/left}{					
  				\draw[-Stealth] (q2) edge[bend \b=60, myRed] node[pos=0.25, fill=white, inner sep=1pt] {$-1$} (p\i0h2);
  				\draw[-Stealth] (q2) edge[bend \b=12, myRed] node[pos=0.25, fill=white, inner sep=1pt] {$-1$} (p\i1h2);
  				\draw[-Stealth] (q2) edge[bend \b=-17, myRed] node[pos=0.26, fill=white, inner sep=1pt] {$-1$} (p\i2h2);
  			}
  			
  			\foreach \i / \b in {i/left, j/right}{					
  				\foreach \x in {0h,1h,2h}{	
  					\draw[-Stealth] (p\i\x1) edge[bend left = 15] node[pos=0.3, fill=white, inner sep=1pt] {$1$} (p\i\x2);
  					\draw[-Stealth] (p\i\x2) edge[bend left = 15] node[pos=0.3, fill=white, inner sep=1pt] {$1$} (p\i\x1);
  				}
  				\draw[-Stealth] (p\i0h2) edge[bend \b = 10] node[pos=0.3, fill=white, inner sep=1pt] {$1$} (p\i1h2);
  				\draw[-Stealth] (p\i1h2) edge[bend \b = 10] node[pos=0.3, fill=white, inner sep=1pt] {$1$} (p\i0h2);
  			}
  			\end{scriptsize}
  			\end{tikzpicture}
  			\caption{Preference graph for \cref{thm:EFA_k+delta_path-cycle} for path- and cycle-graphs, where $e_{\ell}=\{v_i, v_j\}$.}
  			\label{fig:EFA_k+delta_path-cycle}
  		\end{figure}
      
      The seat graph $G$ consists of a path (resp.\ cycle) with $\nonisolated\coloneqq h^2(h-1)+h(h-1)$ vertices and $|V(\hat{G})|h(h-1)+|E(\hat{G})|(h^4+1)-k$ isolated vertices.
      Clearly, this instance~$I$ of \EFA can be constructed in polynomial time and it holds $\maxoutdeg = 2h$.
      
      It remains to show that $\hat{G}$ admits a size-$h$ clique if and only if $I$ admits an \efArr.		
      We start proving the ``only if'' part. 
      Let $C$ be a size-$h$ clique. 
      For the path we begin assigning at one of the endpoints, for the cycle we can begin at any non-isolated vertex.
      For each edge $e_\ell = \{v_i,v_j\} \in C$ and some (not yet used) $z,z' \in [h-1]$, assign agents $p_i^{zh}, \ldots, p_i^{(z-1)h+1}, q_{\ell}^0, q_\ell^1, p_j^{(z'-1)h+1}, \ldots, p_j^{z'h}$ in this order to the path (resp.\ cycle). 
      The remaining agents are assigned to isolated vertices.	
      
      This arrangement is envy-free because of the following: 
      \begin{compactenum}[(i)]
        \item Every vertex-agents $p_i^z$ with $v_i \in C, z \in [h(h-1)]$ is envy-free because $p_i^z$ has his maximum possible utility.
        
        \item The vertex-agents $p_i^z$ with $v_i \notin C, z \in [h(h-1)]$ are envy-free since there is no non-isolated agent towards which $p_i^z$ has a positive preference. 
        
        \item Every edge-agent $q_\ell^0, q_{\ell}^1$ with $e_\ell \in C$ is envy-free since he has his maximum possible utility of one.
        
        \item The remaining isolated edge-agents either have no non-isolated agent towards which they have a positive preference, or the corresponding agents $q_{\ell}^0, q_{\ell}^1$ are assigned in such a way that $q_{\ell}^2$ does not envy his neighbors. 
      \end{compactenum}	
      
      Before we prove the ``if'' part, we observe the following properties an \efArr has to satisfy. 
      \begin{claim}\label{claim:EFA_k+delta_path-cycle}
        Every envy-free arrangement~$\sigma$ satisfies:
        \begin{compactenum}[(1)]
          \item\label{eq:EFA_k+delta_path-cycle-no-q} For each edge $e_\ell \in E(\hat{G})$, agents $q_{\ell}^2, \ldots, q_{\ell}^{h^4}$ are always assigned to 
          isolated vertices.
          
          \item\label{eq:EFA_k+delta_path-cycle-vertex-gadget} If a vertex-agent $p_i^z$ with $v_i \in V(\hat{G}), z \in [h^2]$ is assigned to the path (resp.\ cycle), then all agents from $\{p_i^1, \ldots, p_i^{h^2}\}$ are assigned to the path (resp.\ cycle).   
          Moreover, if such a set of agents is non-isolated, then the seats of $p_i^{zh+s}, p_i^{zh+s+1}$ for $z\in \{0\}\cup [h-2]$ and $s\in [h-1]$ are adjacent. 
          
          \item\label{eq:EFA_k+delta_path-cycle-edge-gadget1} If $q_{\ell}^0$ or $q_{\ell}^1$ is non-isolated for some $e_\ell \in E(\hat{G})$, then both $q_{\ell}^0$ and $q_{\ell}^1$ are non-isolated and their seats are adjacent. 
          
          \item\label{eq:EFA_k+delta_path-cycle-edge-gadget2} If a pair $q_{\ell}^0,q_{\ell}^1$ is non-isolated for $e_\ell = \{v_i,v_j\}\in E(\hat{G})$, then both $q_{\ell}^0$ and $q_{\ell}^1$ are adjacent in~$\sigma$ to vertex-agents corresponding to $v_i$ or $v_j$. 
          In particular, $q_{\ell}^0,q_{\ell}^1$ are adjacent to some $p_i^{zh+1}$ or $p_j^{zh+1}$ with $z \in \{0\} \cup [h-2]$, where the next $h$ seats on the path (resp.\ cycle) are assigned to $p_i^{zh+1}, \ldots, p_i^{(z+1)h}$ or $p_j^{zh+1}, \ldots, p_j^{(z+1)h}$. 
        \end{compactenum}
      \end{claim}
      \begin{proof}[Proof of \cref{claim:EFA_k+delta_path-cycle}]
        \renewcommand{\qedsymbol}{$\diamond$}
        \begin{compactenum}[(1)]
          \item Suppose an agent $q_{\ell}^z$ for some $e_\ell \in E(\hat{G}), z \in [h^4]\setminus\{1\}$ is assigned to the path (resp.\ cycle). 
          Then, by \cref{lem:EFA_non-neg_cond} all agents $q_{\ell}^2, \ldots, q_{\ell}^{h^4}$ have to be assigned to the path (resp.\ cycle).
          However, since $\nonisolated < h^4-1$ for $h>1$, there are not enough vertices in the path (resp.\ cycle).
          
          \item By the previous statement we know that for each edge $e_\ell 
          \in E(\hat{G})$, agent $q_\ell^2$ is isolated. 
          Hence, the first part of this statement follows now from \cref{lem:EFA_non-neg_cond}.
          If the set of agents for a vertex $v_i \in V(\hat{G})$ is non-isolated, then the seats of $p_i^{zh+s}, p_i^{zh+s+1}$ for $z\in \{0\}\cup [h-2]$, $s\in [h-1]$ are adjacent because $\sat[p_i^{zh+s}](p_i^{zh+s+1})$ is the largest positive preference of $p_i^{zh+s}$. 
          Hence, if $p_i^{zh+s}$ is not assigned next to $p_i^{zh+s+1}$, he will have an envy.
          
          \item This follows from \cref{lem:EFA_non-neg_cond}. 
          
          \item By the previous statement we know that agents $q_{\ell}^0, q_{\ell}^1$ are both either isolated or non-isolated and in the latter case, their seats are adjacent. 
          If $q_{\ell}^0$ (resp.\ $q_{\ell}^1$) is not adjacent to a vertex-agent, then he has to be adjacent to an edge-agent, which implies that $q_\ell^2$ envies $q_{\ell}^0$ (resp.\ $q_{\ell}^1$). 
          To make $q_\ell^2$ envy-free, the next agent but one from $q_{\ell}^0$ (resp.\ $q_{\ell}^1$) has to be some $p_i^{zh+2}$ or $p_j^{zh+2}$. 
          Combining this with Statement~\eqref{eq:EFA_k+delta_path-cycle-vertex-gadget} proves this statement. \qedhere
        \end{compactenum}
      \end{proof}
      
      From \cref{claim:EFA_k+delta_path-cycle}\eqref{eq:EFA_k+delta_path-cycle-vertex-gadget} it follows that at most $h$ different size-$h(h-1)$ sets of vertex-agents can be assigned to the path (resp.\ cycle). 
      By the number of non-isolated vertices, at least $h(h-1)$ edge-agents are are assigned to non-isolated vertices. 
      Let $E'$ denote the edges from $E(\hat{G})$ which correspond to the non-isolated edge-agents.
      We claim that $E'$ corresponds to a size-$h$ clique.
      To this end, let $V'=\{v,w\mid \{v,w\} \in E'\}$.
      By \cref{claim:EFA_k+delta_path-cycle}\eqref{eq:EFA_k+delta_path-cycle-edge-gadget1}, we have that $|E'|\ge \binom{h}{2}$.
      By \cref{claim:EFA_k+delta_path-cycle}\eqref{eq:EFA_k+delta_path-cycle-edge-gadget2}, for each $e_{\ell}\in E'$ the two size-$h(h-1)$ sets vertex-agents which correspond to the endpoints of $e_{\ell}$ must also be non-isolated.
      This means that $|V'| \ge h$. 
      Since we can only assign~$h$ many such sets of vertex-agents to the non-isolated vertices, we have that $|V'|\le h$, implying that $|V'|=h$, and $(V',E')$ is a clique of size $h$, as desired. 
    \end{proof}
  }

  \looseness=-1
  \section{\ESAf}\label{sec:esa}
  \appendixsection{sec:esa}
  In this section we consider the last concept exchange stability. 
  We first show that an \esArr always exists for clique-graphs and non-negative preferences. %
  \begin{obs} \label{obs:ESA_clique_nonneg}
    For clique-graphs and non-negative preferences, \ESA is polynomial-time solvable.
  \end{obs}
  \begin{proof}
    With non-negative preferences, no agent assigned to non-isolated vertices envies an isolated agent.
    Moreover, each agent is indifferent towards the seats in the clique. 
    Hence, they will not form any exchange-blocking pair.
  \end{proof}
  
\ifarxiv\else
  \noindent In general, the problem remains intractable, even for clique- or matching-graphs. 
\fi 
  Particularly, it is unlikely that an \fpt\ algorithm wrt.\ $\nonisolated$ exists, implied by the following.
  
  \newcommand{\esakcliquematchingwhard}{%
    \ESA is \Wh\ wrt.\ $\nonisolated$, even for a \mbox{clique-,} stars-, or matching-graph. %
  }
  \newcommand{\esakcliqueclaimxtwo}{%
    In every \esArr, agent $x_2$ is assigned to an isolated vertex. 
  }
\statementarxiv{theorem}{thm:ESA_k_clique+matching}{\esakcliquematchingwhard}
  {
\ifarxiv
	\begin{proof} 
		\textbf{Clique-graphs. }
\else
\ifshort
	\begin{proof}[Proof sketch]
		We show the case with clique-graphs. 
\else
	\begin{proof}[Proof sketch]
		We show the case with clique-graphs. The proof for matching-graphs (and hence stars-graphs) is deferred to the appendix.
		
\fi 
\fi        
      We provide a parameterized reduction from \IS parameterized by the solution 
      \ifshort size~$h$~\cite{DF13}. %
      \else
      size~\cite{DF13}. 
      \fi 
      Given an instance $(\hat{G}, h)$ of \IS, we create for each vertex~$v_i \in V(\hat{G})$ a \myemph{vertex-agent}~$p_i$, and two special agents~$x_1$ and~$x_2$. 
      For each two~$v_i, v_j \in V(\hat{G})$, the preferences are defined as follows (see Figure~\ref{fig:ESA_k_clique} for the corresponding preference graph):
      If $\{v_i,v_j\} \in E(\hat{G})$, set 
      $\sat[p_i](p_j)=\sat[p_j](p_i)=-1$.
      Finally, set $\sat[p_i](x_2) = h$,
      $\sat[x_1](x_2) = -h$,
      $\sat[x_2](x_1) = h$,
      $\sat[x_1](p_i) = 1$,	and		
      $\sat[x_2](p_i) = -1$.
      The non-mentioned preferences are zero.
      \begin{figure}[t]
      	\centering
      	\begin{tikzpicture}[>=stealth', shorten <= 2pt, shorten >= 2pt]
	      	\def \xs {10ex}
	      	\def \ys {10ex}
	      	\def \ss {2.5ex}
	      	
	      	\node[nodeW] (x1) {};
	      	\node[nodeW, right = 1.8*\xs  of x1] (x2) {};
	      	\node[nodeU, below left = 0.4*\ys and 0.5*\xs of x1] (p1) {};
	      	\node[nodeU, below right = 0.4*\ys and 0.5*\xs of x2] (p2) {};

	      	\foreach \x in {1, 2} {
	      		\node[above = 0pt of x\x] {\scriptsize$x_{\x}$};
	      	}
	      	\node[below = 0pt of p1] {\scriptsize$p_{1}$};
	      	\node[below = 0pt of p2] {\scriptsize$p_{\hat{n}}$};
	      	\path (p1) -- node[auto=false]{\ldots\ldots} (p2);

	      	\node[ellipse, draw = gray, minimum width = 4.4*\xs, minimum height = 0.5*\ys, label={left:$\hat{G}$}] (ell) at (0.94*\xs,-0.55*\ys) {};

	      	\begin{scriptsize}	      	
		      	\foreach \x / \b / \bb / \p /\pp / \bbb in {1/15/5/0.32/0.15/15,2/10/45/0.12/0.31/0} {
		      		\draw[-Stealth] (x1) edge[black, bend right=\bbb]  node[pos=\p, fill=white, inner sep=1pt] {$1$} (p\x);
		      		\draw[-Stealth] (x2) edge[bend right=\b, myRed]  node[pos=\pp, fill=white, inner sep=1pt] {$-1$} (p\x);
		      		\draw[-Stealth] (p\x) edge[bend right=\bb, black]  node[pos=\pp, fill=white, inner sep=1pt] {$h$} (x2);
		      	}
		      	
		      	\draw[-Stealth] (x1) edge[bend left=30, myRed]  node[pos=0.22, fill=white, inner sep=1pt] {$-h$} (x2);
		      	\draw[-Stealth] (x2) edge[bend left=-15, black]  node[pos=0.25, fill=white, inner sep=1pt] {$h$} (x1);
		      	
		      	\node[nodeU, right = \xs of p1] (pi) {};
		      	\node[nodeU, left = \xs of p2] (pj) {};
		      	\draw[Stealth-Stealth] (pi) edge[bend right=15, gray] node[below, inner sep=1pt] {\begin{tabular}{c}$-1$ if adjacent, $0$ otherwise \end{tabular}} (pj);
	      	\end{scriptsize}
      	\end{tikzpicture}
      	\caption{Preference graph for \cref{thm:ESA_k_clique+matching} for clique-graphs.}
      	\label{fig:ESA_k_clique}
    \end{figure}
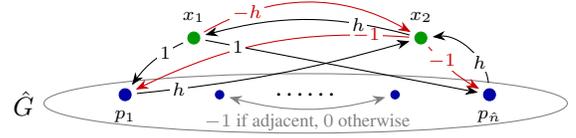       
      The seat graph~$G$ consists of a clique of size $\nonisolated \coloneqq h+1$ and $|V(\hat{G})|+2-\nonisolated$ isolated vertices.
      
      It remains to show that $\hat{G}$ admits a size-$h$ independent set if and only if the constructed instance admits an \esArr.
      First, we observe the following. %
      \begin{claim}%
        \label{claim:ESA_k_clique}
        \esakcliqueclaimxtwo
      \end{claim}
        \begin{proof}[Proof of 
        \cref{claim:ESA_k_clique}]
        \renewcommand{\qedsymbol}{$\diamond$}
        Towards a contradiction, suppose that $\sigma$ is exchange-stable where $x_2$ is non-isolated. %
        If $x_1$ is isolated, then $\util[x_2]{\sigma} =-h<0$ and 
        $\util[x_1]{\swap{x_1}{x_2}}=h>0$, implying that $\{x_1,x_2\}$ is an exchange-blocking pair, a contradiction.
        If $x_1$ is non-isolated, then $\util[x_1]{\sigma}=-1$. 
        Since there is a vertex-agent $p_i \in \agents$ assigned to an isolated vertex and $\util[p_i]{\swap{p_i}{x_1}}>0$, agents $x_1$ and $p_i$ form an exchange-blocking pair, a contradiction.
      \end{proof}

      Hence, at least~$h$ vertex-agents have to be assigned to the non-isolated vertices in an \esArr. 
      If one of these non-isolated vertex-agents~$p_i$ has negative utility, then~$p_i$ envies every isolated agent.
      Depending on whether~$x_1$ is assigned to an isolated vertex, agent~$p_i$ forms an exchange-blocking pair with~$x_1$ or with $x_2$. 		
      Therefore,~$I$ admits an \esArr if and only if every non-isolated vertex-agent has non-negative utility. 
      By Claim~\ref{claim:ESA_k_clique} this is equivalent to each pair of non-isolated vertex-agents is non-adjacent, i.e., $\hat{G}$ admits a size-$h$ independent set. 
      \appendixcontinue{thm:ESA_k_clique+matching}{}{\esakcliquematchingwhard}{
      	\iflong
        We continue with the proof.
        \fi 
        
        \paragraph*{Matching-graphs.} 
        We provide a parameterized reduction from the \Wh\ problem \Clique parameterized by the solution size. 
        Let $(\hat{G}, h)$ be an instance of \Clique and $\hat{n}\coloneqq|V(\hat{G})|$. 
        For each vertex $v_i \in V(\hat{G})$, we create~$2h^2$ \myemph{vertex-agents}~$p_i^1, \ldots, p_i^{h^2}, \tilde{p}_i^1, \ldots, \tilde{p}_i^{h^2}$.
        For each edge $e_\ell \in E(\hat{G})$, create one \myemph{edge-agent}~$q_{\ell}$. 
        The preferences are defined as follows:
        \begin{compactitem}[--]
          \item For each edge $e_\ell \in E(\hat{G})$: 
          \begin{compactitem}
            \item For each $e_{\ell'} \in E(\hat{G}) \setminus \{e_\ell\}$, set $\sat[q_{\ell}](q_{\ell'}) = -1$.
            
            \item For each $v_{i} \in V(\hat{G}), z \in [h^2]$, set $\sat[q_{\ell}](p_{i}^z) = \sat[q_{\ell}](\tilde{p}_{i}^z) = -i-1$.
          \end{compactitem}
          \item For each vertex $v_i \in V(\hat{G})$ and $z, z' \in [h^2]$ with $z\neq z'$: 
          \begin{compactitem}
            \item  $\sat[p_i^z](\tilde{p}_i^{z'})\!=\!1$, $\sat[p_i^z](\tilde{p}_i^{z}) = -1$, and $\sat[p_i^z](p_i^{z'})\!=\!-2$.
            \item $\sat[\tilde{p}_i^z](p_i^{z})\!=\!2$, $\sat[\tilde{p}_i^z](\tilde{p}_i^{z'})\!=\!-2$, and $\sat[\tilde{p}_i^z](p_i^{z'})\!=\!-1$.

          \end{compactitem}
          
          \item For each pair of vertices $v_i, v_j \in V(\hat{G})$ and 
          $z, z' \in [h^2]$:
          \begin{compactitem}
            \item $\sat[p_i^z](p_j^{z'}) = 
            \sat[p_i^z](\tilde{p}_j^{z'}) = -h^2-1-j$,
            \item $\sat[\tilde{p}_i^z](p_j^{z'}) = 
            \sat[\tilde{p}_i^z](\tilde{p}_j^{z'}) = -h^2-1-j$.
          \end{compactitem}
          \item For each $v_i \in V(\hat{G}), z \in [h^2]$ and $e_\ell \in E(\hat{G})$:
          \begin{compactitem}
            \item $\sat[p_i^z](q_{\ell}) = h^2-1$ if $z \neq 1$ and 
            $v_i \notin e_\ell$,
            \item $\sat[p_i^1](q_{\ell}) = 1$ if $v_i \in e_\ell$,
            \item $\sat[\tilde{p}_i^z](q_{\ell}) = -h^2-1$.
          \end{compactitem}
        \end{compactitem}
        To complete the construction of our \ESA instance~$I$, we create a seat graph~$G$ with $\nonisolated/2 \coloneqq h^3 + \binom{h}{2}/2$ disjoint edges (i.e., stars with one leaf) and $2h^2|V(\hat{G})|+|E(\hat{G})|-k$ isolated vertices. 
        This construction of $I$ can be done in polynomial time and $\nonisolated \coloneqq 2h^3 + \binom{h}{2}$. 
        
        It remains to prove the correctness. 
        That is, we show that~$\hat{G}$ has a size-$h$ clique iff.~$I$ has an exchange-stable arrangement.
        To this end, we say that an agent~$p$ is \myemph{matched} to another agent~$q$ if they are assigned to the endpoints of the same edge in the seat graph; we also say that~$p$ is \myemph{matched} in this case.
        
        To prove the ``only if'' part, assume $C$ is a size-$h$ clique in~$\hat{G}$. We give an arrangement $\sigma$ as follows.
        For each $v_i \in C$ we assign the corresponding vertex-agents in the following way:
        For $\ell \in [h^2-1]$ assign $p_i^\ell, \tilde{p}_i^{\ell+1}$ as well as $p_i^{h^2}, \tilde{p}_i^{1}$ to the two endpoints of the same edge. 
        Moreover, since $C$ induces a clique, there exists an edge between every two vertices in $C$. 
        Since there are~$\binom{h}{2}$ such edges, we assign the corresponding edge-agents to the remaining endpoints of edges in $G$. 
        All other agents are assigned to isolated vertices.        
        This arrangement $\sigma$ is exchange-stable because of the following:
        \begin{compactenum}
          \item Since edge-agents have a negative preference towards every other agent, the edge-agents in isolated 
          vertices are never part of an exchange-blocking pair. 
          
          \item The matched edge-agents do not form blocking. First, we argue that an edge-agent does not form a blocking with other matched agents. Edge-agents prefer other edge-agents more than vertex-agents.
          Since each non-isolated edge-agent is matched to another edge-agent he does not envy another non-isolated agent. 
          Next, we show he doesn't form blocking with an isolated agent. 
          An edge-agent envies isolated agents. 
          The only two agents, which could form an exchange-blocking pair with an edge-agent~$q_\ell$ are the vertex-agents $\tilde{p}_i^1, \tilde{p}_j^1$ where $e_\ell = \{i,j\}$. 
          However, $\tilde{p}_i^1$ and $\tilde{p}_j^1$ are also non-isolated.
          Therefore, no edge-agent is part of an exchange-blocking pair.
          
          \item For each $v_i \in V(\hat{G}) \setminus C, z \in [h^2]$, agents $p_i^z, \tilde{p}_i^z$ have no neighbor with positive preference assigned to the endpoints of an edge. Hence, they are envy-free and do not form a blocking.
          
          \item For each $v_i \in C, z \in [h^2]$, agent $p_i^z$ is envy-free because he has his maximum possible utility. 
          Agent $\tilde{p}_i^z$ is not envy-free, but cannot find a partner to form an exchange-blocking pair.
        \end{compactenum}	
        
        For the ``if'' part, assume that $\sigma$ is an \esArr. 
        We observe some properties that $\sigma$ must satisfy.  
        \begin{claim} \label{claim:ESA_k_matching-vertex}
          For each $(i, z) \in [\hat{n}]\times [h^2]$, the following holds in every \esArr:
          \begin{compactenum}[(i)]
            \item If $p_i^z$ is matched, then $\tilde{p}_i^z$ as well. 
            \item If $\tilde{p}_i^z$ is matched, then he is matched to $p_i^{z'}$ for some $z' \in [h^2]\setminus \{z\}$ or all agents from $\{p_i^1, \ldots, p_i^{z-1},p_i^{z+1}, \ldots, p_i^{h^2} \}$ are matched. 
            \item If $p_i^{z}$ is matched to $\tilde{p}_i^z$, then all agents from $\{p_i^1, \ldots, p_i^{h^2} \}$ are matched. 
            \item If $p_i^{z}$ is matched to $\tilde{p}_i^{z'}$ for some $z' \in [h^2]\setminus \{z\}$, then $\tilde{p}_i^{z}$ is also matched. 
            \item\label{claim:ESA_k_matching-vertex(allornone)} Either all agents from $\{p_i^1, \ldots, p_i^{h^2}, \tilde{p}_i^1, \ldots, \tilde{p}_i^{h^2} \}$ or none of them are matched.
            \item If $p_i^z$ or $\tilde{p}_i^{z}$ is matched, then he is not matched to a vertex-agent of a different vertex $v_j \in V(\hat{G})$.
          \end{compactenum}
        \end{claim}
        \begin{proof}[Proof of 
          \cref{claim:ESA_k_matching-vertex}]
          \renewcommand{\qedsymbol}{$\diamond$}
          We show each case by contradiction. Suppose that there is an \esArr~$\sigma$ 
          and $(i, z) \in [\hat{n}]\times [h^2]$, where
          \begin{compactenum}[(i)]
            \item $p_i^z$ is matched to an agent~$p$ and $\tilde{p}_i^z$ is assigned to an isolated vertex. 
            Then $\util[p]{\sigma} < 0$ since each agent (except $\tilde{p}$) has a negative preference to $p_i^z$, i.e., $(p, \tilde{p}_i^z)$ is an exchange-blocking pair, a contradiction.
            
            \item $\tilde{p}_i^z$ is matched to $p \notin \{p_i^1, \ldots, p_i^{z-1}, p_i^{z+1}, \ldots, p_i^{h^2}\}$ and some $p_i^{z'}, z' \in [h^2] \setminus \{z\}$ is assigned to an isolated vertex. 
            Then $\util[p]{\sigma} < 0$, i.e., $(p, p_i^{z'})$ is an exchange-blocking pair, a contradiction.
            
            \item $p_i^{z}$ is matched to $\tilde{p}_i^z$ and some $p_i^{z'}, z' \in [h^2] \setminus \{z\}$ is assigned to an isolated vertex. 
            Then $\util[p_i^z]{\sigma} = -1$, i.e., $(p_i^z, p_i^{z'})$ is an exchange-blocking pair, a contradiction.
            
            \item $p_i^{z}$ is matched to some $\tilde{p}_i^{z'}, z' \in [h^2] \setminus \{z\}$ and $\tilde{p}_i^{z}$ is assigned to an isolated vertex. 				
            Then $\util[\tilde{p}_i^{z'}]{\sigma} = -1$, i.e., $\tilde{p}_i^{z'}$ and $\tilde{p}_i^z$ is an exchange-blocking pair, a contradiction. 
            
            \item This follows from the previous statements. 
            
            \item $x_i^z$ is matched to $y_j^{z'}$ with $x,y \in 
            \{p,\tilde{p}\}, i \neq j$. 
            Without loss of generality, assume $i < j$. 
            By the previous statement we know that all agents from $\{p_j^1, \ldots, p_j^{h^2}, \tilde{p}_j^1, 
            \ldots, \tilde{p}_j^{h^2} \}$ are matched.
            Since we created an even number of agents for each vertex in $V(\hat{G})$, there is an agent 
            $\tilde{y}_j^{z''}, \tilde{y} \in \{p,\tilde{p}\}, z'' \in 
            [h^2]$, which is matched to $p \notin 
            \{p_j^1, \ldots, p_j^{h^2}, \tilde{p}_j^1, \ldots, 
            \tilde{p}_j^{h^2} \}$.
            Agent~$y_j^{z'}$ envies $p$ since he would rather be matched to a vertex-agent of his own vertex than to anybody else.
            Since we assumed $i < j$, we obtain the following:			
            If $p$ is an edge-agent, then $\util[p]{\sigma} = -j-1 < -i-1 = \util[p]{\swap{p}{y_j^{z'}}}$.			
            If $p$ is a vertex-agent, then $\util[p]{\sigma} = -h^2-j-1 <  -i-1 = \util[p]{\swap{p}{y_j^{z'}}}$.
            Hence, $(p, y_j^{z'})$ is an exchange-blocking pair, a contradiction. \qedhere
          \end{compactenum}
        \end{proof}
        
        \begin{claim} \label{claim:ESA_k_matching-edge}
          In an \esArr~$\sigma$ the following holds for each edge $e_\ell = \{v_i,v_j\} \in E(\hat{G})$:
          \begin{compactenum}
            \item If $q_{\ell}$ is matched, then he is not matched to a vertex-agent.
            \item If $q_{\ell}$ is matched, then all agents from $\{p_i^1, 
            \ldots, p_i^{h^2}, \tilde{p}_i^1, \ldots, 
            \tilde{p}_i^{h^2}, p_j^1, \ldots, p_j^{h^2}, \tilde{p}_j^1, \ldots, \tilde{p}_j^{h^2} \}$ are 
            matched.
          \end{compactenum}
        \end{claim}
        \begin{proof}[Proof of 
          \cref{claim:ESA_k_matching-edge}]
          \renewcommand{\qedsymbol}{$\diamond$}
          Suppose towards a contradiction, that there are an \esArr~$\sigma$ 
          and an edge $e_\ell = \{v_i,v_j\} \in E(\hat{G})$, where
          \begin{compactenum}

            \item $q_{\ell}$ is matched to $x \in \{p_{i'}^1, \ldots, 
            p_{i'}^{h^2}, \tilde{p}_{i'}^1, 
            \ldots, \tilde{p}_{i'}^{h^2}\}$  for some $v_{i'} \in 
            V(\hat{G})$.				
            By the previous claim we know that all agents from $\{p_{i'}^1, \ldots, p_{i'}^{h^2}, \tilde{p}_{i'}^1, 
            \ldots, \tilde{p}_{i'}^{h^2} \}$ are matched and no two vertex-agents corresponding to different vertices in 
            $\hat{G}$ are matched together. 
            Since $2h^2$ is even, there is another agent $y \in \{p_{i'}^1, \ldots, p_{i'}^{h^2}, 
            \tilde{p}_{i'}^1, \ldots, \tilde{p}_{i'}^{h^2} \}\setminus\{x\}$ matched to an edge-agent.
            Hence, $\util[y]{\sigma} = -h^2-1 < -2 \leq 
            \util[y]{\swap{y}{q_{\ell}}}$. 
            
            \item $q_{\ell}$ is matched to $p$ and one of the set of agents $\{p_i^1, \ldots, p_i^{h^2}, \tilde{p}_i^1, \ldots, \tilde{p}_i^{h^2} \}$ or $\{p_j^1, \ldots, p_j^{h^2}, \tilde{p}_j^1, \ldots, 
            \tilde{p}_j^{h^2} \}$ is assigned to isolated vertices (by the previous claim either all agents of the same vertex in $\hat{G}$ or none of them are matched).
            Since by the previous statement $q_{\ell}$ is not matched to any vertex-agent, we know $p \neq \tilde{p}_i^1, \tilde{p}_j^1$, i.e., $\util[p]{\sigma} < 0$. 
            If $\tilde{p}_i^1$ or $ \tilde{p}_j^1$ is isolated, then $(p, \tilde{p}_i^1)$ resp.\ $(p, \tilde{p}_j^1)$ is an exchange-blocking pair.\qedhere
          \end{compactenum}
        \end{proof}
        By Claim~\ref{claim:ESA_k_matching-vertex} and $2h^3 + \binom{h}{2} < 2h^3 + h^2$ we obtain that at most~$h$ different sets of vertex-agents can be matched. 		
        Suppose we choose~$\binom{h}{2}+1$ edge-agents. 
        Then these edges are incident to at least $h+1$ vertices. 
        Since for each matched edge-agent, we also have to match the set of vertex-agents of both his endpoints, by \cref{claim:ESA_k_matching-edge}, we cannot match $\binom{h}{2}+1$ edge-agents. 
        Hence, at most $\binom{h}{2}$ edge-agents can be matched. 
        Hence, we need to assign exactly~$h$ different sets of vertex-agents and $\binom{h}{2}$ edge-agents to the endpoints of the edges.
        Furthermore, these chosen edge-agents need to be incident to the vertex-agents in $\hat{G}$, which means that they form a clique in~$\hat{G}$.	
        
        \paragraph*{Stars-graphs.} This follows from the previous case since a matching-graph is also a stars-graph. 
        \iflong \hfill\qed \fi 
    	}
    \end{proof}
  }
  
  Even for constant values of~$\maxoutdeg$, \ESA remains intractable for each 
  considered class of seat graphs.
  
  \newcommand{\esadeltanphard}{%
    \ESA is \NPc\ for each considered seat graph class and constant~$\maxoutdeg$. %
  }
\statementarxiv{theorem}{thm:ESA_delta_clique+path-cycle+matching}{\esadeltanphard}
  \appendixproofwithstatement{thm:ESA_delta_clique+path-cycle+matching}{\esadeltanphard}{
    \begin{proof}
      \textbf{Clique-graphs.} 
      We provide a polynomial-time reduction from \IS which is \NPh\ even on cubic graphs~\cite{garey1979}. 
      Let $(\hat{G}, h)$ be an instance of \IS. 
      Without loss of generality, we assume $h \mod 3 \neq 0$. 
      
      For each vertex $v_i \in V(\hat{G})$ we create three \myemph{vertex-agents} $p_i^0, p_i^1, p_i^2$. 
      For each $v_i \in V(\hat{G}), z,z' \in \{0,1,2\}$, the preferences are defined as follows, where the superscript $z\pm 1$ is taken $\mod 3$ (see Figure~\ref{fig:ESA_delta_clique} for the corresponding preference graph):
      \begin{compactitem}[--]
        \item $\sat[p_i^z](p_i^{z+1}) = 10$,
        \item $\sat[p_i^z](p_i^{z-1}) = -11$,
        \item $\sat[p_i^z](p_j^{z'}) = -1$ if $\{v_i, v_j\} \in E(\hat{G})$. 
        \item The non-mentioned preferences are set to zero.
      \end{compactitem}
  		
  		\begin{figure}[t]
  			\centering
  			\begin{tikzpicture}[>=stealth', shorten <= 2pt, shorten >= 2.5pt]
  			\def \xs {15ex}
  			\def \ys {11ex}
  			\def \ss {2.5ex}
  			\node[nodeU] (pi0) {};
  			\node[nodeU, right = 2*\xs  of pi0] (pj0) {};
  			\foreach \x in {pi,pj} {
  				\foreach \y / \z  in {0/1,1/2} {
  					\node[nodeU, below = \ys  of \x\y] (\x\z) {};
  				}
  			}

  			\foreach \x / \mid in {i/left, j/right} {
  				\foreach \y / \pos in {0/above, 2/below} {
  					\node[\pos = 0pt of p\x\y] {$p_\x^{\y}$};
  				}
  				\node[\mid = 0pt of p\x1] {$p_\x^{1}$};
  			}

  			\begin{scriptsize}
  			\foreach \i / \b in {i/left,j/right} {
  				\foreach \x / \y in {0/1, 1/2} {
  					\draw[-Stealth] (p\i\y) edge[bend \b=16, myRed]  node[pos=0.25, fill=white, inner sep=1pt] {$-11$} (p\i\x);
  					\draw[-Stealth] (p\i\x) edge[bend \b=15, black]  node[pos=0.25, fill=white, inner sep=1pt] {$10$} (p\i\y);
  				}
  				\draw[-Stealth] (p\i2) edge[bend \b=45, black]  node[pos=0.2, fill=white, inner sep=1pt] {$10$} (p\i0);
  				\draw[-Stealth] (p\i0) edge[bend \b=-85, myRed]  node[pos=0.2, fill=white, inner sep=1pt] {$-11$} (p\i2);
  			}
  			\end{scriptsize}
  			
  			\begin{tiny}
  			\foreach \x in {0,1,2} {
  				\foreach \y in {0,1,2} {
  					\draw[-Stealth] (pi\x) edge[bend left=6, myRed]  node[pos=0.2, fill=white, inner sep=1pt] {$-1$} (pj\y);
  					\draw[-Stealth] (pj\x) edge[bend left=7, myRed]  node[pos=0.2, fill=white, inner sep=1pt] {$-1$} (pi\y);
  				}
  			}
  			\end{tiny}
  			\end{tikzpicture}
  			\caption{Preference graph for \cref{thm:ESA_delta_clique+path-cycle+matching} for clique-graphs, where $\{v_i, v_j\} \in E(\hat{G})$.}
  			\label{fig:ESA_delta_clique}
  		\end{figure}
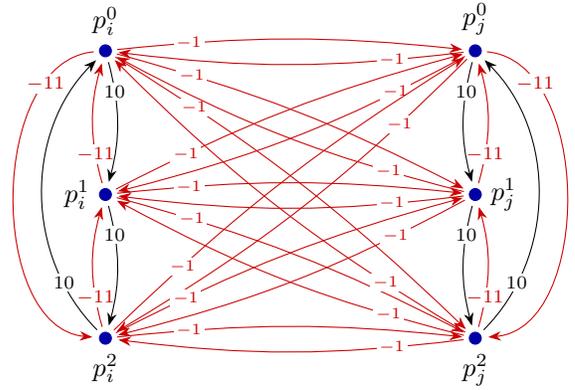
  	
      The seat graph $G$ consists of a clique with $\nonisolated\coloneqq h$ 
      vertices and $3|V(\hat{G})|-h$ isolated vertices.	      
      This completes the construction of the \ESA instance $I$ which can clearly be done in polynomial time. 
      Since we can assume that $\hat{G}$ is cubic, it holds that $\maxoutdeg = 11$. 
      
      Before we prove the correctness, we observe the following property of an \esArr:
      \begin{claim}\label{claim:ESA_delta_clique}
        For each $v_i \in V(\hat{G})$, at most one of the vertices $p_i^1, p_i^2, p_i^3$ can be assigned to non-isolated vertices.
      \end{claim}
      \begin{proof}[Proof of 
        \cref{claim:ESA_delta_clique}]
        \renewcommand{\qedsymbol}{$\diamond$}
        Suppose towards a contradiction, that for some vertex $v_i \in \hat{G}$ at least two of the agents $p_i^1, p_i^2, p_i^3$ are assigned to non-isolated vertices in an \esArr~$\sigma$. 
        We distinguish the following two cases (where $z\pm 1$ is taken $\mod 3$): 
        \begin{compactenum}[(i)]
          \item Two vertex-agents, $p_i^z, p_i^{z+1}, z \in \{0,1,2\}$ are assigned to non-isolated vertices. 
          It holds, $\util[p_i^{z+1}]{\sigma} \leq -1$. 
          Clearly, $p_i^{z+2}$ is isolated and envies $p_i^{z+1}$, i.e., $(p_i^{z+1}, p_i^{z+2})$ is an exchange-blocking pair. 
          
          \item All three agents are assigned to non-isolated vertices. 
          In this case all three agents $p_i^0, p_i^1, p_i^2$ envy the agents in the isolated vertices. 
          Moreover, since $h \mod 3 \neq 0$ and by the previous case, we can find $v_j \in V(\hat{G})$ where only one $p_j^z,  z \in \{0,1,2\}$ is non-isolated. 
          Hence, $(p_i^0, p_j^{z-1})$ is an exchange-blocking pair. \qedhere
        \end{compactenum}
      \end{proof}
      By this claim we can conclude that agents of $h$ different vertices of $\hat{G}$ are assigned to non-isolated vertices. 
      For each non-isolated agent $p_i^z$, agent~$p_i^{z-1}$ envies the non-isolated agents. To obtain an \esArr it is necessary that each non-isolated agent has non-negative utility. 
      Since the preferences between different sets of vertex-agents is non-positive, we proved the following:
      
      $I$ has an \esArr if and only if each of the~$h$ different non-isolated agents has non-negative utility if and only if the vertices corresponding to each pair of non-isolated agents are not adjacent, which is equivalent to~$\hat{G}$ admits a size-$h$ independent set.
      
      \paragraph*{Cycle-graphs.} 
      We provide a polynomial-time reduction from \HamPath which is \NPh\ even on cubic graphs~\cite{garey1979}. 
      Let $\hat{G}$ be an instance of \HamPath. 
      
      We start with creating three special agents named $x_1, x_2, x_3$.
      Then, for each vertex $v_i \in V(\hat{G})$ we create one \myemph{vertex-agent} named $p_i$. 
      For each $v_i \in V(\hat{G})$, we define the preferences as follows 
      (see Figure~\ref{fig:ESA_delta_cycle} for the corresponding preference 
      graph):
      \begin{compactitem}[--]
        \item $\sat[p_i](p_j) = \sat[p_j](p_i) = 1$ if $\{v_i,v_j\} \in E(\hat{G})$,			
        \item $\sat[p_i](x_1) = \sat[p_i](x_3) = 1$,	$\sat[p_i](x_2) = -2$,
        \item $\sat[x_1](x_2) = \sat[x_3](x_2) = 1$,
        \item $\sat[x_2](x_1) = \sat[x_2](x_3) = -1$.			
        \item The non-mentioned preferences are set to zero.
      \end{compactitem}
  
  		\begin{figure}[t]
  			\centering
  			\begin{tikzpicture}[>=stealth', shorten <= 2pt, shorten >= 2pt]
  			\def \xs {9ex}
  			\def \ys {10ex}
  			\def \ss {2.5ex}
  			
  			\node[nodeW] (x1) {};
  			\foreach \x / \y in {1/2, 2/3} {
  				\node[nodeW, right = 2*\xs  of x\x] (x\y) {};
  			}
  			\node[nodeU, below right = 0.65*\ys and \xs  of x1] (p1) {};
  			\node[nodeU, right = 2*\xs  of p1] (p2) {};

  			\foreach \x in {1, 2, 3} {
  				\node[above = 0pt of x\x] {$x_{\x}$};
  			}
  			\node[below = 0pt of p1] {$p_{1}$};
  			\node[below = 0pt of p2] {$p_{|V(\hat{G})|}$};
  			\path (p1) -- node[auto=false]{\ldots\ldots} (p2);

  			\node[ellipse, draw = gray, minimum width = 6.5cm, minimum height = 1.4cm, label={left:$\hat{G}$}] (ell) at (2*\xs+1.1ex,-0.8*\ys) {};

  			\begin{scriptsize}
  			\foreach \x in {1,2} {
  				\draw[-Stealth] (p\x) edge[black]  node[pos=0.2, fill=white, inner sep=1pt] {$1$} (x1);
  				\draw[-Stealth] (p\x) edge[myRed]  node[pos=0.2, fill=white, inner sep=1pt] {$-2$} (x2);
  				\draw[-Stealth] (p\x) edge[black]  node[pos=0.2, fill=white, inner sep=1pt] {$1$} (x3);
  			}
  			
  			\foreach \x / \b in {1/right, 3/left} {
  				\draw[-Stealth] (x\x) edge[bend \b=10, black]  node[pos=0.25, fill=white, inner sep=1pt] {$1$} (x2);
  				\draw[-Stealth] (x2) edge[bend \b=10, myRed]  node[pos=0.25, fill=white, inner sep=1pt] {$-1$} (x\x);
  			}
  			
  			\node[nodeU, right = \xs/2  of p1] (pi) {};
  			\node[nodeU, left = \xs/2  of p2] (pj) {};
  			\draw[Stealth-Stealth] (pi) edge[bend right=20, gray] node[below, inner sep=1pt] {\begin{tabular}{c}$1$ if adjacent, \\ $0$ otherwise \end{tabular}} (pj);
  			\end{scriptsize}
  			\end{tikzpicture}
  			\caption{Preference graph for \cref{thm:ESA_delta_clique+path-cycle+matching} for cycle-graphs.}
  			\label{fig:ESA_delta_cycle}
  		\end{figure}
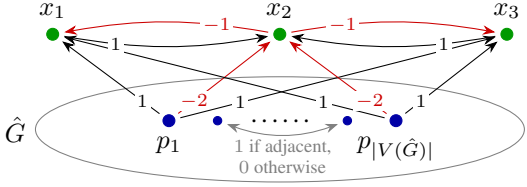

      The seat graph $G$ consists of a single cycle with $|V(\hat{G})|+3$ vertices.	
      
      This completes the construction of the \ESA instance $I$ which can clearly be done in polynomial time. 
      Since we can assume that $\hat{G}$ is cubic, it holds that $\maxoutdeg = 6$. 
      
      Next, we show the correctness of our reduction. 
      Let $P$ be a Hamiltonian path of $\hat{G}$.
      To obtain an \esArr of $I$ we first assign $x_1$ to an arbitrary seat in the seat graph. 
      Then, we continue with assigning $x_2$ next to $x_1$ and $x_3$ next to $x_2$.
      The remaining vertex-agents are assigned to the remaining seats according to the order in~$P$.
      In this arrangement all agents (except $x_2$) have their maximum possible utility and are therefore never part of an exchange-blocking pair.
      Hence, $x_2$ does not have an exchange-blocking partner, which means this arrangement is exchange-stable.
      
      For the backward direction, let $\sigma$ be an exchange-stable arrangement of $I$. 
      We observe that $\sigma$ satisfies the following:
      \begin{claim}\label{claim:ESA_delta_cycle}
        Agents $x_1, x_2, x_3$ are assigned consecutively. 
      \end{claim}
      \begin{proof}[Proof of 
        \cref{claim:ESA_delta_cycle}]
        \renewcommand{\qedsymbol}{$\diamond$}
        Suppose towards a contradiction, that~$x_1$ is not assigned next to 
        $x_2$, i.e., there is an agent $p \notin 
        \{x_1, x_2, x_3\}$ with $\sigma(p) \in \Neigh{G}(\sigma(x_2))$. 
        Then, $\util[x_1]{\sigma} = 0$ and $\util[p]{\sigma} \leq -1$. 
        Since $\util[p]{\swap{x_1}{p}} \geq 0$, $(p, x_1)$ is an exchange-blocking pair. 
        
        With an analogous argument we can see that $x_3$ also has to be a neighbor of $x_2$. 
        Hence, $x_1, x_2, x_3$ have to be assigned consecutively in every 
        \esArr. 
      \end{proof}
      
      From this claim it also follows that $x_1$ and $x_3$ are envy-free, but $x_2$ envies every vertex-agent.
      Since $\sigma$ is exchange-stable, we can conclude that $\util[p_i]{\sigma} = 2$, i.e., there is a Hamiltonian path 
      in $\hat{G}$. 
      
      \paragraph*{Path-graphs.} 
      For this case we reduce from \HamPath, which remains \NPh\ even on 
      cubic graphs and when one endpoint of the path is specified~\cite{garey1979}. 
      Let $(\hat{G}, s)$ be an instance of \HamPath. 
      We start with creating two special agents named~$x_1$ and~$x_2$.
      Then, for each vertex $v_i \in V(\hat{G})$ we create one agent named $p_i$. 
      We define the preferences as follows (see Figure~\ref{fig:ESA_delta_path} for the corresponding preference 
      graph):
      \begin{compactitem}[--]
        \item $\sat[p_i](p_j) = \sat[p_j](p_i) = 1$ if $\{v_i,v_j\} \in E(\hat{G})$,			
        \item for each $v_i \in V(\hat{G})\setminus\{s\}$, set	 
        $\sat[p_i](x_1) = 3$,
        \item for each $v_i \in V(\hat{G})$, set $\sat[p_i](x_2) = -3$,
        \item $\sat[x_1](x_2) = 1$, $\sat[x_2](x_1) = -1$, $\sat[p_s](x_1) = 2$.
        \item The non-mentioned preferences are set to zero.
      \end{compactitem}
  
  		\begin{figure}[t]
  			\centering
  			\begin{tikzpicture}[>=stealth', shorten <= 2pt, shorten >= 2pt]
  			\def \xs {9ex}
  			\def \ys {10ex}
  			\def \ss {2.5ex}
  			
  			\node[nodeW] (x1) {};
  			\node[nodeW, right = 2*\xs  of x1] (x2) {};
  			\node[nodeU, below left = 0.8*\ys and \xs/2+2ex  of x1] (p1) {};
  			\node[nodeU, right = 2.5*\xs  of p1] (p2) {};
  			\node[nodeU, right = \xs  of p2] (p3) {};

  			\foreach \x / \pos in {x/left,p/below} {
  				\node[\pos = 0pt of \x1] {$\x_1$};
  			}
  			\node[right = 0pt of x2] {$x_2$};
  			\node[below = 0pt of p2] {$p_{\hat{n}-1}$};
  			\node[below = 0pt of p3] {$p_{\hat{n}}$};
  			\path (p1) -- node[auto=false]{\ldots\ldots} (p2);

  			\node[ellipse, draw = gray, minimum width = 7.5cm, minimum height = 1.5cm, label={left:$\hat{G}$}] (ell) at (\xs+1ex,-0.9*\ys) {};

  			\begin{scriptsize}
  			\foreach \x in {1,2} {
  				\draw[-Stealth] (p\x) edge[black]  node[pos=0.17, fill=white, inner sep=1pt] {$2$} (x1);
  				\draw[-Stealth] (p\x) edge[myRed]  node[pos=0.17, fill=white, inner sep=1pt] {$-3$} (x2);
  			}
  			\draw[-Stealth] (p3) edge[black]  node[pos=0.17, fill=white, inner sep=1pt] {$1$} (x1);
  			\draw[-Stealth] (p3) edge[myRed]  node[pos=0.17, fill=white, inner sep=1pt] {$-2$} (x2);
  			
  			\foreach \x / \b in {1/right} {
  				\draw[-Stealth] (x\x) edge[bend \b=10, black]  node[pos=0.25, fill=white, inner sep=1pt] {$1$} (x2);
  				\draw[-Stealth] (x2) edge[bend \b=10, myRed]  node[pos=0.25, fill=white, inner sep=1pt] {$-1$} (x\x);
  			}
  			
  			\node[nodeU, right = \xs-2ex  of p1] (pi) {};
  			\node[nodeU, left = \xs-2ex  of p2] (pj) {};
  			\draw[Stealth-Stealth] (pi) edge[bend right=17, gray] node[below, inner sep=1pt] {\begin{tabular}{c}$1$ if adjacent, \\ $0$ otherwise \end{tabular}} (pj);
  			\end{scriptsize}
  			\end{tikzpicture}
  			\caption{Preference graph for \cref{thm:ESA_delta_clique+path-cycle+matching} for path-graphs, where $v_{\hat{n}} = s$.}
  			\label{fig:ESA_delta_path}
  		\end{figure}
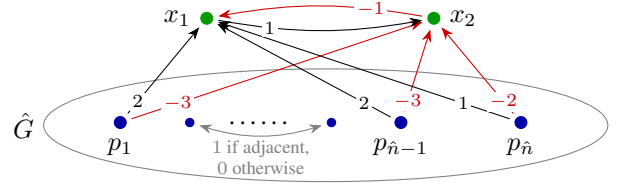

      The seat graph $G$ consists of a single path with $|V(\hat{G})|+2$ 
      vertices.
      This completes the construction of the \ESA instance $I$ which can clearly be done in polynomial time. 
      Since we can assume that $\hat{G}$ is cubic, it holds that $\maxoutdeg = 5$. 
      
      Next, we show the correctness of our reduction. 
      Let $P$ be a Hamiltonian path of $\hat{G}$ ending in $s$.
      To obtain an \esArr of $I$ we first assign $x_2$ to one of the endpoints of the path. 
      Then, we continue with assigning $x_1$ next to $x_2$. 
      The remaining vertex-agents are assigned to the remaining path according to the order in $P$ in such a way that~$p_s$ is assigned to the other endpoint of the seat graph.
      In this arrangement all agents (except $x_2$ and $p_s$) have their maximum possible utility and are therefore never part of an exchange-blocking pair. 
      Since $p_s$ does not envy $x_2$, we obtain that this arrangement is exchange-stable.
      
      For the backward direction, let $\sigma$ be an exchange-stable arrangement of $I$. 
      We observe that $\sigma$ satisfies the following:
      \begin{claim}\label{claim:ESA_delta_path}
        Agents $x_1, x_2$ are assigned consecutively and agents $x_2, p_s$ are assigned to the endpoints of the path in an \esArr.
      \end{claim}
      \begin{proof}[Proof of \cref{claim:ESA_delta_path}]
        \renewcommand{\qedsymbol}{$\diamond$}
        Suppose $x_1$ is not assigned next to~$x_2$, i.e., there is an agent $p \notin \{x_1, x_2\}$ with $\sigma(p) \in \Neigh{G}(\sigma(x_2))$. 
        Then, $\util[x_1]{\sigma} = 0$ and $\util[p]{\sigma} \leq -1$. 
        Since $\util[p]{\swap{x_1}{p}} \geq 0$, $(p, x_1)$ is an exchange-blocking pair.
        
        Agent $x_2$ has to be assigned to one endpoint of the path, since otherwise there is an agent $p \notin 
        \{x_1, x_2\}$ which is a neighbor of $x_2$. 
        Then, $\util[p]{\sigma} \leq -2$ and $\util[x_2]{\sigma} = -1$, i.e., $(p, x_2)$ is an exchange-blocking pair.
        
        Agent $p_s$ has to be assigned to the other endpoint of the path, since otherwise there is an agent $p \neq 
        p_s$ with $\util[p]{\sigma} \leq 1$ and $\sigma(x_1) \notin \Neigh{G}(\sigma(p))$, i.e., $(p, x_2)$ is an 
        exchange-blocking pair.
      \end{proof}
      
      From this claim it follows that the arrangement on the paths has the following form: $x_2, x_1, p_1, \ldots, p_{\hat{n}} = p_s$.	
      Moreover, we can conclude that $x_1$ is envy-free, but $x_2$ envies every vertex-agent $p_2, \ldots,p_{{\hat{n}}-1}, p_s$.
      Since $\sigma$ is exchange-stable, we can conclude that $\util[p_i]{\sigma} = 2$ for each $2 \leq i \leq {\hat{n}}-1$, i.e., $p_1, \ldots, p_n$ is a Hamiltonian path in~$\hat{G}$.
      
      \paragraph*{Matching-graphs.} 
      We provide a polynomial-time reduction from the \NPc\ problem \ESRthree~\cite{chen21}. 
      \decprob{\ESRthree}
      {Preference profile $\mathcal{P}$ with preferences of bounded length three.}
      {Does $\mathcal{P}$ admit an exchange-stable and perfect matching?}
      
      Let $\mathcal{P} = (\hat{V}, (\succ_i)_{i\in \hat{V}})$ be an instance of \ESRthree. 
      First, we create two special agents $x_1, x_2$. 
      Then, for each $v_i \in \hat{V}$ we create one agent $p_i$. 
      For each $v_i \in \hat{V}$, the preferences are defined as follows (see \cref{fig:ESA_delta_matching} for the corresponding preference graph):
      \begin{compactitem}[--]
        \item $\sat[p_i](x_1) = -1$, $\sat[p_i](x_2) = 1$,
        \item for each $v_j \in \hat{V}_i$, set $\sat[p_i](p_j) = |\{v \in \hat{V}_i\setminus\{v_j\}\colon v_j \succ_i v\}|+1$,
        \item $\sat[x_1](x_2) = -1$, $\sat[x_2](x_1) = 1$.
        \item The non-mentioned preferences are set to zero.
      \end{compactitem}
  		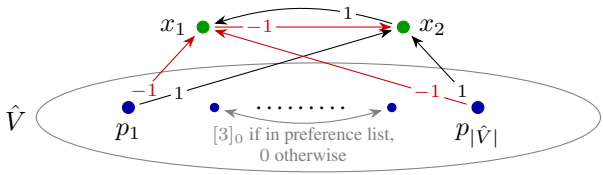
\begin{figure}[t]
  			\centering
  			\begin{tikzpicture}[>=stealth', shorten <= 2pt, shorten >= 2pt]
  			\def \xs {9ex}
  			\def \ys {8ex}
  			\def \ss {2.5ex}
  			
  			\node[nodeW] (x1) {};
  			\node[nodeW, right = 1.75*\xs of x1] (x2) {};
  			\node[nodeU, below left = 0.75*\ys and \xs/2+1ex  of x1] (p1) {};
  			\node[nodeU, below right = 0.75*\ys and \xs/2+1ex  of x2] (p3) {};
  			\begin{scriptsize}
  			\node[nodeU, right = \xs of p1] (p2) {};
  			\node[nodeU, left = \xs of p3] (pi) {};
  			\end{scriptsize}
  			\path (pi) -- node[auto=false]{\ldots\ldots\ldots} (p2);

  			\foreach \x / \pos in {x/left,p/below} {
  				\node[\pos = 0pt of \x1] {$\x_1$};
  			}
  			\node[right = 0pt of x2] {$x_2$};
  			\node[below = 0pt of p3] {$p_{|\hat{V}|}$};

  			\node[ellipse, draw = gray, minimum width = 7.5cm, minimum height = 1.45cm, label={left:$\hat{V}$}] (ell) at (\xs+0.7ex,-0.95*\ys) {};

  			\begin{scriptsize}
  			\foreach \x in {1,3} {
  				\draw[-Stealth] (p\x) edge[myRed]  node[pos=0.17, fill=white, inner sep=1pt] {$-1$} (x1);
  				\draw[-Stealth] (p\x) edge[black]  node[pos=0.17, fill=white, inner sep=1pt] {$1$} (x2);
  			}
  			
  			\foreach \x / \b in {1/right} {
  				\draw[-Stealth] (x\x) edge[bend \b=0, myRed]  node[pos=0.25, fill=white, inner sep=1pt] {$-1$} (x2);
  				\draw[-Stealth] (x2) edge[bend \b=17, black]  node[pos=0.25, fill=white, inner sep=1pt] {$1$} (x\x);
  			}
  			
  			\draw[Stealth-Stealth] (pi) edge[bend left=16, gray] node[below, inner sep=1pt] {\begin{tabular}{c}$[3]_0$ if in preference list, \\ $0$ otherwise \end{tabular}} (p2);
  			\end{scriptsize}
  			\end{tikzpicture}
  			\caption{Preference graph for \cref{thm:ESA_delta_clique+path-cycle+matching} for matching-graphs.}
  			\label{fig:ESA_delta_matching}
  		\end{figure}
      The seat graph $G$ consists of $(|\hat{V}|+2)/2$ disjoint edges.       
      This completes the construction of the \ESA instance $I$ which can clearly be done in polynomial time. 
      Since the length of each preference list is bounded by three, it holds that $\maxoutdeg = 5$. 
      
      Before we prove the correctness, we observe the following necessary condition for an \esArr:
      \begin{claim}\label{claim:ESA_delta_matching}
        In every \esArr agent $x_1$ is matched to $x_2$.
      \end{claim}
      \begin{proof}[Proof of 
        \cref{claim:ESA_delta_matching}]
        \renewcommand{\qedsymbol}{$\diamond$}
        Suppose that $x_1$ is matched to $p \neq x_2$ and $x_2$ is matched to $p' \neq x_1$. Then, $\util[p]{\sigma} = 
        -1$ and $\util[x_2]{\sigma} = 0$.
        Because $\sat[p](p') \geq 0$ and $\sat[x_2](x_1) = 1$ it follows that $(p, x_2)$ is an exchange-blocking pair. 
      \end{proof}
      Hence, in every \esArr agent~$x_1$ envies the agents assigned to any other matching-edge.
      This means that $x_2$ is not in an exchange-blocking pair if the utility of each original agent $p_i$ is positive, i.e., $p_i$ is matched to an agent in his preference list. 
      Moreover, agents $p_1, \ldots, p_{|\hat{V}|}$ are exchange-stable in~$I$ if and only if this is an exchange-stable matching of~$\mathcal{P}$, which completes this proof. 	
      
      \paragraph*{Stars-graphs.} This follows from the previous case since a matching-graph is also a stars-graph. 
    \end{proof}
  }
  
  \noindent For the combined parameters $(\nonisolated, \maxoutdeg)$, \ESA becomes \fptf\ using the same idea as \cref{thm:MUA_k+delta}. 
  
  \newcommand{\esakdeltafpt}{%
    \ESA is \fpt\ wrt.\ $\nonisolated+\maxoutdeg$. %
  }
\statementarxiv{theorem}{thm:ESA_k+delta}{\esakdeltafpt}
  \appendixproofwithstatement{thm:ESA_k+delta}{\esakdeltafpt}{
    \begin{proof}
      The idea is to obtain a polynomial-size problem kernel using the same algorithm as \cref{thm:MUA_k+delta}. We state it here for completeness.
      First, we can assume that the seat graph has isolated vertices since otherwise $\nonisolated = |\agents|$ and we have a linear kernel.
      Therefore, $\egal \leq 0$ for every arrangement~$\sigma$.
      
      Next, we claim that in polynomial time, we either find $\nonisolated$ agents which have zero preference towards each other (i.e., they are independent in the preference graph), or obtain a problem kernel.

      Before we proceed with the algorithm, we observe that in every digraph with maximum out-degree $\maxoutdeg$, there is always a vertex with in-degree bounded by $\maxoutdeg$. Hence, there is a vertex with sum of in- and out-degrees at most $2\maxoutdeg$.

      Now, we select in each step an agent $p$ with minimum in-degree in the preference graph, put him to our solution. %
      After that we delete all in- and out-neighbors of $p$. 
      If, after $\nonisolated$ steps, we can find a set~$S$ of $\nonisolated$ ``independent'' agents (they do not have arcs towards each other), then we can assign them arbitrarily to the non-isolated vertices.
      Since the preference between every pair of agents in $S$ is zero, the utility of each of these agents is also zero and no agent in~$S$ envies an isolated agent. 
      Hence, there is no exchange-blocking pair between an isolated and non-isolated agent. 
      Moreover, since these $\nonisolated$ agents are independent, they do not envy each other, i.e., this arrangement is exchange-stable.
      
      However, if we cannot find $\nonisolated$ independent agents, then the instance has at most $\nonisolated(1+2\maxoutdeg)$ agents since in each step we deleted at most $1+2\maxoutdeg$ agents.
      It is straightforward that the approach above can be done in polynomial time.
    \end{proof}
  }
\ifarxiv
	\section{Additional Complexity Results for \EFA\ and \ESA} \label{sec:complexity-EFA-ESA}	
	\newE{
		In this section we give further complexity results for finding an envy-free resp.\ exchange-stable arrangement.
		
		Despite the tractability result for clique-graphs in \cref{thm:EFA_clique_nonneg+symm}, \EFA remains \NPh\ even for binary and symmetric preferences. 
		Hardness holds for simple graph classes as path- or cycle-graphs as the next theorem shows. 
		
		\begin{theorem} \label{thm:EFA_NP_path-cycle_bin+symm}
			\EFA is \NPh\ even for binary and symmetric preferences and a single path or cycle as seat graph. 
		\end{theorem}
		\begin{proof}
			\textbf{Path-graph.} 
			We prove this via a polynomial-time reduction from \HamPath which is \NPh\ even if one endpoint is specified~\cite{garey1979}. 
			
			Let $(\hat{G},s)$ with $s \in V(\hat{G})$ be an instance of \HamPath. 
			Without loss of generality, assume~$s=v_1$. 
			We define~$\hat{n} \coloneqq |V(\hat{G})|$ and create an instance of \EFA in the following way:
			For each vertex $v_i \in V(\hat{G})$ we create one \myemph{vertex-agent}~$p_i$. 
			In addition to that, we create $4 \hat{n} +1$ special agents, i.e., 
			\begin{align*}
				\agents = &\{p_1, \ldots, p_{\hat{n}}\} \cup \{x_1, \ldots, x_{2\hat{n}} \} \cup \{y_1, \ldots, y_{2\hat{n}-1}\} \\
				&\cup \{z_1, z_2\}.
			\end{align*} 
			
			Since we will have symmetric preferences, for each pair of agents we only specify one value.
			We define the preferences in the following way (see \cref{fig:EFA_NP_path_bin+symm} for the corresponding preference graph), where~$v_i, v_{i'} \in V(\hat{G})$, $j \in [2\hat{n}]$, and $\ell \in [2\hat{n}-1]$: 
			\begin{compactitem}
				\item $\sat[p_i](p_{i'}) = 1$ if and only if $\{v_i, v_{i'}\} \in E(\hat{G})$,
				\item $\sat[p_i](x_j) = \sat[x_j](y_\ell) =1$,
				\item $\sat[x_1](z_1)=1$,
				\item $\sat[p_1](z_2)=1$. 
				\item The non-mentioned preferences are set to zero. 
			\end{compactitem}
			
			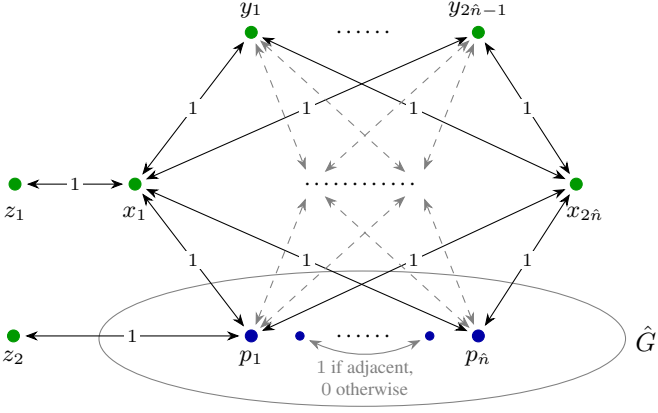
\begin{figure}[t]
				\centering
				\begin{tikzpicture}[>=stealth', shorten <= 2pt, shorten >= 2pt]
					\def \xs {9ex}
					\def \ys {12ex}
					\def \ss {2.5ex}
					
					\node[nodeW] (x1) {};
					\node[nodeW, right = 4*\xs  of x1] (x2) {};
					\node[right = 2*\xs of x1] (x3) {};
					\node[right = 1.5*\xs of x1] (x4) {};
					\node[right = 2.5*\xs of x1] (x5) {};
					
					\node[nodeU, below right = \ys and \xs  of x1] (p1) {};
					\node[nodeU, right = 2*\xs  of p1] (p2) {};

					\node[nodeW, above right = \ys and \xs  of x1] (y1) {};
					\node[nodeW, right = 2*\xs  of y1] (y2) {};
					
					\node[nodeW, left =\xs  of x1] (z1) {};
					\node[nodeW, left =2.1*\xs  of p1] (z2) {};

					\begin{footnotesize}
						\node[below = 3pt of x1] {$x_1$};
						\node[below = 3pt of x2] {$~~~x_{2\hat{n}}$};
						
						\node[above = 0pt of y1] {$y_1$};
						\node[above = 0pt of y2] {$y_{2\hat{n}-1}$};
						
						\node[below = 0pt of p1] {$p_{1}$};
						\node[below = 0pt of p2] {$p_{\hat{n}}$};
						
						\node[below = 3pt of z1] {$z_1$};
						\node[below = 0pt of z2] {$z_{2}$};
						
						\path (p1) -- node[auto=false]{\ldots\ldots} (p2);
						\path (y1) -- node[auto=false]{\ldots\ldots} (y2);
						\path (x1) -- node[auto=false]{~~\ldots\ldots\ldots\ldots} (x2);
					\end{footnotesize}

					\node[ellipse, draw = gray, minimum width = 7cm, minimum height = 1.8cm, label={right:$\hat{G}$}] (ell) at (2*\xs+1.1ex,-\ys-1ex) {};

					\begin{scriptsize}
						\foreach \x / \y in {x/y, x/p} {
							\foreach \i in {1,2} {
								\foreach \j in {1,2} {
									\draw[Stealth-Stealth] (\x\i) edge[black] node[pos=0.5, fill=white, inner sep=1pt] {$1$} (\y\j);
								}
							}
						}
						\foreach \i in {1,2} {
							\foreach \j in {4,5} {
								\draw[Stealth-Stealth] (y\i) edge[dashed, gray] (x\j);
								\draw[Stealth-Stealth] (p\i) edge[dashed, gray] (x\j);
							}
						}
						
						\draw[Stealth-Stealth] (x1) edge[black] node[pos=0.5, fill=white, inner sep=1pt] {$1$} (z1);
						\draw[Stealth-Stealth] (p1) edge[black] node[pos=0.5, fill=white, inner sep=1pt] {$1$} (z2);
						
						\node[nodeU, right = \xs/2  of p1] (pi) {};
						\node[nodeU, left = \xs/2  of p2] (pj) {};
						\draw[Stealth-Stealth] (pi) edge[bend right=25, gray] node[below, inner sep=2pt] {\begin{tabular}{c}$1$ if adjacent, \\ $0$ otherwise \end{tabular}} (pj);
					\end{scriptsize}
				\end{tikzpicture}
				\caption{Preferences defined in the proof for \cref{thm:EFA_NP_path-cycle_bin+symm} for path-graphs, where vertex-agents are colored blue, and special agents are colored green. Dashed arcs have weight~$1$. }
				\label{fig:EFA_NP_path_bin+symm}
			\end{figure}
			
			The seat graph~$G$ consists of a single path with $5\hat{n}+1$ vertices. 
			This completes the construction of an \EFA instance with binary and symmetric preferences, which can clearly be done in polynomial time. 
			
			Next, we show the correctness of our reduction, i.e., there is a Hamiltonian path in~$\hat{G}$ ending in~$s$ if and only if there is an \efArr of~$I$. 		
			
			For the forward direction, let $v_t-v_1$ be a Hamiltonian path of $(\hat{G},s)$ where $v_t$ is the second endpoint of the path.
			To obtain an envy-free arrangement of~$I$ we assign the agents to the seat graph in the following order:
			\begin{align*}
				z_1, x_1,y_1, x_2, y_2, \ldots, x_{2\hat{n}-1}, y_{2\hat{n}-1}, x_{2\hat{n}}, p_t - p_1, z_2. 
			\end{align*} 
			The term $p_t - p_1$ means that we are assigning the agents $p_1, \ldots, p_{\hat{n}}$ according to their order in the Hamiltonian path~{$v_t-v_1$}. 
			Observe that each agent (except~$z_1$ and~$z_2$) has utility two and agents~$z_1$ and~$z_2$ have utility one.
			Since each agent has his maximum possible utility, this arrangement is envy-free. 
			
			For the backward direction, let~$\sigma$ be an \efArr of~$I$. 
			We can observe the following. 
			
			\begin{claim} \label{claim:EFA_NP_path}
				Every \efArr~$\sigma$ fulfills the following:
				\begin{compactenum}[(i)]
					\item \label{itm:EFA_NP_path-cycle_bin-symm1} The utility of each agent in~$\sigma$ is positive.
					
					\item \label{itm:EFA_NP_path-cycle_bin-symm2} For each $i \in [2\hat{n}]$, it holds that $\util[x_i]{\sigma}=2$.
					
					\item \label{itm:EFA_NP_path-cycle_bin-symm3} For each $i \in [2\hat{n}-1]$, it holds that $\util[y_i]{\sigma}=2$. 
				\end{compactenum}
			\end{claim}
			\begin{proof}[Proof of \cref{claim:EFA_NP_path}]
				\renewcommand{\qedsymbol}{$\diamond$}
				Let~$\sigma$ be an \efArr. 
				\begin{compactenum}[(i)]
					\item This has to hold because the out-degree of each agent in the preference graph is at least one and the preferences are non-negative. 
					
					\item By the previous statement we know that the utility of each agent is positive. 
					Suppose towards a contradiction, that $\util[x_i]{\sigma}=1$ for some $i \in [2\hat{n}]$. 
					We define the set $Q:=\{p_1,\ldots,p_{\hat{n}},y_1, \ldots, y_{2\hat{n}-1}\}$.
					
					Since $\sigma$ is envy-free, there are no two agents $q_1, q_2 \in Q$ which share a common neighbor in~$\sigma$, i.e., $N_G(\sigma(q_1))\cap N_G(\sigma(q_2)) = \emptyset$, as otherwise~$x_i$ would prefer the seat of this common neighbor because this gives him utility two.
					This means, the seats of two agents in~$Q$ are either adjacent or have at least two other seats in between them.
					Moreover, the seats of no three agents in~$Q$ are adjacent. 
					Therefore, we can partition~$Q$ into two disjoint subsets
					\begin{align*}
						&Q_1 := \{q_1 \in Q~|~\forall q_2 \in Q: \sigma(q_2) \notin N_G(\sigma(q_1))\}, \\
						&Q_2 := \{q_1 \in Q~|~\exists q_2 \in Q: \sigma(q_2) \in N_G(\sigma(q_1))\},
					\end{align*}
					where $Q_1$ is the set of agents, which have no agent in~$Q$ as a neighbor in~$\sigma$.  
					The set~$Q_2$ describes the set of agents, which have exactly one neighbor of~$Q$ in~$\sigma$. 
					Note that this neighbor is obviously also contained in~$Q_2$.
					Hence, the set~$Q_2$ can be seen as pairs of agents.
					Clearly, it holds $Q_1 \dot{\cup} Q_2 = Q$, which implies $|Q_1|+|Q_2|=3\hat{n}-1$.
					
					We now use the fact that the seats of two agents in~$Q$ are either adjacent or have at least two other seats in between them. 
					The adjacent agents are the ones in~$Q_2$, i.e., we have $|Q_2|/2$ adjacent pairs. 
					Moreover, $|Q_1|$ agents have no neighbor in~$Q$.
					Since we always need to assign at least two agents from $\agents \setminus Q$ between each two agents in~$Q_1$ and each two pairs in~$Q_2$, we obtain that the minimum number of agents needed from $\agents \setminus Q$ is
					\begin{align*}
						2 \left(|Q_1|+\frac{|Q_2|}{2} -1 \right) = 2|Q_1|+|Q_2|-2 \geq 3\hat{n}-3.
					\end{align*}
					However, there are only $|\agents \setminus Q|=2\hat{n}+2$ agents. 
					Hence, for $\hat{n}>5$ this is a contradiction to our assumption that~$\sigma$ is envy-free. 
					
					\item Again, suppose $\util[y_i]{\sigma}=1$ for some $i \in [2\hat{n}-1]$ in~$\sigma$. 
					Similar to the previous claim we will show using a counting argument that this leads to a contradiction. 
					
					Since $\sigma$ is envy-free, there are no two distinct agents~$x_j, x_\ell$ with $j,\ell \in [2\hat{n}]$ which share a neighbor in $\sigma$, i.e., $N_G(\sigma(x_j))\cap N_G(\sigma(x_\ell)) = \emptyset$. 
					From statement~\ref{itm:EFA_NP_path-cycle_bin-symm2} we know that $\util[x_j]{\sigma}=2$ for each $j \in [2\hat{n}]$. 
					Using arguments similar to the proof of \ref{itm:EFA_NP_path-cycle_bin-symm2} we obtain that the minimum number of agents needed to be assigned between $x_1,\ldots, x_{2\hat{n}}$ is at least $2(2\hat{n}-1)+2 = 4\hat{n}$. 
					However, there are only $|\agents \setminus \{x_1,\ldots, x_{2\hat{n}}\}|=3\hat{n}+1$ possible agents, which is a contradiction. 
					\qedhere
				\end{compactenum}
			\end{proof}
			
			Using this claim we show now that if there is an \efArr~$\sigma$ of~$I$, then~$\hat{G}$ contains a Hamiltonian path ending at~$s$. 
			
			By \cref{claim:EFA_NP_path}(\ref{itm:EFA_NP_path-cycle_bin-symm1}), we know that $\sigma(z_1)$ and $\sigma(x_1)$ as well as~$\sigma(z_2)$ and $\sigma(p_1)$ are adjacent in~$G$, otherwise~$z_1$ (resp.~$z_2$) is envious. 
			Moreover, from \cref{claim:EFA_NP_path}(\ref{itm:EFA_NP_path-cycle_bin-symm3}) we obtain that the utility of each $y_1,\ldots, y_{2\hat{n}-1}$ is two. 
			Hence, each of these agents is assigned between two agents of $x_1, \ldots, x_{2\hat{n}}$. 
			However, this is only possible, if the arrangement of these agents form an alternating sequence starting or ending with~$x_1$. 
			From this we can conclude that $\util[p_i]{\sigma}=2$, otherwise~$p_i$ envies each $y_1, \ldots, y_{2\hat{n}-1}$. 			
			Hence, agents~$z_1$ and~$z_2$ have to be assigned to the two endpoints of the seat graph. 
			This means, $\sigma(p_1),\ldots, \sigma(p_n)$ are connected in~$\sigma$ and form a Hamiltonian path in~$\hat{G}$, which completes the proof.
			
			\paragraph*{Cycle-graph.} 
			Now, we turn to a cycle-graph as seat graph and modify the reduction above. 
			For this case we do not need a specified endpoint in the \HamPath instance. 
			Therefore, we also do not create the special agents~$z_1, z_2$ as in the previous reduction. 
			The preferences for the remaining agents are defined exactly as for the path-graph. 
			The seat graph $G$ consists of a single cycle with $5\hat{n}-1$ vertices. 
			
			The proof of correctness of this reduction can be done analogously to the previous case.
		\end{proof}
		
		Bodlaender et al.\cite{bodlaender2020tech} proved that \EFA is \NPh, even if the preference graph is a directed acyclic graph and the seat graph is a tree. 
		We strengthen their hardness result:
		
		\begin{theorem}\label{thm:EFA_NP_matching_bin-DAG}
			\EFA is \NPh\ even for a matching-graph and for binary preferences, where the preference graph is a directed acyclic graph.
		\end{theorem}
		\begin{proof}
			We show hardness via a polynomial time reduction from \ExactCover~\cite{garey1979}. 
			
			Let $\hat{I} = (U,(S_j)_{j\in [m]})$ be an instance of \ExactCover and $U = \{u_1, \ldots, u_{3 \hat{n}}\}$. 
			Without loss of generality, assume that $m \geq \hat{n}$. 
			We construct an instance~$I$ of \EFA as follows: 
			For each element $u_i \in U$ we create one \myemph{element-agent} named~$p_i$. 
			For each subset $S_j, j\in [m]$, we create five \myemph{set-agents} $q_j^1, q_j^2, q_j^3, v_j, w_j$. 
			Finally, we add two special agents named~$x_1$ and~$x_2$. 
			This means the set of agents consists of the following $3\hat{n}+5m+2$ agents:
			\begin{align*}
				\agents = &\{p_1, \ldots, p_{3\hat{n}}\} \cup \{x_1, x_2\} \\
				& \cup \{ q_1^1, q_1^2, q_1^3, v_1, w_1, \ldots, q_m^1, q_m^2, q_m^3, v_m, w_m\}. 
			\end{align*}				
			We define the binary preferences as follows, where $i \in [3\hat{n}]$ and each $j \in [m]$ (see \cref{fig:EFA_edges-iv_DAG} for the corresponding preference graph):
			\begin{compactitem}
				\item If $u_i \in S_j$, then set $\sat[p_i](q_j^z) = 1$ for each $z \in [3]$ and set $\sat[p_i](x_1)=1$. 
				
				\item Set $\sat[v_j](q_j^z) = \sat[v_j](w_j) = 1$ for each $z \in [3]$.
				
				\item Set $\sat[x_1](p) = 1$ for each $p \in \{q_j^1, q_j^2, q_j^3, v_j, w_j\}$. 
				
				\item The non-mentioned preferences are set to zero. 
			\end{compactitem}
			
			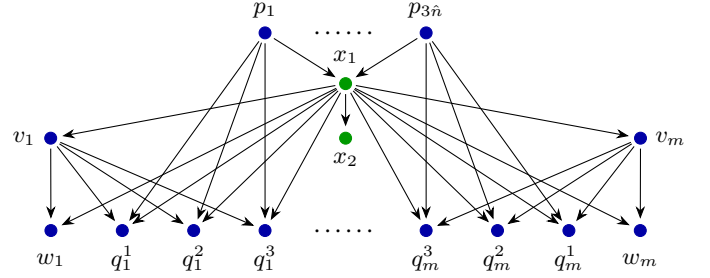
\begin{figure}[t]
				\centering
				\begin{tikzpicture}[>=stealth', shorten <= 1.5pt, shorten >= 2pt]
					\def \xs {6ex}
					\def \ys {7ex}
					
					\node[nodeW] (x1) {};
					\node[nodeW, below = 0.5*\ys of x1] (x2) {};
					
					\foreach \i / \pos in {1/left, 3n/right}{
						\node[nodeU, above \pos = 0.5*\ys and \xs of x1] (e\i) {};
					}
					\path (e1) -- node[auto=false]{\ldots\ldots} (e3n);
					
					\foreach \i / \pos in {1/left, m/right}{
						\node[nodeU, \pos = 3.95*\xs of x2] (u\i) {};
						\node[nodeU, below \pos = 1*\ys and 4*\xs of x2] (w\i) {};
						\node[nodeU, below \pos = 1*\ys and 3*\xs of x2] (s\i1) {};
						\node[nodeU, below \pos = 1*\ys and 2*\xs of x2] (s\i2) {};
						\node[nodeU, below \pos = 1*\ys and 1*\xs of x2] (s\i3) {};
					}
					\path (s13) -- node[auto=false]{\ldots\ldots} (sm3);

					\begin{footnotesize}
						\node[above = 2pt of x1] {$x_1$};
						\node[below = 0pt of x2] {$x_2$};
						
						\foreach \i / \ii / \j / \pos in {1/1/1/left, 3n/3\hat{n}/m/right}{
							\node[above = 0pt of e\i] {$p_{\ii}$};
							\node[\pos = 0pt of u\j] {$v_{\j}$};
							\node[below = 3pt of w\j] {$w_{\j}$};
							
							\foreach \x in {1, 2, 3} {
								\node[below = 0pt of s\j\x] {$q_{\j}^{\x}$};
							}
						}
					\end{footnotesize}

					\begin{scriptsize}
						\draw[-Stealth] (x1) edge (x2);
						\foreach \i / \j in {1/1, 3n/m}{
							\draw[-Stealth] (e\i) edge (x1);
							\draw[-Stealth] (x1) edge (u\j);
							\draw[-Stealth] (x1) edge (w\j);
							\draw[-Stealth] (u\j) edge (w\j);
							
							\foreach \x in {1,2,3} {
								\draw[-Stealth] (u\j) edge (s\j\x);
								\draw[-Stealth] (e\i) edge (s\j\x);
								\draw[-Stealth] (x1) edge (s\j\x);
							}
						}
					\end{scriptsize}
				\end{tikzpicture}
				\caption{Preference defined in the proof for \cref{thm:EFA_NP_matching_bin-DAG}. Since the agents have binary preferences, the weight of each edge is~$1$. }
				\label{fig:EFA_edges-iv_DAG}
			\end{figure}
			
			The seat graph~$G$ consists of $4\hat{n}+1$ disjoint edges and $5m-5\hat{n}$ isolated vertices.
			This completes the construction of an \EFA instance~$I$ which can clearly be done in polynomial time. 		
			Obviously, the preferences are binary. 
			To see that the preference graph contains no directed cycles we give a linear ordering of the agents:
			\begin{align*}
				p_1, \ldots, p_{3\hat{n}}, x_1, x_2, v_1, w_1, q_1^1, q_1^2, q_1^3, \ldots, v_m, w_m, q_1^m, q_1^m, q_1^m. 
			\end{align*}
			It can be easily verified that this is a topological ordering of the preference graph (see also \cref{fig:EFA_edges-iv_DAG}).
			
			To prove the correctness, let $J \subseteq [m]$ be an exact cover of~$\hat{I}$.
			Since we know that each element $u_i \in U$ is covered by exactly one set $S_j, j \in J$, we can assign each element-agent~$p_i$ together with one of the set-agents $q_j^1,q_j^2$, or~$q_j^3$ to the endpoints of the same edge. 
			Here, we can choose arbitrarily between the three set-agents. 
			Since each set contains exactly three elements and $|J|=\hat{n}$ we can match all~$3\hat{n}$ elements with a set-agent which contains it. 				
			Moreover, for each $j \in J$ we assign agent~$u_j$ and~$w_j$ to the endpoints of the same edge. 				
			Furthermore, we assign agent~$x_1$ and~$x_2$ to the endpoints of the last edge. This means we have assigned $8\hat{n}+2$ agents to the endpoints of the $4\hat{n}+1$ edges. 
			All remaining agents are assigned to the isolated vertices. 
			
			This arrangement is envy-free because of the following:
			\begin{compactitem}
				\item Agents $q_j^1, q_j^2, q_j^3, w_j$ with $ j \in [m]$, and $x_2$ are always envy-free since they have zero preference towards every agent. Therefore, their utility is always zero. 
				
				\item Agents $x_1, p_1, \ldots, p_{3\hat{n}}$ and $v_j, j \in J$ have their maximum utility of one. Therefore, they are also envy-free. 
				
				\item Agents $u_j, j \in [m]\setminus J$ have $\util[v_j]{\sigma}=0$. But the agents towards which he has a positive preference, are all isolated.
			\end{compactitem}		
			
			Before we prove the backward direction, we define 
			\begin{align*}
				J \coloneqq \{ j \in [m]\colon \exists (i,\ell) \in [3\hat{n}] \times [3] \text{ such that } \\
				\{\sigma(p_i), \sigma(q_j^\ell )\} \in E(G) \}.
			\end{align*}
			We want to show that~$J$ is a set cover of~$\hat{I}$. 
			First, we observe the following:
			\begin{claim} \label{claim:EFA_NP_matching_bin}
				Every \efArr~$\sigma$ of $I$ satisfies:
				\begin{compactenum}[(i)]
					\item \label{itm:EFA_NP_DAG1} Agent~$x_1$ is not isolated.
					
					\item \label{itm:EFA_NP_DAG2} Each element-agent $p_i$ is matched to a set-agent~$q_j$ which ``contains'' it, i.e., $u_i \in S_j$. 
					
					\item \label{itm:EFA_NP_DAG3} $|J| = \hat{n}$ .
				\end{compactenum}
			\end{claim}
			\begin{proof}[Proof of \cref{claim:EFA_NP_matching_bin}]
				\renewcommand{\qedsymbol}{$\diamond$}
				Let~$\prefgraph$ be the preference graph and~$\sigma$ an \efArr of~$I$. 
				\begin{compactenum}[(i)]
					\item Since $|\Nout{\prefgraph}(x_1)|=5m+1>5m-5\hat{n}$, there is at least one agent $p \in \Nout{\prefgraph}(x_1)$, which is not isolated.
					This means, if $\degr{G}(\sigma(x_1))=0$, then $x_1$ envies the agent, which is assigned to the same edge as $p$. 
					Therefore, it holds that $\degr{G}(\sigma(x_1))=1$. 
					
					\item Let $p_i$ with $i \in [3\hat{n}]$ be an element-agent. 
					
					By Statement~(\ref{itm:EFA_NP_DAG1}) we know that~$x_1$ is not isolated in~$\sigma$. 
					Since it holds $\sat[p_i](x_1)=1$ for all $i \in [3\hat{n}]$, agent~$p_i$ has to be matched to an agent, towards which he has a positive utility, i.e., to some agent in
					\begin{align*}
						\Neigh{\prefgraph}(\sigma(p_i)) \subseteq \{\sigma(q_j^1),\sigma(q_j^2), \sigma(q_j^3)~|~u_i\in S_j\}\cup \{x_1\}.
					\end{align*}
					
					However, if~$p_i$ and~$x_1$ are matched, then $\util[x_1]{\sigma}=0$, which implies that~$x_1$ is envious since by the previous reasoning at least one set-agent has to be non-isolated. 
					Therefore,~$x_1$ and~$p_i$ cannot be assigned to the endpoints of the same edge, so~$p_i$ is matched to a set-agent which contains it. 
					
					\item By Statement~(\ref{itm:EFA_NP_DAG2}) and the fact that~$|U|=3\hat{n}$ it follows that $|J|\geq \hat{n}$. 
					
					Furthermore, for all $j\in J$, agent~$v_j$ has to be matched and get utility one, since it holds that $\sat[v_j](q_j^\ell)$ for each $\ell \in [3]$. 
					Moreover, for each distinct $j,j' \in J$, no two agents~$v_j$ and~$v_{j'}$ can be matched to each other in an \efArr. 
					Combining this with the previous two statements we obtain that the following has to be satisfied:
					\begin{align*}
						2|U|+ |\{x_1\}| + 2|J| \le 2\cdot(4\hat{n}+1) \Leftrightarrow 2|J| \le 2\hat{n}+1.
					\end{align*}
                                        Hence, it holds that $|J|\leq \hat{n}$, which shows $|J|= \hat{n}$. \qedhere
				\end{compactenum}
			\end{proof}
			
		\noindent	\cref{claim:EFA_NP_matching_bin}(\ref{itm:EFA_NP_DAG2})--(\ref{itm:EFA_NP_DAG3}) imply that~$J$ is an exact cover of~$\hat{I}$.
		\end{proof}
	
		\ESA also remains intractable for cycle-graphs even for non-negative preferences.
			
		\begin{theorem}\label{thm:ESA_NP_cycle_nonneg}
			\ESA is \NPh\ even for non-negative preferences and a single cycle as seat graph.
		\end{theorem}
		\begin{proof}
			We provide a polynomial-time reduction from the \NPc\ problem \HamPath\cite{garey1979}.
			Let $\hat{G}$ be an instance of \HamPath. 		
			We start with creating three special agents named $x_1, x_2$, and~$x_3$.
			For each vertex $v_i \in V(\hat{G})$, we create one \myemph{vertex-agent} named~$p_i$ and we define the preferences as follows (see \cref{fig:ESA_NP_cycle_nonneg} for the corresponding preference graph):
			\begin{compactitem}
				\item If $\{v_i,v_j\} \in E(\hat{G})$, then set $\sat[p_i](p_j) = \sat[p_j](p_i) = 3$; otherwise set the preference $\sat[p_i](p_j) = \sat[p_j](p_i) = 2$. 
				\item Set $\sat[p_i](x_1) = \sat[p_i](x_3) = 3$.
				\item Set $\sat[x_2](x_1) = \sat[x_3](x_1) = 1$.
				\item The non-mentioned preferences are set to zero.
			\end{compactitem}
			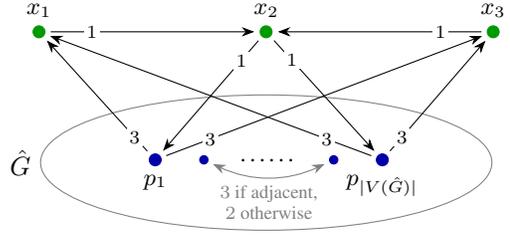
\begin{figure}[t]
				\centering
				\begin{tikzpicture}[>=stealth', shorten <= 2pt, shorten >= 2pt]
					\def \xs {9ex}
					\def \ys {10ex}
					\def \ss {2.5ex}
					
					\node[nodeW] (x1) {};
					\foreach \x / \y in {1/2, 2/3} {
						\node[nodeW, right = 2*\xs  of x\x] (x\y) {};
					}
					\node[nodeU, below right = \ys and \xs  of x1] (p1) {};
					\node[nodeU, right = 2*\xs  of p1] (p2) {};

					\begin{footnotesize}
						\foreach \x in {1, 2, 3} {
							\node[above = 0pt of x\x] {$x_{\x}$};
						}
						\node[below = 0pt of p1] {$p_{1}$};
						\node[below = 0pt of p2] {$p_{|V(\hat{G})|}$};
						\path (p1) -- node[auto=false]{\ldots\ldots} (p2);
					\end{footnotesize}

					\node[ellipse, draw = gray, minimum width = 6cm, minimum height = 1.8cm, label={left:$\hat{G}$}] (ell) at (2*\xs+1.1ex,-\ys-1ex) {};

					\begin{scriptsize}
						\foreach \x in {1,2} {
							\draw[-Stealth] (p\x) edge[black]  node[pos=0.15, fill=white, inner sep=1pt] {$3$} (x1);
							\draw[-Stealth] (p\x) edge[black]  node[pos=0.15, fill=white, inner sep=1pt] {$3$} (x3);
							\draw[-Stealth] (x2) edge[black]  node[pos=0.2, fill=white, inner sep=1pt] {$1$} (p\x);
						}
						
						\foreach \x / \b in {1/right, 3/left} {
							\draw[-Stealth] (x\x) edge[bend \b=0, black]  node[pos=0.2, fill=white, inner sep=1pt] {$1$} (x2);
						}
						
						\node[nodeU, right = \xs/2  of p1] (pi) {};
						\node[nodeU, left = \xs/2  of p2] (pj) {};					\draw[Stealth-Stealth] (pi) edge[bend right=25, gray] node[below, inner sep=3pt, align=center]{$3$ if adjacent, \\ $2$ otherwise} (pj);
					\end{scriptsize}
				\end{tikzpicture}
				\caption{Preferences defined in the proof for \cref{thm:ESA_NP_cycle_nonneg}.}
				\label{fig:ESA_NP_cycle_nonneg}
			\end{figure}
			The seat graph~$G$ consists of a single cycle with $|V(\hat{G})|+3$ vertices.			
			This completes the construction of an \ESA instance~$I$ with non-negative preferences, which can clearly be done in polynomial time. 
			
			Next, we show the correctness of our reduction. 
			Let~$\pathgraph$ be a Hamiltonian path of~$\hat{G}$.
			To obtain an \esArr of~$I$ we first assign~$x_1$ to an arbitrary seat in the seat graph. 
			Then, we continue with assigning~$x_2$ next to~$x_1$ and~$x_3$ next to~$x_2$.
			The remaining vertex-agents are assigned to the remaining seats according to the order in~$\pathgraph$.
			In this arrangement all agents (except~$x_2$) have their maximum possible utility and are therefore never part of an exchange-blocking pair.
			Hence, $x_2$ does not have a partner to build an exchange-blocking pair, which means this arrangement is exchange-stable.
			
			For the backward direction, let~$\sigma$ be an \esArr of~$I$. 
			We observe that~$\sigma$ satisfies the following:
			\begin{claim}\label{claim:ESA_NP_cycle}
				Agents $x_1, x_2$, and~$x_3$ are assigned consecutively. 
			\end{claim}
			\begin{proof}[Proof of \cref{claim:ESA_NP_cycle}]
				\renewcommand{\qedsymbol}{$\diamond$}
				Suppose towards a contradiction, that~$x_1$ is not assigned next to~$x_2$, i.e., there is an agent $p \notin \{x_1, x_2, x_3\}$ with $\sigma(p) \in \Neigh{G}(\sigma(x_2))$. 
				Then, it holds $\util[x_1]{\sigma} = 0$ and $\util[p]{\sigma} \leq 3$. 
				Since $\util[p]{\swap{x_1}{p}} \geq 4$, the pair $(p, x_1)$ is exchange-blocking. 
				
				With an analogous argument we can see that $x_3$ also has to be a neighbor of $x_2$. 
				Hence, agents $x_1, x_2$, and~$x_3$ have to be assigned consecutively in every \esArr. 
			\end{proof}
			
			From \cref{claim:ESA_NP_cycle} it follows that~$x_1$ and~$x_3$ are envy-free, but~$x_2$ envies every vertex-agent.
			Since~$\sigma$ is exchange-stable, we can conclude that $\util[p_i]{\sigma} = 6$, i.e., there is a Hamiltonian Path in $\hat{G}$. 
		\end{proof}
	
	From \cref{thm:ESA_delta_clique+path-cycle+matching} we can directly infer the following result: 
	
	\begin{corollary}\label{cor:ESA_NP_path}
		\ESA remains \NPh\ even for a single path as seat graph.
	\end{corollary}
	}
\fi 

	\section{Conclusion}\label{sec:conclude}
	We obtained a complete complexity picture for \MWA, \MUA, and \EFA,
	and left some open questions for \ESA\ (see \cref{tbl:overview-results}).
	Among these open questions, it would be interesting to know whether the \Wh{ness} result for stars-graphs can be extended to the case with path/cycle-graphs.
	Another research direction would be to look for arrangements that maximize welfare and are also envy-free or exchange stable. 
	In particular, is it \fpt\  wrt.\ $\nonisolated+\Delta^+$ to find an arrangement that maximizes welfare among the exchange-stable arrangements or maximize the minimum utility among the envy-free arrangements?

  \clearpage
  \section*{Acknowledgements}
  This work and Jiehua Chen have been funded by the Vienna Science and Technology Fund (WWTF) [10.47379/ VRG18012].
  
  \bibliographystyle{named}
  \bibliography{references}

\begin{thebibliography}{}

\bibitem[\protect\citeauthoryear{Agarwal \bgroup \em et al.\egroup
  }{2021}]{agarwal2021schelling}
Aishwarya Agarwal, Edith Elkind, Jiarui Gan, Ayumi Igarashi, Warut Suksompong,
  and Alexandros~A. Voudouris.
\newblock Schelling games on graphs.
\newblock {\em Artificial Intelligence}, 301:103576, 2021.

\bibitem[\protect\citeauthoryear{Alcalde}{1994}]{alcalde94}
Jos{\'e} Alcalde.
\newblock Exchange-proofness or divorce-proofness? {S}tability in one-sided
  matching markets.
\newblock {\em Economic Design}, 1(1):275--287, 1994.

\bibitem[\protect\citeauthoryear{Alon \bgroup \em et al.\egroup
  }{1995}]{AYZ95colorcoding}
Noga Alon, Raphael Yuster, and Uri Zwick.
\newblock Color-coding.
\newblock {\em Journal of {ACM}}, 42(4):844--856, 1995.

\bibitem[\protect\citeauthoryear{Aziz \bgroup \em et al.\egroup
  }{2013}]{aziz13}
Haris Aziz, Felix Brandt, and Hans~G. Seedig.
\newblock Computing desirable partitions in additively separable hedonic games.
\newblock {\em Artificial Intelligence}, 195:316--334, 2013.

\bibitem[\protect\citeauthoryear{Bachrach \bgroup \em et al.\egroup
  }{2013}]{bachrach2013optimal}
Yoram Bachrach, Pushmeet Kohli, Vladimir Kolmogorov, and Morteza Zadimoghaddam.
\newblock Optimal coalition structure generation in cooperative graph games.
\newblock In {\em Proceedings of the 27th {AAAI} Conference on Artificial
  Intelligence}, pages 81--87, 2013.

\bibitem[\protect\citeauthoryear{Bil{\`{o}} \bgroup \em et al.\egroup
  }{2022}]{bilo22}
Vittorio Bil{\`{o}}, Gianpiero Monaco, and Luca Moscardelli.
\newblock Hedonic games with fixed-size coalitions.
\newblock In {\em Proceedings of the 36th {AAAI} Conference on Artificial
  Intelligence}, pages 9287--9295, 2022.

\bibitem[\protect\citeauthoryear{Bodlaender \bgroup \em et al.\egroup
  }{2020a}]{bodlaender2020}
Hans~L. Bodlaender, Tesshu Hanaka, Lars Jaffke, Hirotaka Ono, Yota Otachi, and
  Tom~C. van~der Zanden.
\newblock Hedonic seat arrangement problems.
\newblock In {\em Proceedings of the 19th International Conference on
  Autonomous Agents and Multiagent Systems}, pages 1777--1779, 2020.

\bibitem[\protect\citeauthoryear{Bodlaender \bgroup \em et al.\egroup
  }{2020b}]{bodlaender2020tech}
Hans~L. Bodlaender, Tesshu Hanaka, Lars Jaffke, Hirotaka Ono, Yota Otachi, and
  Tom~C. van~der Zanden.
\newblock Hedonic seat arrangement problems.
\newblock Technical report, arXiv:2002.10898, 2020.

\bibitem[\protect\citeauthoryear{Bogomolnaia and Jackson}{2002}]{bogomolnaia02}
Anna Bogomolnaia and Matthew~O. Jackson.
\newblock The stability of hedonic coalition structures.
\newblock {\em Games and Economic Behavior}, 38(2):201--230, 2002.

\bibitem[\protect\citeauthoryear{Brandt \bgroup \em et al.\egroup
  }{2016}]{brandt2016}
Felix Brandt, Vincent Conitzer, Ulle Endriss, J{\'{e}}r{\^{o}}me Lang, and
  Ariel~D. Procaccia.
\newblock {\em Handbook of Computational Social Choice}.
\newblock Cambridge University Press, 2016.

\bibitem[\protect\citeauthoryear{Bredereck \bgroup \em et al.\egroup
  }{2020}]{BredHeeKnoNie2020multidimensional}
Robert Bredereck, Klaus Heeger, Dusan Knop, and Rolf Niedermeier.
\newblock Multidimensional stable roommates with master list.
\newblock In {\em Proceedings of the 16th International Conference on Web and
  Internet Economics}, volume 12495 of {\em LNCS}, pages 59--73, 2020.

\bibitem[\protect\citeauthoryear{Cai}{2008}]{Cai2008}
Leizhen Cai.
\newblock Parameterized complexity of cardinality constrained optimization
  problems.
\newblock {\em Computer Journal}, 51(1):102--121, 2008.

\bibitem[\protect\citeauthoryear{Caragiannis \bgroup \em et al.\egroup
  }{2012}]{caragiannis12}
Ioannis Caragiannis, Christos Kaklamanis, Panagiotis Kanellopoulos, and Maria
  Kyropoulou.
\newblock The efficiency of fair division.
\newblock {\em Theory of Computing Systems}, 50(4):589--610, 2012.

\bibitem[\protect\citeauthoryear{Cechl{\'a}rov{\'a} and
  Manlove}{2005}]{cechlarova05}
Katar{\'\i}na Cechl{\'a}rov{\'a} and David~F. Manlove.
\newblock The exchange-stable marriage problem.
\newblock {\em Discrete Applied Mathematics}, 152(1-3):109--122, 2005.

\bibitem[\protect\citeauthoryear{Chen and Roy}{2022}]{ChenRoy2022esa}
Jiehua Chen and Sanjukta Roy.
\newblock Multi-dimensional stable roommates in 2-{D}imensional {E}uclidean
  space.
\newblock In {\em 30th Annual European Symposium on Algorithms}, volume 244 of
  {\em Leibniz International Proceedings in Informatics}, pages 36:1--36:16,
  2022.

\bibitem[\protect\citeauthoryear{Chen \bgroup \em et al.\egroup
  }{2021}]{chen21}
Jiehua Chen, Adrian Chmurovic, Fabian Jogl, and Manuel Sorge.
\newblock On (coalitional) exchange-stable matching.
\newblock In {\em Algorithmic Game Theory: 14th International Symposium},
  volume 12885 of {\em Lecture Notes in Computer Science}, pages 205--220,
  2021.

\bibitem[\protect\citeauthoryear{Cseh \bgroup \em et al.\egroup
  }{2019a}]{CFH2019}
{\'A}gnes Cseh, Tam{\'a}s Fleiner, and Petra Harj{\'a}n.
\newblock Pareto optimal coalitions of fixed size.
\newblock {\em Journal of Mechanism and Institution Design}, 4(1):87--108,
  2019.

\bibitem[\protect\citeauthoryear{Cseh \bgroup \em et al.\egroup
  }{2019b}]{CsehIrvingManlove2019}
{\'{A}}gnes Cseh, Robert~W. Irving, and David~F. Manlove.
\newblock The stable roommates problem with short lists.
\newblock {\em Theory of Computing Systems}, 63(1):128--149, 2019.

\bibitem[\protect\citeauthoryear{Cygan \bgroup \em et al.\egroup
  }{2015}]{cygan15}
Marek Cygan, Fedor~V. Fomin, {\L}ukasz Kowalik, Daniel Lokshtanov, D{\'a}niel
  Marx, Marcin Pilipczuk, Micha{\l} Pilipczuk, and Saket Saurabh.
\newblock {\em Parameterized algorithms}.
\newblock Springer, 2015.

\bibitem[\protect\citeauthoryear{Downey and Fellows}{2013}]{DF13}
Rodney~G. Downey and Michael~R. Fellows.
\newblock {\em Fundamentals of Parameterized Complexity}.
\newblock Texts in Computer Science. Springer, 2013.

\bibitem[\protect\citeauthoryear{Dyer and Frieze}{1985}]{dyer85}
Martin~E. Dyer and Alan~M. Frieze.
\newblock On the complexity of partitioning graphs into connected subgraphs.
\newblock {\em Discrete Applied Mathematics}, 10(2):139--153, 1985.

\bibitem[\protect\citeauthoryear{Gale and Shapley}{1962}]{gale62}
David Gale and Lloyd~S. Shapley.
\newblock College admissions and the stability of marriage.
\newblock {\em The American Mathematical Monthly}, 69(1):9--15, 1962.

\bibitem[\protect\citeauthoryear{Garey and Johnson}{1979}]{garey1979}
Michael~R. Garey and David~S. Johnson.
\newblock {\em Computers and Intractability: {A} Guide to the Theory of
  NP-completeness}.
\newblock Mathematical Sciences Series. W. H. Freeman and Company, 1979.

\bibitem[\protect\citeauthoryear{Guti{\'{e}}rrez \bgroup \em et al.\egroup
  }{2016}]{GABMC2016MultiTeam}
Jimmy~H. Guti{\'{e}}rrez, C{\'{e}}sar~A. Astudillo, Pablo
  Ballesteros{-}P{\'{e}}rez, Daniel Mora{-}Meli{\`{a}}, and Alfredo
  Candia{-}V{\'{e}}jar.
\newblock The multiple team formation problem using sociometry.
\newblock {\em Computers \& Operations Research}, 75:150--162, 2016.

\bibitem[\protect\citeauthoryear{Kreisel \bgroup \em et al.\egroup
  }{2022}]{kreisel2022equilibria}
Luca Kreisel, Niclas Boehmer, Vincent Froese, and Rolf Niedermeier.
\newblock Equilibria in schelling games: Computational hardness and robustness.
\newblock In {\em Proceedings of the 21st International Conference on
  Autonomous Agents and Multiagent Systems}, pages 761--769, 2022.

\bibitem[\protect\citeauthoryear{Lewis and
  Carroll}{2016}]{LewisCarrollSeatPlan2016}
Rhyd Lewis and Fiona Carroll.
\newblock Creating seating plans: a practical application.
\newblock {\em Journal of the Operational Research Society}, 67(11):1353--1362,
  2016.

\bibitem[\protect\citeauthoryear{Massand and Simon}{2019}]{massand19}
Sagar Massand and Sunil Simon.
\newblock Graphical one-sided markets.
\newblock In {\em Proceedings of the 28th International Joint Conference on
  Artificial Intelligence}, pages 492--498, 2019.

\bibitem[\protect\citeauthoryear{Niedermeier}{2006}]{Nie06}
Rolf Niedermeier.
\newblock {\em Invitation to Fixed-Parameter Algorithms}.
\newblock Oxford University Press, 2006.

\bibitem[\protect\citeauthoryear{Shapley and Roth}{2012}]{shapley12}
Lloyd Shapley and Alvin Roth.
\newblock Stable matching: {T}heory, evidence, and practical design.
\newblock 2012.

\bibitem[\protect\citeauthoryear{Shehory and Kraus}{1998}]{shehory1998methods}
Onn Shehory and Sarit Kraus.
\newblock Methods for task allocation via agent coalition formation.
\newblock {\em Artificial intelligence}, 101(1-2):165--200, 1998.

\bibitem[\protect\citeauthoryear{Tomi{\'c} and
  Uro{\v{s}}evi{\'c}}{2021}]{tomic21}
Milan Tomi{\'c} and Dragan Uro{\v{s}}evi{\'c}.
\newblock A heuristic approach in solving the optimal seating chart problem.
\newblock In {\em International Conference on Mathematical Optimization Theory
  and Operations Research}, pages 271--283, 2021.

\bibitem[\protect\citeauthoryear{van Bevern \bgroup \em et al.\egroup
  }{2017}]{bevern17}
René van Bevern, Robert Bredereck, Laurent Bulteau, Jiehua Chen, Vincent
  Froese, Rolf Niedermeier, and Gerhard~J. Woeginger.
\newblock Partitioning perfect graphs into stars.
\newblock {\em Journal of Graph Theory}, 85(2):297--335, 2017.

\bibitem[\protect\citeauthoryear{Vangerven \bgroup \em et al.\egroup
  }{2022}]{vangerven22}
Bart Vangerven, Dirk Briskorn, Dries~R. Goossens, and Frits C.~R. Spieksma.
\newblock Parliament seating assignment problems.
\newblock {\em European Journal of Operational Research}, 296(3):914--926,
  2022.

\bibitem[\protect\citeauthoryear{Woeginger}{2013}]{Woeginger2013}
Gerhard~J. Woeginger.
\newblock Core stability in hedonic coalition formation.
\newblock In {\em Proceedings of the 39th Conference on Current Trends in
  Theory and Practice of Computer Science}, volume 7741 of {\em Lecture Notes
  in Computer Science}, pages 33--50, 2013.

\end{thebibliography}

  \iflong
  \clearpage
  
  \begin{table}[t!]
    \centering
    \Large \textbf{\appendixtitle}
  \end{table}
  \bigskip

  \appendix
  \appendixtext
  \fi
\end{document}